\pdfoutput=1
\documentclass[10pt,a4paper]{article}

\usepackage[top=1.0in,bottom=1.1in,left=1.1in,right=1.1in]{geometry}
\usepackage{amsmath,amssymb,amsthm,mathtools}
\usepackage{booktabs}
\usepackage{array}
\usepackage{tabularx}
\usepackage{longtable}
\usepackage{pdflscape}
\newcolumntype{P}[1]{>{\raggedright\arraybackslash}p{#1}}
\usepackage{xcolor}
\usepackage{tikz}
\usetikzlibrary{arrows.meta}
\usepackage[labelfont=bf]{caption}
\PassOptionsToPackage{hyphens}{url}
\usepackage{hyperref}
\usepackage{enumitem}
\usepackage[final,expansion=false]{microtype}
\newcommand{\claimcoderef}{the claim-to-code table of the public artifact (Section~\ref{sec:artifact})}
\newcommand{\Claimcoderef}{The claim-to-code table of the public artifact (Section~\ref{sec:artifact})}

\makeatletter
\g@addto@macro\UrlBreaks{\do\/\do\.\do\_\do\-\do\:\do\,}
\g@addto@macro\UrlBreaks{%
  \do\a\do\b\do\c\do\d\do\e\do\f\do\g\do\h\do\i\do\j\do\k\do\l\do\m%
  \do\n\do\o\do\p\do\q\do\r\do\s\do\t\do\u\do\v\do\w\do\x\do\y\do\z%
  \do\A\do\B\do\C\do\D\do\E\do\F\do\G\do\H\do\I\do\J\do\K\do\L\do\M%
\do\N\do\O\do\P\do\Q\do\R\do\S\do\T\do\U\do\V\do\W\do\X\do\Y\do\Z}
\makeatother
\newcommand{\lean}[1]{\path{#1}}
\newtheorem{theorem}{Theorem}[section]
\newtheorem{definition}[theorem]{Definition}
\newtheorem{proposition}[theorem]{Proposition}

\newtheorem{corollary}[theorem]{Corollary}

\newtheorem{remark}[theorem]{Remark}

\hypersetup{
  colorlinks=true,
  linkcolor=blue!70!black,
  citecolor=green!50!black,
  urlcolor=blue!60!black,
  pdftitle={Observer Determinacy of Termination Certificates: Sufficient Statistics, Blackwell Comparison, and Simple Projections for Step-Duplicating Recursion},
  pdfauthor={Moses Rahnama}
}

\newcommand{\baseax}{\{\mathrm{propext},\ \mathrm{Classical.choice},\ \mathrm{Quot.sound}\}}

\title{\textbf{Observer Determinacy of Termination Certificates:\\
Sufficient Statistics, Blackwell Comparison, and Simple Projections for Step-Duplicating Recursion}}

\author{Moses Rahnama\\Mina Analytics\\[0.2em]\small August 2026}
\date{}

\begin{document}

\maketitle
\vspace{-1.6em}

\begin{abstract}\small\vspace{-0.6em}
  The recursor $F(x,y,0)\to x$, $F(x,y,S(n))\to G(y,F(x,y,n))$ terminates and is confluent under all contexts, yet every orienting expression of the stated direct-measure grammar ignores the copied argument~$y$. A payload-sensitive orienter exists outside the grammar. We distinguish the information an observer retains from the targets a proof language can define. An observer $q:X\to Q$ \emph{licenses} a target $P:X\to V$ when $P$ is constant on its fibers. For a sound and complete language whose observer sees the input dimension, operational inexpressibility is equivalent to a pair of context-sharing worlds with equal observations and different target values. Each hypothesis is necessary. The counter observer licenses an orienting target that this grammar cannot define.

  On finite sets with a rational prior, a repair using $k\ge1$ side symbols and refusal has minimum refused weight $\varepsilon_\mu(k)=1-V_k(\mu)$, where $V_k$ is the best success probability of guessing the target within $k$ tries per observation. This curve is nonincreasing and convex. For set-valued certificate tasks with a certificate at every state, the least side alphabet is the maximum fiber chromatic number of the hypergraph of subsets with no common certificate. Observer refinement equals deterministic Blackwell comparison, and licensing equals statistical sufficiency under deterministic observation and a full-support prior.

  In every faithful recursor realization, natural-valued rankings of the root and extracted-call relations each have infinitely many full-order classes. Counter-determined extracted-call rankings have one. The occurrence-role channel separates the active and frame copies of the same payload and resolves one bit under the uniform binary law. Signature homomorphisms constant in the counter slot cannot recover the counter projection; no two-variable term over the guarded system's signature defines disequality through the void test. The declared cost model gives quadratic omitted mass against linear residual work. The general results and recursor instances are formalized in Lean~4.
\end{abstract}

\section{Introduction}\label{sec:intro}

The first-order rewrite rule
\[
  F(x,y,S(n)) \;\to\; G(y,F(x,y,n)),
\]
with the base case $F(x,y,0)\to x$, is the base/step/counter schema for primitive recursion on natural numbers~\cite{baader1998termrewriting,terese2003trs}; finite-type primitive recursion also appears in G\"odel's Dialectica interpretation~\cite{godel1958dialectica}. We call the two rules the \emph{recursor}. The counter decreases from $S(n)$ to $n$, the recursive call persists, and the step argument~$y$, the \emph{payload}, is copied into the wrapper $G(y,\cdot)$. For normal ground terms $a$ and $b$, every derivation from $F(a,b,S^k(0))$ has length $k+1$ and ends at $G^k(b,a)$. Figure~\ref{fig:recursor-trace} shows the run for $k=3$: each step removes one successor from the counter and writes one frame $G(b,\cdot)$ that keeps the payload, and the base rule then replaces the recursor by $a$. The frames are the record of the run, and a step rule that writes such a frame and continues the recursion holds its generator variable at least twice, once in the frame and once in the recursive call (Theorem~\ref{thm:record-emission-necessity}, in the syntax of Section~\ref{sec:arch}). Each state of the run offers two readings, the counter and the number of payload copies; the counter descends while the copies increase, and this paper asks which reading a proof language sees and what that reading determines. Proof assistants use this recursor for their own termination arguments: Lean's kernel accepts a recursive definition only after elaboration into uses of the recursors it derives for its inductive types~\cite{leanreference2026}. Two results in term rewriting concern rules that copy a variable: no standard first-order Knuth-Bendix order orients such a rule~\cite{baader1998termrewriting}, and the disjoint union of two terminating systems can fail to terminate when one of them contains such a rule~\cite{toyama1987counterexample}. The system terminates and is confluent under all contexts, and a computable normalizer decides joinability (Theorem~\ref{thm:free-context-normalization}). Yet every expression of the closed direct-measure grammar of Theorem~\ref{thm:canonical-instance} that orients the duplicating rule ignores the payload, while a measure outside the grammar reads the payload and orients the rule (Corollary~\ref{cor:universal-oi}). In the declared certificate hierarchy, a termination certificate first occurs at depth one, while a certificate satisfying the boundary conditions first occurs at depth two (Proposition~\ref{prop:boundary-kappa} and Theorem~\ref{thm:free-witness-depths}).

\begin{figure}[t]
  \centering
  \resizebox{\linewidth}{!}{%
  \begin{tikzpicture}[
      rec/.style={draw, minimum width=3.1cm, minimum height=0.5cm, inner sep=2pt, font=\footnotesize},
      act/.style={draw, thick, minimum width=3.1cm, minimum height=0.5cm, inner sep=2pt, font=\footnotesize},
      lab/.style={font=\footnotesize}]
    \foreach \i/\X/\c in {0/0/3, 1/4/2, 2/8/1, 3/12/0} {
      \ifnum\i>0
        \foreach \j in {1,...,\i} { \node[rec] at (\X, {0.58*\j}) {$G(b,\;\cdot\;)$}; }
      \fi
      \node[act] at (\X, {0.58*(\i+1)}) {$F(a,b,S^{\c}(0))$};
      \node[lab] at (\X, 0) {counter $\c$, copies $\i$};
    }
    \foreach \j in {1,2,3} { \node[rec] at (16, {0.58*\j}) {$G(b,\;\cdot\;)$}; }
    \node[act] at (16, {0.58*4}) {$a$};
    \node[lab] at (16, 0) {result $G^{3}(b,a)$, copies $3$};
    \foreach \A/\B in {0/4, 4/8, 8/12} {
      \draw[-{Stealth}] ({\A+1.6}, 0.87) -- ({\B-1.6}, 0.87) node[midway, above, font=\scriptsize] {step};
    }
    \draw[-{Stealth}] (13.6, 0.87) -- (14.4, 0.87) node[midway, above, font=\scriptsize] {base};
  \end{tikzpicture}}
  \caption{The run of the recursor from $F(a,b,S^3(0))$. Each step removes one successor from the counter and writes one frame $G(b,\cdot)$ that keeps the payload; the base rule replaces the recursor by $a$, and the result is $G^3(b,a)$. The counter reading descends $3,2,1,0$ while the number of payload copies ascends $0,1,2,3$.}
  \label{fig:recursor-trace}
\end{figure}

The recursor is one instance of a general question. A proof language reads a problem through a fixed observation, and this paper asks when a target value is determined by that observation and what the language must add when it is not. For an observer $q:X\to Q$ and a target $P:X\to V$, a decoder $d$ on the attained observations $q(X)$ with $P=d\circ q$ exists if and only if $P$ is constant on every fiber of $q$ (Proposition~\ref{prop:license-criterion}); we then say that $q$ \emph{licenses} $P$. In database theory this condition is determinacy of a query by a view~\cite{nashSegoufinVianu2010}, and the observers of a set, ordered by refinement, form the lattice of information of Landauer and Redmond~\cite{landauerRedmond1993}. The same condition is a functional dependency from $q$ to $P$~\cite{armstrong1974}, Padoa's method of definability on sets with a nonempty value set, whose first-order form is Beth's theorem~\cite{beth1953}, and sufficiency of $q$ for $P$ under deterministic observation and a full-support rational prior on a finite set~\cite{lehmannScheffe1950,bahadur1954} (Theorems~\ref{thm:fd-license} and~\ref{thm:sufficiency-definability}). Its failures are the confusable pairs of Witsenhausen's zero-error side-information problem~\cite{witsenhausen1976zeroerror} (Proposition~\ref{prop:side-channel-minimum}), so one criterion appears in five literatures: database theory, information-flow security, logic, statistics, and zero-error information theory.

A failure of the license is a \emph{collision}: a pair of states that $q$ identifies and $P$ separates. On the recursor, the counter observer licenses the denotation of every orienting expression of the closed direct-measure grammar, and the readings $(0,0)$ and $(0,1)$ of counter and payload collide for the payload target (Proposition~\ref{prop:grammar-license}). A proof language is \emph{operationally inexpressible} for an input at a dimension when the dimension is present and every derivable statement either ignores it or leaves the target unconstrained (Definition~\ref{def:op-incomplete}). Theorem~\ref{thm:definition-unification} proves that operational inexpressibility of a sound and complete language whose observer sees the dimension is an observer collision, and that each of the three hypotheses is necessary. Theorem~\ref{thm:beth-gap} splits the remaining licensed targets into those with an explicit defining statement and those for which the language fails the Beth property; the direct grammar fails the Beth property on orienting targets, while the strictly increasing functions of the counter, the language of the dependency-pair projection, have it (Proposition~\ref{prop:beth-grammar-projection}).

Two classes of termination methods appear in the later sections. Direct methods of the termination literature, such as polynomial interpretations and path orders~\cite{baader1998termrewriting}, add comparison data and verify the original relation; we call them \emph{construction methods} (Definition~\ref{def:construction}). Transformational methods, such as dependency-pair projection~\cite{artsgiesl2000}, the subterm criterion~\cite{hirokawa2004}, size-change summaries, and argument filtering, transform the problem, omit a source coordinate from the residual comparison, and transport the result through a soundness theorem; the short name used here is \emph{confession methods} (Definition~\ref{def:confession}).

\paragraph{Main results.} The paper proves five groups of results. First, the license criterion holds on arbitrary sets and carries a lattice theory: the target-kernel quotient is the least licensing observer, the licensed Boolean targets and the observer-preserving permutations each reconstruct the kernel, and rule families and preserved invariants form an antitone Galois connection. The unification theorem and the Beth gap of the preceding paragraphs belong to this group (Theorems~\ref{thm:definition-unification} and~\ref{thm:beth-gap}), and a single pair of worlds that share a context and an observation and differ in the dimension and in the target proves that a sound language is operationally inexpressible (Proposition~\ref{prop:fixed-input-boundary}).

Repair and comparison form the second group. An unlicensed target is repaired by a side channel that separates colliding states (a \emph{lift}), by refusing states (a \emph{refusal}), or by both (Theorem~\ref{thm:lift-or-refuse}). On finite sets the least refused weight is the Bayes risk of guessing the target; with $k$ side symbols it is attained by lists of $k$ target values, is nonincreasing and convex in $k$, and vanishes under full support if and only if $k$ reaches the largest number of target values on a fiber. On finite sets, set-valued recovery also needs the largest fiber chromatic number of side symbols (Theorem~\ref{thm:set-valued-capacity}), refinement of deterministic observers is the order of their Bayes risks over all targets and rational priors (Theorem~\ref{thm:blackwell-deterministic}), and finite experiments compare by garbling (Theorem~\ref{thm:blackwell-finite}). Through Theorem~\ref{thm:fd-license} and Proposition~\ref{prop:frontier-ktry}, the subset repairs of database theory are the maximal retained sets, the consistent answers are the states every repair keeps, the least fraction of states that a subset repair deletes is the Bayes risk under the uniform law, and the vulnerability of guessing within $k$ tries is one minus the least weight refused with $k$ side symbols.

The third group is the recursor realization. The direct grammar is operationally inexpressible at the step argument (Theorem~\ref{thm:canonical-instance}), and the projection it lacks separates two terms that every signature homomorphism constant in the counter slot identifies (Theorem~\ref{thm:dp-not-signature-derivable}). By the same homomorphic collapse criterion, no term over the signature defines the disequality guard that restores confluence in the second realization of Section~\ref{sec:sibling}, the confluence axis (Theorem~\ref{thm:homomorphic-inexpressibility} and Corollary~\ref{cor:guard-instance}). Record emission forces two generator occurrences in the step rule and at every subsequent live stage (Theorem~\ref{thm:record-emission-necessity} and Proposition~\ref{prop:record-emission-dynamic}). In the calibrated fork model of Theorem~\ref{thm:one-bit} the license resolves one role bit, and every separator of the two roles lies outside the observers that factor through the value (Proposition~\ref{prop:role-split-law}). Together these results identify two things the dependency-pair method adds on this rule besides its soundness theorem: a projection to the counter, the only one of the three trace coordinates that is invariant under permutations of the payload copies and strictly descends (Proposition~\ref{prop:S-retained}), and one role bit.

Ranks also admit a full-order classification: two ranks are identified when they induce the same strict order on all states. In every faithful recursor realization, natural-valued root rankings and extracted-call rankings each have infinitely many such classes, while counter-determined extracted-call rankings have one, whatever their construction method (Corollary~\ref{cor:faithful-rank-classification}). Theorem~\ref{thm:construction-confession-orders} gives the counter-system hypotheses for this classification.

Costs and stopping make up the fourth group. The payload mass that the projection omits is quadratic in the depth against linear residual work (Proposition~\ref{prop:confession-dominance}), and when the transformed proof is charged a license price $L$ once and one unit per recursive step, the depth at which it becomes cheaper is computable for a positive payload and at most $L$ (Proposition~\ref{prop:confession-crossing}). At every finite horizon the recursor-depth decision has an attained optimum whose values are the Snell envelope of the stopping rewards (Theorems~\ref{thm:recursor-stopping} and~\ref{thm:snell-envelope}), and an uninformative channel is refused at every positive cost (Proposition~\ref{prop:zero-deficit-stopping}). In the stated cost model, the choice between the direct representation and the transformed proof is therefore a computable comparison of a quadratic cost with a license price paid once plus linear work.

The fifth group is transport and calibration. Normal-form-faithful simulations transfer unique normal forms, stuttering simulations with a strict rank transfer termination, and together they transfer confluence (Theorems~\ref{thm:normal-form-faithful-transport} and~\ref{thm:confluence-transport}); each of sixteen hypotheses of the paper has a fixture in which it fails together with its conclusion (Theorem~\ref{thm:hypothesis-necessity}). The Arts and Giesl equivalence holds for finite systems with the variable condition, and the recursor's projected pair has height $\omega$, its predecessor descent derived over a finite base theory from one bounded induction sentence (Theorems~\ref{thm:ag-equivalence} and~\ref{thm:ag-rca}). By Theorem~\ref{thm:confluence-calibration}, the bounded termination statement is $\Pi^0_2$, the well-founded-induction sentence that Newman's lemma~\cite{newman1942} consumes is $\Pi^1_1$, the Dershowitz and Manna order on finite multisets of natural numbers has order type $\omega^\omega$, and the measure of the guarded relation on the confluence axis stays below $\omega^\omega\cdot2$ on every term. Table~\ref{tab:theorem-map} locates every group.

\begin{table}[htbp]
  \centering\footnotesize
  \setlength{\tabcolsep}{3pt}
  \caption{Theorem map. Each row names the mathematical question in declarative form, the results that answer it, and the scope of those results.}
  \label{tab:theorem-map}
  \begin{tabularx}{\textwidth}{@{}P{0.21\textwidth} X P{0.21\textwidth} P{0.07\textwidth}@{}}
    \toprule
    Subject & Results & Scope & Sections \\
    \midrule
    Reading a target through an observer & fiber constancy; partition lattice; least licensing observer; two kernel reconstructions; rule and invariant Galois connection; agreement & arbitrary sets and relations; finite weights for agreement & \ref{sec:observer}, \ref{sec:lattice} \\
    Inexpressibility in a language & observer collision under soundness, completeness, and dimension visibility, each necessary; fixed-input boundary; Beth gap & arbitrary worlds, statements, and observers & \ref{sec:observer} \\
    Repair of an unlicensed target & lift or refuse; Bayes risk; $k$-lists; universal lift; functional dependencies and subset repairs & finite presentations and rational weights for the costs & \ref{sec:capacity}, \ref{sec:noisy} \\
    Side information for recovery & fiber multiplicity; fiber chromatic number; joint targets; compactness; recurrence depth & arbitrary sets; finite carriers for the searches & \ref{sec:capacity}, \ref{sec:dynamics} \\
    Comparison of observers & refinement is the order of Bayes risks; garbling; entropy bound; sufficiency and definability & finite rational models; real finite experiments & \ref{sec:noisy} \\
    Licenses along a process & one loss; finite stabilization; retained histories; first sufficient horizon; access and hidden statistics & monotone or compatible observer sequences & \ref{sec:dynamics} \\
    The direct grammar on the recursor & normalization; failure of additive measures; operational inexpressibility at the step argument; transfer to subgrammars; license instances & the closed direct-measure grammar & \ref{sec:recursor}, \ref{sec:oi} \\
    The projection and the signature & homomorphic collapse criterion; factorization form on the fold; necessity of the constancy hypothesis; guard instance & finitary signatures; the signature of the Concrete Witness System & \ref{sec:oi}, \ref{sec:sibling} \\
    Generator duplication & two occurrences forced in the step rule and at every subsequent live stage; three further regimes; $r$-ary scaling & many-sorted first-order syntax & \ref{sec:arch} \\
    Content of the license & one role bit; role-split law; role exchange at finite and countable support; scheduled and evidence forms of the gap & finite forks; countable role alphabets & \ref{sec:arch}, \ref{sec:license} \\
    Cost of the direct representation & quadratic omitted mass; crossing depth; proof entropy; inefficiency rate; certificate parsers & the recursor trace and the stated cost models & \ref{sec:price} \\
    Stopping & Snell envelope; zero-deficit stopping; well-founded licensed deliberation & finite-horizon and finite decision models & \ref{sec:price} \\
    Transport along simulations & unique normal forms; termination with a stutter rank; confluence; quotient simulation; fixtures for sixteen hypotheses & arbitrary relations & \ref{sec:transport} \\
    The confluence axis & terminal support characterizes confluence; repair covers; finite algorithms; $\Pi^1_1$ Newman sentence; order on finite multisets of natural numbers of type $\omega^\omega$; guarded measure below $\omega^\omega\cdot2$ & finite relations; the guarded relation of the Concrete Witness System & \ref{sec:sibling} \\
    Rankings and certificate depths & one counter-determined order class; infinitely many full-order classes; minimum certificate depths one and two & faithful recursor realizations; declared certificate languages & \ref{sec:witness} \\
    Strength of the license & construction and confession information; Arts and Giesl equivalence; $\Pi^0_2$ presentation; height $\omega$ from one bounded induction sentence; six-slot comparison & finite systems with the variable condition; finite arithmetic theories & \ref{sec:witness}, \ref{sec:weight} \\
    \bottomrule
  \end{tabularx}
\end{table}

\paragraph{Organization.} Sections~\ref{sec:observer} to~\ref{sec:dynamics} develop the general theory, including the Beth property, the two Galois connections, agreement, garbling of finite experiments, functional dependencies, and recurrence depth. Sections~\ref{sec:recursor} to~\ref{sec:license} present the free recursor as a realization, with the signature obstruction, record emission, and the role license, and Section~\ref{sec:price} treats costs, stopping, and certificates. Section~\ref{sec:transport} proves the transport theorems and gives, for sixteen hypotheses of the paper, fixtures in which each fails with its conclusion; Section~\ref{sec:sibling} treats the confluence axis and its calibration. Sections~\ref{sec:witness} and~\ref{sec:weight} classify the certificate layers, construction versus confession information, and the reverse mathematics of the projected pair. Section~\ref{sec:formalization} describes the formalization, Section~\ref{sec:related} relates the results to prior work, and Section~\ref{sec:artifact} names the public artifact and its claim-to-code table.

\paragraph{Reading guide.} Readers from term rewriting can start at Section~\ref{sec:recursor}, read Sections~\ref{sec:recursor} to~\ref{sec:license}, and continue with Sections~\ref{sec:transport} to~\ref{sec:weight}. Readers from database theory and information theory can read Sections~\ref{sec:observer} to~\ref{sec:dynamics}, where Section~\ref{sec:noisy} holds the Bayes risk, the repairs, and the functional dependencies. Readers from logic can read Sections~\ref{sec:observer}, \ref{sec:noisy}, \ref{sec:sibling}, and~\ref{sec:weight}, for the Beth property, Padoa's method, the $\Pi^1_1$ calibration, and the reverse mathematics of the projected pair.

\paragraph{Relation to the orientation-boundary and unique-normal-forms papers.} The orientation-boundary paper of Rahnama~\cite{rahnamaOrientation} proves barriers for twelve direct-measure families and proves termination of this rule by dependency pairs with the subterm criterion, whose simple projection~\cite{hirokawa2004} selects the counter argument, by polynomial interpretations, and by path orders. The unique-normal-forms preprint of Rahnama~\cite{rahnamaDistinction} treats the confluence axis of the same concrete system.

\paragraph{Scope.} The theorems of Sections~\ref{sec:observer} to~\ref{sec:dynamics} and the transport theorems of Section~\ref{sec:transport} quantify over arbitrary sets or arbitrary relations at the scopes stated in their hypotheses; the two exceptions, Propositions~\ref{prop:beth-grammar-projection} and~\ref{prop:grammar-license}, are stated on the carrier $\mathbb N\times\mathbb N$ of counter and payload readings of the direct grammar. Finite capacity, entropy, Bayes optimization, experiment comparison, agreement, and stopping also assume the finite presentations or real weightings named in those theorems. The reverse-mathematics theorems quantify over the stated finite arithmetic theories and the recursor's projected pair. The direct-grammar, certificate-depth, cost, extraction, certificate, and occurrence-role statements concern the free base/step/counter system and its faithful realizations. Every general theorem quantifies over sets, relations, or arithmetic theories; apart from those two propositions and Remark~\ref{rem:orientation-consumer}, the results about the rewrite system lie in the realization sections.

\section{Observers, targets, and the license criterion}\label{sec:observer}

An observer is a map $q:X\to Q$ from states to observations, and a target is a map $P:X\to V$ from states to values. The sets $X$, $Q$, and $V$ are arbitrary and may be infinite; a statement of this section uses finiteness or decidability only when it says so. Reading $P$ through $q$ asks for a decoder $d$ with $P=d\circ q$ on the states.

\begin{proposition}[License criterion]\label{prop:license-criterion}
  Let $q:X\to Q$ be an observer and $P:X\to V$ a target. Reading $P$ through $q$ is licensed if $P$ factors through the quotient $X/{\ker q}$, a condition equivalent to factorization through the attained image $q(X)$. This holds if and only if $P$ is constant on every fiber of $q$. A decoder needs values only on $q(X)$. One observer refines another, in the sense that it separates every pair the other separates, if and only if it licenses every target the other licenses, targets taken over the other observer's universe; two observers refine each other if and only if they have the same kernel. An injective observer licenses every target.
\end{proposition}

\begin{proof}
  If $P=d\circ q$, equal observations give equal target values. In the other direction, if $P$ is constant on the fibers of $q$, put $d(q(x))=P(x)$ on the attained image $q(X)$; constancy makes $d$ well defined, and values outside $q(X)$ are never evaluated. An observer $q_1$ refines $q_2$ when $\ker q_1\subseteq\ker q_2$. A target licensed by $q_2$ is then licensed by $q_1$; in the other direction, the target $q_2$ is licensed by $q_2$, so an observer $q_1$ that licenses it satisfies $\ker q_1\subseteq\ker q_2$. Mutual refinement is equality of kernels, and an injective observer has the diagonal kernel.
\end{proof}

\begin{proposition}[Executable licensing]\label{prop:executable-license}
  Let $X$ have an explicit duplicate-free list of its elements, and let $Q$ and $V$ have decidable equality. The Boolean check that compares $P$ on every listed pair with equal observations returns true if and only if $q$ licenses $P$, and the check that compares $q_2$ on every listed pair identified by $q_1$ returns true if and only if $q_1$ refines $q_2$. The collision search returns only pairs $x,y$ with $q(x)=q(y)$ and $P(x)\neq P(y)$, and it returns none if and only if $q$ licenses $P$. Computed from the list, the target classes are the fibers of $P$; they cover $X$, and their number is the number of attained target values. On the empty and on the one-point carrier the check returns true for every observer and target. On $\mathrm{Bool}\times\mathrm{Bool}$ the first coordinate fails to license the parity of the two coordinates, the search returns the collision of $(\mathsf{false},\mathsf{false})$ and $(\mathsf{false},\mathsf{true})$, and the identity observer licenses parity.
\end{proposition}
\begin{proof}
  The list contains every element, so a statement over all listed pairs is a statement over all pairs of $X$, and each check decides a finite conjunction. The search returns the first listed pair that violates the license, which gives soundness and the characterization of the empty answer. Filtering the list by each attained target value gives the classes. The three fixtures are evaluated.
\end{proof}

\subsection{Operational inexpressibility of a language}\label{subsec:op-language}

The license criterion concerns an observer; operational inexpressibility concerns a proof language.

\begin{definition}[Operational inexpressibility]\label{def:op-incomplete}
  Fix a proof language for a target question together with predicates
  \[
    \mathrm{Derivable}(\psi),\qquad \mathrm{Depends}_\pi(\psi),\qquad \mathrm{Constrains}_R(\psi)
  \]
  for a fixed input $a$. The language is \emph{operationally inexpressible for
  input $a$ at dimension $\pi$} relative to target question $R$ if both:
  \begin{enumerate}[label=(\roman*),leftmargin=1.6em]
    \item \textbf{Presence of the dimension.} The value $\pi(a)$ is non-degenerate and relevant to the structure of $a$.
    \item \textbf{Absence of incorporating derivations.} For every derivable statement $\psi$, at least one of the following holds:
      \begin{itemize}[leftmargin=1.4em]
        \item $\psi$ is independent of $\pi(a)$;
        \item the truth of $\psi$ leaves $R(a)$ unconstrained.
      \end{itemize}
  \end{enumerate}
\end{definition}

To compare Definition~\ref{def:op-incomplete} with the license criterion, fix a set $W$ of worlds and a relation $\sim$ on $W$ whose related worlds share the input context. Fix also statements with a derivability predicate and denotations $\mathrm{den}_\psi:W\to V$, an observer $o:W\to O$, a dimension map $\pi:W\to D$, and a target $R:W\to V$. The dimension is \emph{present} when two worlds of one context differ in $\pi$. A statement \emph{depends on the dimension} when its denotation changes across such a pair, and it \emph{constrains the target} when, inside one context, equal denotations force equal target values. The language is operationally inexpressible when the dimension is present and no derivable statement both depends on it and constrains the target, which is Definition~\ref{def:op-incomplete} read on worlds; the observer is operationally inexpressible when the dimension is present and two worlds of one context share the observation and differ in the target. A language is \emph{sound} for its observer when every derivable denotation is constant on observer fibers, and \emph{complete} when one derivable statement separates every pair of worlds with different observations; the observer \emph{sees the dimension} when some pair of worlds in one context differs both in $\pi$ and in $o$.

\begin{theorem}[Language-level and observer-level inexpressibility]\label{thm:definition-unification}
  For a language that is sound and complete for its observer and whose observer sees the dimension, the language is operationally inexpressible if and only if the observer is. If every two worlds share a context, both conditions are equivalent to presence of the dimension together with failure of the license of $R$ by $o$. Each of the three hypotheses is necessary: dropping any one of them while keeping the other two admits a language in which one of the two conditions holds and the other fails. When dependence on the dimension is failure of the license by an observer that forgets the dimension, clause~(ii) of Definition~\ref{def:op-incomplete} holds if and only if that observer licenses every derivable statement that constrains the target.
\end{theorem}
\begin{proof}
  A sound language gives equal denotations to worlds with equal observations, so a collision of the observer with the target defeats every derivable statement that would constrain the target; this proves the observer-to-language direction. In the other direction, completeness supplies a derivable statement that separates worlds with different observations. Without a collision this statement constrains the target, and dimension visibility makes it depend on the dimension, contrary to the condition on the language. With a single context, a collision is an unlicensed pair. The fixtures are a sound and complete language whose observer does not see the dimension, and a constant-statement language, sound with the dimension visible and incomplete, in both of which the condition on the language holds and the condition on the observer fails; and an unsound language, complete with the dimension visible, in which the condition on the observer holds and the condition on the language fails. The license form rewrites clause~(ii).
\end{proof}

\begin{proposition}[Fixed-input boundary]\label{prop:fixed-input-boundary}
  Fix an input $a$ and a language whose derivable statements have denotations that factor through an observer $o$ on worlds. Suppose two worlds $w_1,w_2$ of input $a$ share a context and an observation, differ in the dimension $\pi$, and differ in the target $R$. Then the dimension is present; every derivable statement has a denotation constant on the fibers of $o$; no derivable statement constrains $R$ inside the shared context, and none decides $R$ on every world; and $R$ factors neither through $o$ nor through the kernel quotient of $o$. The language is therefore operationally inexpressible for $a$ at $\pi$ in the sense of Definition~\ref{def:op-incomplete}, with dependence and constraint read as the semantic predicates of Theorem~\ref{thm:definition-unification}.
\end{proposition}
\begin{proof}
  The pair $w_1,w_2$ shows presence of the dimension and is an observer collision for $R$. A derivable statement factors through $o$, so its denotation agrees on $w_1$ and $w_2$; if it constrained $R$ inside the shared context, the two target values would agree, contrary to the choice of the pair. A statement that decides $R$ pointwise constrains it. The collision refutes factorization of $R$ through $o$, and factorization through the kernel quotient is equivalent to factorization through $o$ by Proposition~\ref{prop:target-kernel-quotient}. Clause~(ii) of Definition~\ref{def:op-incomplete} follows.
\end{proof}

\begin{definition}[Explicit definition and the Beth property]\label{def:beth-property}
  In an observed language, a derivable statement $\psi$ explicitly defines a target $P$ when its
  denotation equals $P$. The language has the Beth property on a class of targets when every target of
  the class that the observer licenses is explicitly defined by a derivable statement.
\end{definition}

\begin{theorem}[Beth gap]\label{thm:beth-gap}
  Let the language be sound.
  \begin{enumerate}[label=(\alph*),leftmargin=1.6em]
    \item An explicitly defined target is licensed by the observer.
    \item A target without explicit definition is unlicensed, or licensed without explicit definition;
      on a class where the language has the Beth property, only the first case occurs.
    \item A complete language closed under composition with maps of answer values, with a nonempty answer
      type, has the Beth property on all targets. Soundness, completeness, and closure are each necessary.
  \end{enumerate}
\end{theorem}
\begin{proof}
  (a) Soundness applies to the defining statement. (b) The two cases exhaust the possibilities by (a). (c) Completeness makes each fiber of the statement's denotation lie within an observer fiber. A licensed target therefore factors through that denotation; extend the factor using the nonempty answer type and apply post-composition. The fixtures separate the hypotheses.
\end{proof}

\begin{proposition}[Beth property of two proof languages]\label{prop:beth-grammar-projection}
  On $\mathbb N\times\mathbb N$, for targets that orient the duplicating step:
  \begin{enumerate}[label=(\alph*),leftmargin=1.6em]
    \item with the full reading as observer, the direct grammar fails the Beth property: the orienter
      $\mu_*$ of Corollary~\ref{cor:universal-oi} is licensed and no expression denotes it;
    \item with the counter as observer, every orienting grammar expression is licensed, and the grammar
      still fails the Beth property: the rank that is $0$ at counter $0$ and $c+1$ at counter $c\ge1$ orients,
      is licensed, and no expression denotes it, since every expression $e$ with $e(0,0)=0$ satisfies
      $2e(1,0)\le e(2,0)$;
    \item the language of strictly increasing functions of the counter is sound and has the Beth
      property: an orienting target is explicitly defined if and only if the counter licenses it.
  \end{enumerate}
\end{proposition}
\begin{proof}
  (a) The payload-sensitive orienter of Proposition~\ref{prop:grammar-license} is licensed by the full reading and is not a grammar denotation. (b) Orienting expressions are counter-licensed; the rank jumps by one, and doubling the value at counter one bounds the value at counter two for every origin-vanishing expression. (c) A strictly increasing function of the counter is the explicit definition, and it exists if and only if the target is licensed.
\end{proof}

Part~(c) is the dependency-pair projection read as a language: it defines every orienting target that the counter determines, which is the separation of determinacy from rewriting of Nash, Segoufin, and Vianu~\cite{nashSegoufinVianu2010} decided for these two languages; Beth's theorem~\cite{beth1953} is the first-order case in which the Beth property always holds.

\begin{proposition}[The grammar boundary as a license instance]\label{prop:grammar-license}
  On the carrier $\mathbb N\times\mathbb N$ of (counter, payload) readings, a measure $m$ is payload blind, that is $m(c,p)=m(c,p')$ for all $c,p,p'$, if and only if the counter observer $(c,p)\mapsto c$ licenses $(c,p)\mapsto m(c,p)$. Hence an expression of the reflected grammar of Theorem~\ref{thm:canonical-instance} reads the payload if and only if the counter observer fails to license its denotation, and every grammar expression that orients the duplicating carrier, the edges $(c+1,p)\to(c,p+L)$ with $L\ge1$, has a denotation licensed by the counter observer; the same holds for every syntax whose denotations are realized in the grammar. The payload target $(c,p)\mapsto p$ is unlicensed by the counter observer, shown by $(0,0)$ and $(0,1)$. The counter projection orients the step and is licensed; the orienter $\mu_*$ of Corollary~\ref{cor:universal-oi} orients the step and is unlicensed.
\end{proposition}
\begin{proof}
  Payload blindness is constancy on counter fibers, which is the license. Theorem~\ref{thm:canonical-instance} shows that an orienting grammar expression is payload blind, and a denotational subgrammar inherits the statement through its realization map. The pair $(0,0),(0,1)$ has equal counter and different payload. The last two statements combine the counter case of Theorem~\ref{thm:canonical-instance} with Corollary~\ref{cor:universal-oi}.
\end{proof}

\section{The license lattice and the target-kernel quotient}\label{sec:lattice}

Licensing depends on an observer only through its kernel, the equivalence relation of equal observations. This section identifies observers with their kernels, describes the resulting lattice and the least observer that licenses a given target, and gives two reconstructions of the kernel.

\begin{proposition}[The target-kernel quotient is the least licensed observer]\label{prop:target-kernel-quotient}
  For a target $P:X\to V$ let $X/P$ be the quotient of $X$ by equality of target values, with projection $\pi_P$. Then $\pi_P$ licenses $P$; an observer $q$ licenses $P$ if and only if $q$ refines $\pi_P$; and an observer that licenses $P$ and is refined by $\pi_P$ has the kernel of $P$. So $\pi_P$ is, up to kernel, the one observer that licenses $P$ and is refined by every observer that licenses $P$. This holds for arbitrary types, free of finiteness, decidability, and surjectivity assumptions.
\end{proposition}

\begin{proof}
  The target $P$ descends to $X/P$ by construction, so $\pi_P$ licenses $P$. An observer $q$ licenses $P$ when equal $q$-values force equal $P$-values, which says that $q$ refines $\pi_P$. If $q$ licenses $P$ and $\pi_P$ refines $q$, the two kernel inclusions give the kernel of $P$. No step uses finiteness, decidability, or surjectivity.
\end{proof}

\begin{remark}[Use in the orientation analysis]\label{rem:orientation-consumer}
  Proposition~\ref{prop:target-kernel-quotient} applies to the counter $c$ of a counter system in the sense of Definition~\ref{def:counter-system}, the setting of the orientation analysis~\cite{rahnamaOrientation}: the quotient by equality of counter values is sufficient for $c$, an observer is sufficient for $c$ if and only if it refines that quotient, and the quotient is isomorphic to $\mathbb N$. On the free recursive-call relation, two observers license a natural-valued rank and neither refines the other, and no observer that licenses a rank is refined by every natural-valued observer that licenses a rank.
\end{remark}

\begin{theorem}[The license lattice]\label{thm:license-lattice}
  Let $X$ be any set, and identify two observers on $X$ that refine each other. Refinement orders the classes, and sending an observer to its kernel is an order isomorphism from the classes onto the equivalence relations on $X$ ordered by inclusion, hence onto the partitions of $X$. The joint observer $x\mapsto(q_1(x),q_2(x))$ has kernel $\ker q_1\cap\ker q_2$, and the joint observer of a family has the intersection of the kernels. The join $q_1\vee q_2$ in this order, the quotient by the equivalence relation generated by $\ker q_1\cup\ker q_2$, is refined by both observers, and it refines every observer that both of them refine. A target is licensed by $q_1\vee q_2$ if and only if it is licensed by $q_1$ and by $q_2$, and it is licensed by the join of a family if and only if every member licenses it. Two states are identified by $q_1\vee q_2$ if and only if a finite chain of states connects them in which consecutive states agree in $q_1$ or in $q_2$. A target licensed by $q_1$ or by $q_2$ is licensed by the joint observer, and on $\mathrm{Bool}\times\mathrm{Bool}$ the parity of the two coordinates is licensed by their joint observer and by neither coordinate alone.
\end{theorem}
\begin{proof}
  Refinement is inclusion of kernels, and every equivalence relation is the kernel of its quotient map; together these give the order isomorphism, and the partition form is the correspondence between equivalence relations and partitions. Equal joint observations are equal pairs of observations. The generated equivalence relation is the least one containing both kernels, which gives the refinement statements for $q_1\vee q_2$. An observer licenses a target when its kernel lies in the kernel of the target, and a least upper bound of kernels lies in a kernel if and only if each of them does. The chain statement unfolds the generated equivalence relation. For parity, the joint observer is injective, while each coordinate identifies two states of different parity.
\end{proof}

In the reverse order, which Aumann uses, $q_1\vee q_2$ is the meet of information partitions with which he describes common knowledge~\cite{aumann1976}; it is also the meet in the lattice of information of Landauer and Redmond~\cite{landauerRedmond1993}. Theorem~\ref{thm:license-lattice} records its license theory.

\begin{theorem}[Agreement]\label{thm:aumann-agreement}
  Let $X$ be finite with a real weighting $\mu$ of either sign, let $q_1$ and $q_2$ be observers, let $E\subseteq X$, and
  let $M(\omega)$ be the fiber of $q_1\vee q_2$ through $\omega$. Say that $q_i$ assigns the signed fiber ratio $a$ to
  $E$ at $y$ when $\mu(E\cap q_i^{-1}q_i(y))=a\,\mu(q_i^{-1}q_i(y))$. An event is common knowledge at
  $\omega$ when it contains every state reached from $\omega$ by a finite chain whose consecutive states
  share a fiber of $q_1$ or of $q_2$; this holds if and only if the event contains $M(\omega)$. If
  $\mu(M(\omega))\neq0$ and it is common knowledge at $\omega$ that $q_1$ assigns $a$ and $q_2$ assigns $b$
  to $E$, then $a=b$.
\end{theorem}
\begin{proof}
  $M(\omega)$ is a union of fibers of $q_1$ and a union of fibers of $q_2$, so summing the assignments over those fibers gives $\mu(E\cap M(\omega))=a\,\mu(M(\omega))=b\,\mu(M(\omega))$.
\end{proof}

When $\mu$ is nonnegative with total weight one, the signed fiber ratio on a fiber of positive weight is the conditional probability of $E$ given that fiber, and the theorem specializes to the agreement theorem of Aumann~\cite{aumann1976} for finite state spaces.

\begin{theorem}[Reconstruction from licensed Boolean targets]
  \label{thm:license-galois}
  For an equivalence relation $E$ on any set $X$, write $B(E)$ for the
  Boolean targets constant on its classes. For a family $\mathcal S$ of
  Boolean targets, write $K(\mathcal S)$ for the intersection of their
  kernels. Then
  \[
    E\subseteq K(\mathcal S)\quad\Longleftrightarrow\quad
    \mathcal S\subseteq B(E),\qquad K(B(E))=E.
  \]
  Thus $B$ and $K$ form an antitone Galois connection. The map $B\circ K$
  is an extensive, monotone, idempotent closure on target families.
  Its fixed points are the families $B(E)$, with every such family fixed.
  Kernel inclusion is equivalent to reverse inclusion of licensed Boolean
  families, and proper inclusion is equivalent to proper reverse inclusion.
\end{theorem}
\begin{proof}
  Both sides of the first equivalence state that every member of
  $\mathcal S$ is constant on every $E$-class. For the second equality,
  the inclusion $E\subseteq K(B(E))$ follows. If $x,y$ are
  not $E$-equivalent, the Boolean indicator of the class of $x$ belongs
  to $B(E)$ and separates them, giving the reverse inclusion.
  The closure laws and the inclusion classification follow from these
  two identities. In particular, $B$ is injective, so a proper kernel
  inclusion cannot give equal target families.
\end{proof}

\begin{theorem}[Rules and invariants]\label{thm:rule-invariant-galois}
  Let $\mathrm{step}$ assign to each rule $r$ of a set $\mathcal R$ a relation $\mathrm{step}_r$ on a set $X$. A rule $r$ preserves a predicate $P\subseteq X$ when $P(y)$ holds whenever $\mathrm{step}_r(x,y)$ and $P(x)$ hold. For a family $S\subseteq\mathcal R$ let $\mathrm{Inv}(S)$ be the set of predicates preserved by every rule of $S$, and for a family $F$ of predicates let $\mathrm{Rules}(F)$ be the set of rules preserving every predicate of $F$. Then
  \[
    S\subseteq\mathrm{Rules}(F)\quad\Longleftrightarrow\quad F\subseteq\mathrm{Inv}(S),
  \]
  so $\mathrm{Inv}$ and $\mathrm{Rules}$ form an antitone Galois connection, and the composites $\mathrm{Rules}\circ\mathrm{Inv}$ and $\mathrm{Inv}\circ\mathrm{Rules}$ are idempotent closure operators. For a requested family $S$ and a family $F$ of invariants, the repair $S\cap\mathrm{Rules}(F)$ is contained in $S$, preserves $F$, contains every family $T\subseteq S$ that preserves $F$, and its complement in $S$ lies inside $S\setminus T$ for every family $T$ that preserves $F$; the repair and its complement partition $S$.
\end{theorem}
\begin{proof}
  Both sides of the displayed equivalence say that every rule of $S$ preserves every predicate of $F$; the closure laws follow. The repair keeps the rules of $S$ that preserve $F$, which gives containment, preservation, and maximality among the families inside $S$. A rule excluded by the repair fails some predicate of $F$ and therefore lies outside every family that preserves $F$.
\end{proof}

Theorem~\ref{thm:license-galois} relates kernels to licensed targets; Theorem~\ref{thm:rule-invariant-galois} relates rule families to preserved invariants, and deleting from a requested family the rules that break an invariant is the least exclusion that restores it.

\begin{theorem}[Observer-preserving permutations]
  \label{thm:license-stabilizer}
  For any observer $q:X\to Q$, define
  \[
    G_q=\{\sigma\in\operatorname{Sym}(X):q\circ\sigma=q\}.
  \]
  For every target $P:X\to V$, licensing through $q$ is equivalent to
  $P\circ\sigma=P$ for every $\sigma\in G_q$. Two states have the same
  observation if and only if some $\sigma\in G_q$ maps one to the other:
  \[
    q(x)=q(y)\quad\Longleftrightarrow\quad
    \exists\sigma\in G_q,\ \sigma(x)=y.
  \]
  Thus $q_1$ refines $q_2$ if and only if $G_{q_1}\subseteq G_{q_2}$;
  proper refinement is proper subgroup inclusion. Equal observer kernels
  are equivalent to equal stabilizers, and $G_q$ is trivial if and only
  if $q$ is injective. There is a group isomorphism
  \[
    G_q\cong\prod_{a\in q(X)}\operatorname{Sym}(q^{-1}(a)).
  \]
  If $X$ is finite, then
  \[
    |G_q|=\prod_{a\in q(X)}|q^{-1}(a)|!.
  \]
  The preceding equivalences and the group isomorphism hold for arbitrary
  sets $X$ and $Q$.
\end{theorem}
\begin{proof}
  Every observer-preserving permutation acts within the fibers of $q$.
  If $q(x)=q(y)$, the transposition of $x$ and $y$ preserves $q$.
  Hence a target is constant on fibers if and only if it is invariant
  under $G_q$, and the orbits of $G_q$ are the fibers. These identities
  give the refinement, kernel and injectivity classifications.
  Restricting a permutation to each fiber defines the displayed
  homomorphism. A family of fiber permutations defines its inverse by
  acting on each point in its own fiber. For finite $X$, counting the
  permutations in each factor gives the product formula.
\end{proof}

These groups preserve observation values and can move other primitive
data: on a two-point set with a named constant, the transposition
preserves a constant observer but moves the named point. The group of
permutations that also fix the primitive data is then a proper
subgroup of the observer stabilizer.

Theorem~\ref{thm:sufficiency-definability}(b) states the definability form of Theorem~\ref{thm:license-stabilizer}.

\section{Set-valued recovery and finite capacity}\label{sec:capacity}

A certificate task may have several correct answers at one state. This section treats an acceptance relation $A\subseteq X\times W$ and asks when a decoder on the observer quotient, with an optional side channel, returns an accepted certificate at every state. Function-valued targets are the case $A(x,w)\equiv(w=P(x))$.

\begin{proposition}[Minimum side channel for a hidden output]\label{prop:side-channel-minimum}
  Let $q:X\to Q$ and $P:X\to V$ be arbitrary maps, and put
  \[
    T_o=\{P(x):x\in X,\ q(x)=o\}\qquad(o\in Q).
  \]
  For any alphabet $C$, a channel $s:X\to C$ through which $(q,s)$ licenses
  $P$ exists if and only if every $T_o$ admits an injection into $C$.
  In particular, an $n$-symbol channel exists if and only if every $T_o$ is
  finite and $|T_o|\leq n$; a $k$-bit channel exists if and only if every
  $T_o$ is finite and $|T_o|\leq2^k$.
  These statements hold for arbitrary sets $X$, $Q$, and $V$.
  An infinite $T_o$, or unbounded finite sizes of the $T_o$, excludes every
  finite alphabet. A one-symbol channel suffices if and only if $q$ already
  licenses $P$; a zero-symbol channel exists if and only if $X$ is empty.

  A channel resolves the target if and only if it separates every pair
  that $q$ identifies and $P$ separates. These pairs are the edges of the
  confusability graph, so resolving channels are its proper colorings.
  For finite nonempty $X$, the minimum alphabet size is
  $\max_o|T_o|$, and the minimum binary width is
  $\lceil\log_2\max_o|T_o|\rceil$; both bounds are attained.
\end{proposition}
\begin{proof}
  Given a resolving channel, choose one state for each target in an
  observation fiber. The channel values of these representatives are
  distinct, giving an injection $T_o\hookrightarrow C$. In the other direction, choose
  injections $e_o:T_o\hookrightarrow C$ and define
  $s(x)=e_{q(x)}(P(x))$. Equal augmented observations have equal $q$ values
  and equal images under the same injection, hence equal targets.
  Taking alphabets of sizes $n$ and $2^k$ gives the finite and binary
  criteria. The empty-alphabet case tests whether a function $X\to\varnothing$
  exists. The graph statement is the same separation condition on each edge.
\end{proof}

\begin{theorem}[Relational recovery on observer fibers]\label{thm:relational-recovery}
  Let $q:X\to Q$ be an observer and let $A\subseteq X\times W$ be a
  set-valued certificate relation, where $A(x,w)$ means that certificate $w$ is
  accepted at state $x$.  A decoder on the observer quotient can return an
  accepted certificate for every source state if and only if every realized
  $q$-fiber has one certificate accepted at every state in that fiber.
  A pair $x,y$ in one fiber with no common accepted certificate therefore refutes
  recovery.

  Pairwise compatibility is weaker than recovery.  On $X=\mathrm{Fin}(3)$
  with constant observation and $A(i,w)\equiv(w\ne i)$, every pair has a common
  certificate, but the three-state fiber has none.  Function-valued targets are the
  special case $A(x,w)\equiv(w=P(x))$.  Information sufficiency and decoder
  language sufficiency are also distinct: on the identity observation of
  $\mathrm{Bool}$ the identity certificate is recoverable, while a
  language containing only the constant-false decoder cannot recover it.
\end{theorem}
\begin{proof}
  A quotient decoder supplies one certificate on each equivalence class, hence one
  common certificate on every fiber.  In the other direction, choose a common certificate for each
  observer class and use it as the quotient decoder.  The conflicting-pair
  obstruction is immediate.  In the three-state system the third index is a
  common certificate of any chosen pair, while choosing the decoder's own output as a
  source state contradicts global acceptance.  Singleton acceptance reduces
  the statement to ordinary target factorization.  The Boolean language control
  follows because the only permitted decoder returns false at the true input.
\end{proof}

\begin{theorem}[Set-valued side channels are fiber-hypergraph colorings]
  \label{thm:relational-side-channel}
  Let $s:X\to C$ be a side channel.  There exists a partial decoder, defined on
  all realized pairs $(q(x),s(x))$, that returns a certificate accepted at $x$ if
  and only if every subset of a single $q$-fiber having no common accepted
  certificate is non-monochromatic under $s$.  Thus the obstruction for a set-valued
  certificate problem is a hypergraph of incompatible fiber subsets; for
  singleton acceptance it reduces to a graph of incompatible pairs.

  The distinction is strict.  For the three-state system of
  Theorem~\ref{thm:relational-recovery}, every pair is compatible, a one-symbol
  side channel fails, and a two-symbol side channel succeeds.
\end{theorem}
\begin{proof}
  If a decoder exists, any monochromatic subset inside one observer fiber is
  contained in one realized observer/color cell and receives the cell's common
  certificate; hence a bad subset cannot be monochromatic.  In the other direction, if one
  realized cell had no common certificate, that cell itself would be a bad
  monochromatic fiber subset.  Choosing one certificate for every nonempty cell then
  constructs the partial decoder.  For the three-state system, the channel uses one color for
  states zero and one and a second color for state two; certificate two certifies
  the cell of states zero and one, and certificate zero certifies the cell of state two.
\end{proof}

\begin{theorem}[Capacity of set-valued recovery]\label{thm:set-valued-capacity}
  Let $X$ be finite, let $Q$ have decidable equality, and let every state have at least one accepted certificate. For a realized observation $o$, call a coloring of the fiber $q^{-1}(o)$ accepted when every nonempty color class has a certificate accepted at all of its states, and let $\chi_A(o)$ be the least number of colors of an accepted coloring; by Theorem~\ref{thm:relational-side-channel}, $\chi_A(o)$ is the chromatic number of the hypergraph on $q^{-1}(o)$ whose hyperedges are the subsets without a common accepted certificate. Put
  \[
    \kappa(q,A)=\max_{o\in q(X)}\chi_A(o).
  \]
  A side channel with $k$ symbols that admits a partial decoder in the sense of Theorem~\ref{thm:relational-side-channel} exists if and only if $\kappa(q,A)\le k$. For explicit enumerations of $X$ and $W$, decidable equality on $X$ and $W$, and decidable acceptance, an exhaustive coloring search returns $\kappa(q,A)$; it returns none if and only if some state has no accepted certificate, and in that case no side channel of any size exists. With singleton acceptance, $\kappa$ is the fiber multiplicity $\max_o|T_o|$ of Proposition~\ref{prop:side-channel-minimum}. For the three-state system of Theorem~\ref{thm:relational-recovery}, every pair of states has a common certificate and $\kappa=2$.
\end{theorem}
\begin{proof}
  Colorings of different fibers are independent, so a channel with $k$ symbols exists if and only if every fiber has an accepted coloring with $k$ colors. Colorings of the fibers with their least numbers of colors glue to a channel with $\kappa$ symbols, and a channel with fewer symbols would color the fiber attaining the maximum with fewer than its least number of colors. When every state has a certificate, one color per state is accepted, so the search, which tests the alphabet sizes up to the number of states, finds a least successful size, and that size is $\kappa$; a state without a certificate blocks every coloring of its fiber. Under singleton acceptance an accepted coloring separates distinct target values in each fiber. The three-state value is computed.
\end{proof}

For function-valued targets, $\chi_A(o)$ is the chromatic number of the restriction of Witsenhausen's confusability graph to the fiber~\cite{witsenhausen1976zeroerror}.

\begin{proposition}[Joint targets and minimum joint capacity]
  \label{prop:joint-target-capacity}
  For targets $P:X\to V$ and $R:X\to W$, define the joint target
  $J(x)=(P(x),R(x))$.  An observer $q$ licenses $J$ if and only if it licenses
  both $P$ and $R$.  If a side label $s:X\to C$ makes $P$ readable together with $q$, and a side label $t:X\to D$ makes
  $R$ readable together with $q$, then the product label $(s,t)$ makes
  $J$ readable, giving a product-alphabet upper construction.

  For finite nonempty $X$, however, the minimum side-alphabet size is
  \[
    \max_o\bigl|\{J(x):q(x)=o\}\bigr|,
  \]
  which can be smaller than the product of the separate minima.  With a constant observer
  on $\mathrm{Bool}$ and the identity Boolean target, one target needs two
  symbols; duplicating that same target still has joint multiplicity two, so the
  product upper bound four is not tight.
\end{proposition}
\begin{proof}
  Equality of the joint target is the conjunction of equality of its two
  coordinates.  Pairing two sufficient side labels therefore gives the upper
  construction.  Applying Proposition~\ref{prop:side-channel-minimum} to the
  joint target gives the finite minimum.  The duplicated Boolean target
  has only the two realized values $(0,0)$ and $(1,1)$ in its sole observation
  fiber, which proves the strict overcount control.
\end{proof}

\section{Noisy recovery, the Blackwell order, and sufficiency}\label{sec:noisy}

A prior on the states and a stochastic observation kernel measure recovery by Bayes risk. This section connects risk, refinement, entropy, and sufficiency on finite presentations, and it ends with the repair of an unlicensed target by a side channel or by refusal. The risk statements use rational priors and kernels; the entropy statements use real probability laws.

\begin{proposition}[Noisy recovery and budgeted side information]
  \label{prop:noisy-budgeted-recovery}
  Let $X,O,V$ be finite and let a nonnegative rational prior on $X$ together
  with a stochastic kernel $K(o\mid x)$ define the observation model.  For a
  target $P:X\to V$, the joint law determines the observation mass and, at every
  positive-mass observation, a nonnegative posterior summing to one.  The
  finite Bayes decoder minimizes the zero-one risk; its risk lies in $[0,1]$, is no
  larger than the risk of any deterministic or randomized decoder, and is zero
  if and only if the target is pure on every positive-mass observation cell.
  Deterministic or randomized post-processing cannot improve this Bayes optimum.

  Given explicit finite enumerations of a side-input alphabet and a
  side-output alphabet and a computable cost
  $c:(X\to C)\to\mathbb N$, the search enumerates the deterministic encoders,
  keeps those with $c(s)\le B$, and minimizes the resulting Bayes risk.  Whenever the
  search returns $s^*$, it is feasible and has minimum risk among the encoders
  within budget.  It returns none if and only if no feasible encoder exists, and
  when optima exist at $B_1\le B_2$, the optimum risk at $B_2$ is no larger.
  The statement concerns the correctness of this exhaustive search.
\end{proposition}
\begin{proof}
  The posterior is the normalized joint mass, so normalization and
  nonnegativity follow from the model axioms.  Pointwise maximization of posterior
  mass minimizes the zero-one error in each observation cell, and summing the cell
  errors gives the risk bounds and support-relative zero-risk criterion.  The
  post-processing inequalities compare the same joint law after deterministic
  or stochastic garbling.  For the budgeted result, the supplied enumerations
  generate every deterministic function table.  Filtering therefore removes
  every over-budget encoder and keeps every encoder within budget, and finite argmin supplies membership and
  optimality.  Set inclusion under $B_1\le B_2$ gives the risk monotonicity.
\end{proof}

\begin{theorem}[Blackwell order for deterministic observers]\label{thm:blackwell-deterministic}
  Let $X$, $O_1$, and $O_2$ be finite with decidable equality, let $O_2$ have an explicit enumeration, and let $q_i:X\to O_i$. Then $q_1$ refines $q_2$ if and only if, for every finite target type $V$ with an explicit enumeration, every target $P:X\to V$, and every rational prior on $X$, the minimum Bayes risk of guessing $P$ from $q_1(x)$ is at most the minimum Bayes risk of guessing $P$ from $q_2(x)$. For the reverse implication the target $P=q_2$ and priors supported on two states suffice. For every finite rational observation model, every rational post-processing kernel, and every target, post-processing cannot lower the minimum Bayes risk. The quantification over targets is necessary: for the constant target the two coordinate observers of $\mathrm{Bool}\times\mathrm{Bool}$ have equal Bayes risk under every prior, while the first coordinate does not refine the second.
\end{theorem}
\begin{proof}
  If $q_1$ refines $q_2$, then $q_2=g\circ q_1$ on the states, and any decoder from $q_2$ composes with $g$ to a decoder from $q_1$ with the same risk. In the other direction, suppose $q_1(x)=q_1(y)$ and $q_2(x)\neq q_2(y)$. Under the prior with mass one half on $x$ and on $y$, the target $q_2$ has Bayes risk zero through $q_2$ and one half through $q_1$. After post-processing, the joint law of the target and the new observation is a mixture of the joint law before it, and the sum of the maxima of a mixture is at most the mixture of the sums of the maxima. The control is computed.
\end{proof}

\begin{theorem}[Comparison of finite experiments]\label{thm:blackwell-finite}
  Let $X$, $O_1$, and $O_2$ be finite with $O_2$ nonempty, and let $E_i$ be real stochastic matrices from
  $X$ to $O_i$. For a finite nonempty action set $A$ and $u:X\times A\to\mathbb R$ put
  $v(E,u)=\sum_{o}\max_{a\in A}\sum_{x}E(x,o)\,u(x,a)$. Then $v(E_2,u)\le v(E_1,u)$ for every such $A$ and
  $u$ if and only if $E_2=E_1K$ for a stochastic matrix $K$ from $O_1$ to $O_2$. For the matrices of two
  observers this is refinement, and the order of Theorem~\ref{thm:blackwell-deterministic} is its
  restriction to zero-one utilities.
\end{theorem}
\begin{proof}
  If $E_2=E_1K$, an action rule for $E_2$ composed with $K$ is a randomized rule for $E_1$ with the
  same value, and a maximum dominates its averages. In the other direction, suppose $E_2$ lies outside the convex
  hull of the matrices $E_1G$ with $G$ deterministic. That hull is compact and convex, so a linear
  functional $W$ separates it from $E_2$. With actions $O_2$ and utility $u(x,a)=W(x,a)$, the value of
  $E_2$ is at least $\langle W,E_2\rangle$, while the value of $E_1$ equals $\langle W,E_1G\rangle$ for
  the rule $G$ that picks an optimal action at each observation, which is smaller. A point of the hull
  is $E_1K$ for the average $K$ of the deterministic matrices.
\end{proof}

With $u(x,a)=\pi(x)g(a,x)$ for a prior $\pi$ and a gain function $g$, $v$ is the posterior $g$-vulnerability, so Theorem~\ref{thm:blackwell-finite} is the comparison theorem of Blackwell~\cite{blackwell1951,blackwell1953} and Sherman~\cite{sherman1951} for finite experiments, and the leakage ordering theorem of McIver, Morgan, Smith, Espinoza, and Meinicke~\cite{mciverEtAl2014} for channels.

\begin{theorem}[Entropy bound by fiber multiplicity]\label{thm:entropy-fiber-bound}
  Let $X$, $Q$, and $V$ be finite with decidable equality, let $\mu$ be a probability law on $X$, and let $m=\max_o|T_o|$ be the fiber multiplicity. Then the conditional entropy of $P$ given $q$ under $\mu$ satisfies
  \[
    H_\mu(P\mid q)\le\log m
  \]
  in nats, and the corresponding bound $\log_2 m$ holds in bits. Equality holds if and only if every observation of positive mass realizes $m$ target values, each with conditional probability $1/m$. For every nonnegative weighting $\mu$, the conditional entropy vanishes if and only if $q$ licenses $P$ on the support of $\mu$. A nonuniform law on a two-state fiber with two target values lies strictly below $\log 2$, and a collision whose second state has mass zero has zero entropy while $q$ fails to license $P$ on $X$.
\end{theorem}
\begin{proof}
  The conditional entropy is the $\mu$-average over observations of the entropy of the conditional target law, and that law is supported on the at most $m$ target values of the fiber. The entropy of a law on at most $m$ points is at most $\log m$, with equality only for the uniform law on $m$ points. An average of nonnegative terms vanishes if and only if each term of positive weight does, and a conditional law has zero entropy if and only if it is a point mass. The two controls are computed.
\end{proof}

\begin{theorem}[Sufficiency and definability]\label{thm:sufficiency-definability}
  (a) Let $X$ be finite with decidable equality and a rational prior, and let the observation be the state itself. A statistic $T:X\to S$ is sufficient for $P$, in the sense that the posterior law of $P$ is equal at any two positive-mass observations with equal statistic, if and only if $T(x)=T(y)$ implies $P(x)=P(y)$ for all states $x,y$ of positive prior mass; for a prior of full support this is the license of $P$ by $T$. For every observer $q$ into a finite type with decidable equality, the minimum Bayes risk through $q$ equals the minimum Bayes risk of the identity observation if and only if $q$ is sufficient for $P$. For a prior of full support, the target-kernel quotient of Proposition~\ref{prop:target-kernel-quotient} is sufficient, and every sufficient statistic refines it. In the binary model whose observation is independent of the state, the identity statistic is sufficient while $P$ is not constant on the joint support; sufficiency and licensing separate outside deterministic observation.

  (b) For arbitrary sets with $V$ nonempty, a map $d:Q\to V$ with $P=d\circ q$ exists if and only if $P\circ\sigma=P$ for every permutation $\sigma$ of $X$ with $q\circ\sigma=q$. The hypothesis on $V$ is necessary: on the empty set of states with an empty target type every permutation fixes the target, and no such map exists.
\end{theorem}
\begin{proof}
  With the identity observation, the posterior at a positive-mass state is the point mass at its target value, so equality of posteriors is equality of target values. The Bayes risk of the identity observation is zero, and the risk through $q$ is zero if and only if $P$ is constant on the positive-mass states of each $q$-fiber. By Proposition~\ref{prop:target-kernel-quotient}, the target-kernel quotient licenses $P$ and every licensing observer refines it. In the binary model whose observation is independent of the state, the Bayes risk of the identity target is one half. Part~(b) combines Theorem~\ref{thm:license-stabilizer} with the construction of a total decoder from a nonempty target type.
\end{proof}

Part~(a) is the deterministic instance of the minimal sufficient statistic~\cite{lehmannScheffe1950,bahadur1954}, and part~(b) is the form for sets of Padoa's method, whose first-order version is Beth's definability theorem~\cite{beth1953}. The general statistical and first-order statements are cited results.

\begin{theorem}[Minimal sufficiency for finite families]\label{thm:finite-family-sufficiency}
  Let $\Theta$ and $Y$ be finite, let $(p_\theta)_{\theta\in\Theta}$ be rational probability vectors on $Y$,
  and let $\pi$ be a prior on $\Theta$ with full support. In the model with state $(\theta,y)$, weight
  $\pi(\theta)p_\theta(y)$, and observation $y$, a statistic $T$ of the sample is sufficient for $\theta$
  if and only if $p_\theta(y)=g_\theta(T(y))\,h(y)$ for nonnegative $g$ and $h$. Two samples of positive
  likelihood have the same posterior if and only if their likelihood vectors are positive multiples of
  each other, and the posterior map is sufficient and is a function of every sufficient statistic on
  samples of positive likelihood. With a zero in the prior, a constant statistic can be sufficient
  while the likelihood does not factor through it.
\end{theorem}
\begin{proof}
  Completeness of the posterior map is the definition of sufficiency; the factorization is Fisher and Neyman's criterion, and on a full-support prior a zero mass forces every likelihood to vanish, so positive multiples carry the same posterior; the posterior map is sufficient by construction and a function of every sufficient statistic because equal statistics give equal posteriors.
\end{proof}

\subsection{Repair of an unlicensed target}

When $q$ fails to license $P$, Proposition~\ref{prop:license-criterion} supplies a collision: states $x,y$ with $q(x)=q(y)$ and $P(x)\neq P(y)$. Two repairs remove a collision. A side channel $s:X\to C$ with $s(x)\neq s(y)$ separates the pair in the joint observer $(q,s)$; we call this repair a \emph{lift}. A retained set $S\subseteq X$ restricts the answers to the states of $S$, and when $S$ omits $x$ or $y$ the collision leaves $S$; we call this repair a \emph{refusal}. The observer $q$ licenses $P$ on $S$ when $q(x)=q(y)$ implies $P(x)=P(y)$ for all $x,y\in S$. A decoder $d:Q\to V$ has the correctness set $C_d=\{x\in X:P(x)=d(q(x))\}$, and $d$ is \emph{fiber-attained} when $d(q(x))\in T_{q(x)}$ for every state $x$, with $T_o$ as in Proposition~\ref{prop:side-channel-minimum}. For finite $X$ and a rational weight $\mu$ on $X$, the retained set $S$ refuses the weight $\mu(X\setminus S)$. A \emph{$k$-list} assigns to each observation $o$ a set $M(o)\subseteq V$ with $|M(o)|\le k$ and keeps the states $x$ with $P(x)\in M(q(x))$; let $\varepsilon_\mu(k)$ be the least weight that a $k$-list refuses.

\begin{theorem}[Lift or refuse]\label{thm:lift-or-refuse}
  Let $q:X\to Q$ and $P:X\to V$ be arbitrary maps, and put $m=\max_o|T_o|$ when $X$ is finite.
  \begin{enumerate}[label=(\alph*),leftmargin=1.6em]
    \item \emph{Collision law.} If $x$ and $y$ collide and $(q,s)$ licenses $P$ on $S$, then $x\notin S$, or $y\notin S$, or $s(x)\neq s(y)$.
    \item \emph{Refusal.} For nonempty $V$, the observer $q$ licenses $P$ on $S$ if and only if $S\subseteq C_d$ for some decoder $d$, and the maximal sets on which $q$ licenses $P$ are the correctness sets of the fiber-attained decoders. For a collision of $x$ and $y$, some maximal such set contains $x$ and omits $y$, and another contains $y$ and omits $x$.
    \item \emph{Lift.} The least number of side symbols through which $(q,s)$ licenses $P$ is $m$ (Proposition~\ref{prop:side-channel-minimum}). If $(q,s)$ licenses $P$ and $x$ and $y$ collide, then $s$ followed by the transposition of the symbols $s(x)$ and $s(y)$ is a second licensing side channel, and it differs from $s$ at $x$.
    \item \emph{Mixed repair.} For $k\ge1$ and finite $S$, a side channel with $k$ symbols through which $(q,s)$ licenses $P$ on $S$ exists if and only if some $k$-list keeps every state of $S$.
    \item \emph{Costs.} Let $X$, $Q$, and $V$ be finite with decidable equality, let $V$ have an explicit enumeration, and let $\mu\ge0$. When $\mu$ is a prior, some retained set on which $q$ licenses $P$ refuses the minimum Bayes risk of guessing $P$ from $q$, and every such set refuses at least that much. For $k\ge1$, some repair with a $k$-symbol side channel refuses $\varepsilon_\mu(k)$, and every such repair refuses at least $\varepsilon_\mu(k)$. The least refused weights satisfy $\varepsilon_\mu(0)=\mu(X)$; $\varepsilon_\mu(1)$ is the minimum Bayes risk when $\mu$ is a prior; $\varepsilon_\mu(k+1)\le\varepsilon_\mu(k)$ and $2\varepsilon_\mu(k+1)\le\varepsilon_\mu(k)+\varepsilon_\mu(k+2)$; and $\varepsilon_\mu(k)=0$ for $k\ge m$. When $\mu$ is positive at every state, $\varepsilon_\mu(k)=0$ holds if and only if $k\ge m$, and for a prior of full support $q$ licenses $P$ if and only if $m\le1$, if and only if the minimum Bayes risk is zero.
    \item \emph{Two states.} For $X=\{0,1\}$, a constant observer, and the identity target, $m=2$, the maximal retained sets are $\{0\}$ and $\{1\}$, and under the uniform prior $\varepsilon(1)=1/2$ and $\varepsilon(2)=0$.
  \end{enumerate}
\end{theorem}
\begin{proof}
  (a) If $x,y\in S$ and $s(x)=s(y)$, the joint observations of $x$ and $y$ agree, and the license on $S$ forces $P(x)=P(y)$.
  (b) A correctness set is licensed, since $P(x)=d(q(x))=d(q(y))=P(y)$ for $x,y\in C_d$ with $q(x)=q(y)$. In the other direction, on a licensed $S$ the decoder that reads each observation off a retained state with that observation is well defined and correct on $S$. Extended on the fibers that miss $S$ by a value attained there, it is fiber-attained and its correctness set contains $S$, so a maximal $S$ equals that set. If $d$ is fiber-attained and $C_d\subseteq S'$ with $q$ licensing $P$ on $S'$, every $x\in S'$ has a state $y$ in its fiber with $P(y)=d(q(x))$; then $y\in C_d\subseteq S'$, the license gives $P(x)=P(y)$, and $x\in C_d$. The decoders built from $\{x\}$ and from $\{y\}$ give the two maximal sets at a collision.
  (c) Proposition~\ref{prop:side-channel-minimum} gives the least alphabet. A transposition is a bijection of the alphabet, so it keeps every separation, and part~(a) gives $s(y)\neq s(x)$.
  (d) Given $s$, the list $M(o)=P(S\cap q^{-1}(o))$ keeps $S$, and distinct values in one list receive distinct symbols, so $|M(o)|\le k$. Given a $k$-list, choose injections $e_o:M(o)\hookrightarrow\{1,\dots,k\}$ and put $s(x)=e_{q(x)}(P(x))$.
  (e) The weight refused by $C_d$ is the zero-one risk of $d$. Every licensed set lies in a correctness set by part~(b), and the refused weight decreases as the set grows, so the least refused weight is the least risk, which is the Bayes risk. By part~(d), repairs with $k$ symbols and $k$-lists refuse the same weights. A $1$-list is a decoder, which gives $\varepsilon_\mu(1)$; the empty list gives $\varepsilon_\mu(0)$; the list $M(o)=T_o$ keeps every state once $k\ge m$; and when every state has positive weight, a list refusing zero weight keeps every state, so $|M(o)|\ge|T_o|$ at every realized $o$. For convexity, take a least $k$-list $A$ and a least $(k+2)$-list $B$. At each observation where $B(o)\setminus A(o)$ is nonempty, move one of its values from $B(o)$ to $A(o)$. The two new lists have at most $k+1$ entries each, and together they keep every state as often as $A$ and $B$ do, so $2\varepsilon_\mu(k+1)\le\varepsilon_\mu(k)+\varepsilon_\mu(k+2)$. The last statement combines these facts with the zero-risk criterion of Proposition~\ref{prop:noisy-budgeted-recovery} and with Proposition~\ref{prop:side-channel-minimum}.
  (f) The fiber has two target values, the two constant decoders give the two singletons, a $1$-list keeps one state, and a $2$-list keeps both.
\end{proof}

The theorem concerns deterministic observers: parts (a) to (d) hold on arbitrary sets, and part (e) uses finite presentations and rational weights. At $k=1$ the repair is a pure refusal and costs the Bayes risk; at $k=m$ it is a pure lift and refuses zero weight. Between the two, each added side symbol lowers the least refused weight by at most the amount that the symbol before it removed. Part~(a) is the same statement at one collision: a repair separates the two states by a symbol or refuses one of them.

\begin{theorem}[Functional dependencies as licenses]\label{thm:fd-license}
  Let $t:X\to D^A$ be a family of tuples, and for $Y\subseteq A$ write $t_Y(x)$ for the restriction of
  $t(x)$ to $Y$. The dependency $Y\to Z$ holds on $S\subseteq X$ when tuples of $S$ that agree on $Y$
  agree on $Z$.
  \begin{enumerate}[label=(\alph*),leftmargin=1.6em]
    \item \emph{Dependency and license.} $Y\to Z$ holds on $S$ if and only if $t_Y$ licenses $t_Z$ on $S$.
    \item \emph{Implication.} For finite attribute sets and a finite list $\Delta$ of dependencies, the
      following are equivalent: $Y\to Z$ is derivable from $\Delta$ by reflexivity, augmentation, and
      transitivity; every instance that satisfies $\Delta$ satisfies $Y\to Z$; every instance with two
      tuples and Boolean values that satisfies $\Delta$ satisfies $Y\to Z$.
    \item \emph{Repairs and answers.} For maps $q:X\to Q$ and $P:X\to V$ with $V$ nonempty, the subset
      repairs of $q\to P$ are the maximal sets on which $q$ licenses $P$. A state lies in every repair if
      and only if its fiber has one target value, and every state lies in some repair.
    \item \emph{Count.} For finite $X$ the number of subset repairs is $\prod_{o\in q(X)}|T_o|$.
    \item \emph{Error measures.} Let $X$ be finite and nonempty, and let $n_{o,v}$ count the states with
      $q=o$ and $P=v$. The measure $g_3$ equals the Bayes risk of guessing $P$ from $q$ under the uniform law
      on $X$; the measure $g_2$ is the fraction of states outside the intersection of all repairs; the number
      of ordered pairs counted by $g_1$ is $\sum_o\bigl(n_o^2-\sum_v n_{o,v}^2\bigr)$. The three
      measures satisfy $g_3\le g_2$ and $g_1\le g_2$, and each of them is zero if and only if $q$ licenses $P$.
    \item \emph{Sets of dependencies.} A finite list of dependencies holds on $S$ if and only if each
      left restriction licenses its right restriction on $S$. A family of licenses $q_i\to P_i$ is the list
      of dependencies $\{q_i\}\to\{P_i\}$ over the attributes that hold the values of $q_i$ and $P_i$.
  \end{enumerate}
\end{theorem}

\begin{proof}
  (a) Agreement on $Y$ is equality of $t_Y$. (b) Soundness is induction on derivations. For
  completeness let $Y^+$ be the set of attributes $a$ for which $Y\to\{a\}$ is derivable, and take
  the tuple that is true everywhere and the tuple that is true on $Y^+$ and false off it. They satisfy
  $\Delta$ and agree on $Y^+$ and disagree off it, so every dependency they satisfy from $Y$ has its right side
  in $Y^+$. (c) is Theorem~\ref{thm:lift-or-refuse}(b); a state whose fiber has one value lies in
  every correctness set of a fiber-attained decoder, and a collision gives a maximal set that omits
  either state. (d) A repair corresponds to the choice of one attained value on each realized fiber.
  (e) A licensed set meets each fiber inside one cell, so a largest licensed set keeps a most frequent
  value on each fiber, which is a Bayes decoder under the uniform law; the pairs in one fiber number
  $n_o^2$ and the pairs in one cell number $n_{o,v}^2$. (f) Apply (a) to each dependency.
\end{proof}

Part~(b) is the completeness of the rules of Armstrong~\cite{armstrong1974}, and part~(c) gives the
consistent answers of Arenas, Bertossi, and Chomicki~\cite{arenasBertossiChomicki1999} for the query
that returns the retained states. The measures $g_1$, $g_2$, and $g_3$ are those of Kivinen and
Mannila~\cite{kivinenMannila1995}. By part~(f), the least weight refused by a set on which several
licenses hold is the weight of an optimal subset repair for a list of dependencies, the problem
whose complexity Livshits, Kimelfeld, and Roy classify~\cite{livshitsKimelfeldRoy2020}; Livshits,
Kimelfeld, and Wijsen classify the complexity of counting subset repairs~\cite{livshitsKimelfeldWijsen2021},
and part~(d) is the count for one dependency.

\begin{proposition}[Repair curve and guessing within $k$ tries]\label{prop:frontier-ktry}
  Let $X$, $Q$, and $V$ be finite with decidable equality, let $\mu:X\to\mathbb Q$, and let $k\in\mathbb N$.
  Put
  \[
    V_k(\mu)=\sum_{o\in Q}\ \max_{W\subseteq V,\ |W|\le k}\ \mu\{x: q(x)=o,\ P(x)\in W\}.
  \]
  Then $\varepsilon_\mu(k)=\mu(X)-V_k(\mu)$. For a prior, $V_k(\mu)$ is the posterior vulnerability of the
  target value under the gain function for guessing within $k$ tries~\cite{alvimEtAl2012}, and
  $\varepsilon_\mu(k)=1-V_k(\mu)$; for $\mu\ge0$ and $k\le|V|$ the lists of at most $k$ values may be
  replaced by lists of size $k$. In particular $V_1(\mu)$ is one minus the Bayes risk, and
  $k\mapsto V_k(\mu)$ is concave.
\end{proposition}

\begin{proof}
  A menu with at most $k$ values per observation refuses $\mu(X)$ minus the weight it keeps, and the
  kept weight is a sum over observations of terms that depend on one list each. The least refused
  weight therefore takes the best list at each observation. Concavity is the convexity of
  $\varepsilon_\mu$ in Theorem~\ref{thm:lift-or-refuse}(e).
\end{proof}

\begin{proposition}[The lift has a universal solution and refusal has none]\label{prop:lift-universal}
  Let $q:X\to Q$ and $P:X\to V$ be maps.
  \begin{enumerate}[label=(\alph*),leftmargin=1.6em]
    \item The joint observer $\langle q,P\rangle$ refines $q$ and licenses $P$, and every observer that
      refines $q$ and licenses $P$ refines $\langle q,P\rangle$. Its kernel is $\ker q\cap\ker P$, and every
      observer with these three properties has this kernel.
    \item Some set on which $q$ licenses $P$ contains every such set if and only if $q$ licenses $P$. At a
      collision the sets on which $q$ licenses $P$ have at least two maximal members and no greatest member.
  \end{enumerate}
\end{proposition}

\begin{proof}
  (a) Equal joint observations are equal $q$-values and equal $P$-values, so $\langle q,P\rangle$ licenses $P$ and refines $q$. If $r$ refines $q$ and licenses $P$, then $r(x)=r(y)$ forces $q(x)=q(y)$ and $P(x)=P(y)$, hence $\langle q,P\rangle(x)=\langle q,P\rangle(y)$. The kernel of $\langle q,P\rangle$ is $\ker q\cap\ker P$. Any observer with these three properties refines $\langle q,P\rangle$ and is refined by it, so it has this kernel.
  (b) If $q$ licenses $P$, then $X$ is a licensed set containing every licensed set. At a collision, Theorem~\ref{thm:lift-or-refuse}(b) supplies two maximal licensed sets, and each omits a state of the collision that the other retains.
\end{proof}

The lift side is the zero-error side-information problem of Section~\ref{sec:capacity}. The least side alphabet of Theorem~\ref{thm:lift-or-refuse}(c) realizes the universal lift up to a relabeling of symbols on each fiber.

\section{Dynamics, memory, and sequential resolution}\label{sec:dynamics}

Observers change along a process. This section treats monotone sequences of observers, in which each observer refines the next, together with observers evolved by a map, retained histories, and access models in which an output becomes readable at a later stage.

\begin{theorem}[First loss of a license]\label{thm:license-cutoff}
  Let $P:X\to V$ and $q_n:X\to Q_n$ be arbitrary maps such that
  $q_n(x)=q_n(y)$ implies $q_{n+1}(x)=q_{n+1}(y)$.
  Either every $q_n$ licenses $P$, or there is a single $N\in\mathbb N$
  such that $q_n$ licenses $P$ if and only if $n<N$.
  In the latter case, one pair $x,y$ satisfies $P(x)\neq P(y)$ and
  \[
    q_n(x)\neq q_n(y)\ (n<N),\qquad
    q_n(x)=q_n(y)\ (n\geq N).
  \]
\end{theorem}
\begin{proof}
  Licensing at a later stage implies licensing at every earlier stage.
  If licensing fails, take its least failure index $N$ and a pair with
  equal $q_N$ values and different targets. Earlier licensing separates
  this pair; monotonicity identifies it at every later stage.
  The set of licensed indices determines $N$. The argument also covers
  $N=0$ and holds on arbitrary carriers, with or without a license at stage zero.
\end{proof}

\begin{theorem}[A finite bound for persistent licensing]\label{thm:finite-license-bound}
  Let $F:X\to X$, $q:X\to Q$, and $P:X\to V$. Assume the attained image
  $q(X)$ is finite and
  \[
    q(x)=q(y)\ \Longrightarrow\ q(F(x))=q(F(y)).
  \]
  Put $q_n=q\circ F^n$ and $B=\max\{|q(X)|-1,0\}$. Then
  \[
    \ker q_n=\ker q_B\quad(n\geq B),\qquad
    q_B\text{ licenses }P\ \Longleftrightarrow\
    \forall n\in\mathbb N,\ q_n\text{ licenses }P.
  \]
  If the license is lost, its first loss occurs at a stage $N\leq B$.
  The target is evaluated on the original input $x$, and the observer $q_n$
  reads the evolved state $F^n(x)$. The sets $X$, $Q$, and $V$ may be infinite.
  For every $k\geq1$, the bound is attained with
  $X=Q=\{0,\ldots,k\}$, $q=\mathrm{id}$,
  $F(i)=\max\{i-1,0\}$, and $P(i)=[i=k]$: the license holds at stage
  $n$ if and only if $n<k$.
\end{theorem}
\begin{proof}
  Compatibility induces a deterministic transition on the attained
  observations. Every observed orbit reaches a cycle by stage $B$.
  The transition is injective on its cyclic states: if two successors
  coincide, iteration through a common multiple of their positive return
  times recovers equality of the original observations. Thus no new
  observer collision occurs after $B$. Compatibility propagates all
  earlier collisions, proving kernel stability. Licensing depends only
  on the kernel, and Theorem~\ref{thm:license-cutoff} gives the remaining
  statements. If $q(X)$ is empty, then $X$ is empty and all licenses hold.
  For the stated family, $F^n(i)=\max\{i-n,0\}$. Before stage $k$,
  the original maximum has a distinct image from every smaller input;
  at stage $k$, all images coincide while the two target values differ.
\end{proof}

\begin{theorem}[Compactness of finite resolving channels]\label{thm:channel-compactness}
  Let $I$ be any index set, $q_i:X\to Q_i$ and $P_i:X\to V_i$ arbitrary
  maps, and $C$ a finite alphabet. A single channel $s:X\to C$ licenses
  every $P_i$ through $(q_i,s)$ if and only if, for every finite $J\subseteq I$,
  one channel licenses all requirements indexed by $J$.

  Hence for a monotone sequence $q_n:X\to Q_n$ and fixed target
  $P:X\to V$,
  \[
    \bigl(\exists s:X\to C\ \forall n,\ (q_n,s)\text{ licenses }P\bigr)
    \quad\Longleftrightarrow\quad
    \bigl(\forall n\ \exists s_n:X\to C,\ (q_n,s_n)\text{ licenses }P\bigr).
  \]
  The input and observation sets may be infinite, and the kernels may change
  at every stage. The channel is evaluated on the original input.
\end{theorem}
\begin{proof}
  Give $C$ the discrete topology. The product $C^X$ is compact.
  For each requirement $i$, the set of resolving channels is
  \[
    S_i=\bigcap_{q_i(x)=q_i(y),\ P_i(x)\neq P_i(y)}
    \{s:s(x)\neq s(y)\}.
  \]
  Each coordinate constraint is closed, hence $S_i$ is closed. The finite
  requirements give the finite-intersection property, so compactness gives
  a common channel. The reverse implication follows by restriction.
  For a monotone sequence, a channel resolving the greatest index in a
  finite set resolves all earlier indices. The empty finite set has a
  channel supplied at stage zero. Apply the first assertion.
  The proof uses compactness of $C^X$ and is nonconstructive.
\end{proof}

The stagewise conclusion requires monotonicity.
On $X=\{0,1,2\}$, let $q_i(x)=[x=i]$ and $P(x)=x$.
Each observer has a resolving two-symbol channel, but a common channel
would separate all three pairs and therefore need three symbols.

\begin{corollary}[Permanent channel capacity]\label{cor:permanent-channel-capacity}
  Under Theorem~\ref{thm:finite-license-bound}, fix $s:X\to C$, where $C$
  may be infinite. This channel licenses $P$ through every $(q_n,s)$
  if and only if it does so through $(q_B,s)$.
  Thus a permanent $m$-symbol channel exists if and only if every
  $q_B$-fiber contains at most $m$ distinct target values.
  For any alphabet $C$, the existence of a separate channel at each stage
  is equivalent to the existence of one channel for all stages.
\end{corollary}
\begin{proof}
  Every $q_n$ refines $q_B$, by monotonicity before $B$ and kernel stability
  afterward. Appending the same $s$ preserves this refinement, so a license
  at $B$ proves a license at every stage. The finite-alphabet criterion is
  Proposition~\ref{prop:side-channel-minimum} applied to $q_B$. For stagewise
  existence, choose the channel at $B$.
\end{proof}

\begin{theorem}[Retained observation histories keep earlier licenses]
  \label{thm:sequential-memory}
  For stage observations $q_n:X\to Q$, define the retained history through
  horizon $N$ by
  \[
    H_N(x)=(q_0(x),q_1(x),\ldots,q_N(x)),
  \]
  as a function on $\mathrm{Fin}(N+1)$.  Then $H_N$ refines every
  $q_j$ with $j\le N$, and $H_{N+1}$ refines $H_N$.  Hence any target
  licensed by one retained stage is licensed by the history, and once a history
  licenses a target, every extension of the history licenses it.

  The current observation can be less informative than the retained history.  There is
  a two-state sequence whose stage-zero observer is the identity and whose
  stage-one observer is constant: the identity target is licensed at stage zero,
  is not licensed by the current stage-one observation, and remains licensed by
  $H_1$.
\end{theorem}
\begin{proof}
  Projection of $H_N$ to coordinate $j$ recovers $q_j$, proving refinement.
  Restriction of $H_{N+1}$ to its first $N+1$ coordinates recovers $H_N$.
  Licensing is monotone under refinement.  For the control, take Boolean states,
  $q_0(x)=x$, and $q_n(x)=\mathsf{false}$ for $n>0$.  The current stage merges
  false and true, while the retained history still contains the injective
  stage-zero coordinate.
\end{proof}

\begin{theorem}[First sufficient horizon]\label{thm:first-sufficient-horizon}
  For stage observations $q_n:X\to Q$ and a target $P:X\to V$ on arbitrary sets, let $H_N$ be the retained history of Theorem~\ref{thm:sequential-memory}. Either no horizon licenses $P$, or there is a single $N_0$ such that $H_N$ licenses $P$ if and only if $N_0\le N$. When $H_N$ licenses $P$, the history $H_N$ refines the target-kernel quotient $\pi_P$, the quotient licenses $P$, and every memory $M\circ H_N$ that licenses $P$ refines $\pi_P$; the quotient is the minimal licensing memory at that horizon. On Boolean states, the observer that hides stage zero and reveals every later stage licenses the identity target at horizon $N$ if and only if $N\ge1$, and the observer that reveals stage zero and forgets every later stage licenses it at every horizon.
\end{theorem}
\begin{proof}
  A license of $H_N$ persists to every longer horizon, because $H_{N+1}$ refines $H_N$. If some horizon licenses $P$, the least such horizon is $N_0$, and persistence gives the equivalence; two values of $N_0$ with this equivalence bound each other. The memory statement is Proposition~\ref{prop:target-kernel-quotient} applied to $H_N$ and to $M\circ H_N$. The two Boolean observers are evaluated at horizons zero and one.
\end{proof}

\subsection{Access, retrieval, and sequential resolution}
\label{subsec:query-access}

An access model fixes what an observer of a computation can read, and it separates two ways an output becomes available: retrieval at the initial access point, and resolution through a trace whose later states decide it.

\begin{definition}[Access model]\label{def:access-model}
  An \emph{access model} consists of a state space, an observer view $\mathrm{view}$ on the states, a target value $\mathrm{vd}$, and for each stage $k\in\mathbb N$ a trace observation $\mathrm{tr}_k$. A map $g$ is \emph{accessible from} a map $f$ on the same domain when $g=d\circ f$ for some map $d$. The \emph{initial access point} is $x\mapsto(\mathrm{view}(x),\mathrm{tr}_0(x))$, and the \emph{stage-$k$ access point} is $x\mapsto(\mathrm{view}(x),\mathrm{tr}_k(x))$.
\end{definition}

\begin{definition}[Non-vacuous query, direct retrieval, sequential resolution]\label{def:direct-vs-sequential}
  A query about the output is \emph{non-vacuous} when $\mathrm{vd}$ is inaccessible from $\mathrm{view}$ alone. The query is \emph{resolved by direct retrieval} when $\mathrm{vd}$ is accessible from the initial access point, and \emph{resolved by sequential uncertainty reduction at stage $K$} when $\mathrm{vd}$ is inaccessible from the initial access point and accessible from the stage-$K$ access point.
\end{definition}

\begin{proposition}[Retrieval excludes sequential resolution]\label{prop:retrieval-vs-resolution}
  If the output is accessible from the observer view alone, the query is resolved by direct retrieval. Direct retrieval and sequential uncertainty reduction exclude each other at every stage, and a query resolved by sequential uncertainty reduction is non-vacuous.
\end{proposition}

\begin{proof}
  A decoder showing accessibility from the view extends to the initial access point by reading the view component and discarding the trace component. The first clause of sequential resolution denies direct retrieval, which gives the exclusion. For the third claim, accessibility from the view alone would give direct retrieval by the first claim, contradicting the exclusion.
\end{proof}

\begin{proposition}[Sequential resolution requires a hidden statistic]\label{prop:hidden-state}
  Suppose the query is resolved by sequential uncertainty reduction at stage $K$. Then the stage-$K$ trace observation $\mathrm{tr}_K$ is itself inaccessible from the observer view, and some statistic $\varphi$ of the state, inaccessible from the view, makes the output accessible from the pair $x\mapsto(\mathrm{view}(x),\varphi(x))$.
\end{proposition}

\begin{proof}
  Assume $\mathrm{tr}_K$ were accessible from the view through a decoder $\ell$. Composing the stage-$K$ output decoder with $x\mapsto(\mathrm{view}(x),\ell(\mathrm{view}(x)))$ makes the output accessible from the view alone, contradicting non-vacuity of the resolved query (Proposition~\ref{prop:retrieval-vs-resolution}). Taking $\varphi=\mathrm{tr}_K$ supplies the statistic.
\end{proof}

Under a full-support law on a finite state space, the output is accessible from a view if and only if its conditional Shannon entropy given that view is zero, so a non-vacuous query is one whose output retains positive conditional entropy. On the recursor the hidden statistic is the progress coordinate: Proposition~\ref{prop:hidden-progress-entropy} measures its recovery cost through four observer channels, and Corollary~\ref{cor:meta-trace-mi} identifies the terminal record as the channel that resolves the completed depth. \Claimcoderef{} lists the formal statements of the access model and of the two propositions.

\begin{proposition}[A finite observer misses its own recurrence]\label{prop:finite-observer-recurrence}
  Let an observer map the states of an orbit to a finite observation alphabet. The orbit has a first repeated observation, and at that stage the current observation fails to determine whether it has occurred before: an earlier stage has the same observation with the opposite answer. The repeat status is therefore a hidden statistic in the sense of Proposition~\ref{prop:hidden-state} for every finite alphabet, whatever the observer records.
\end{proposition}

\begin{proof}
  An orbit is infinite and the alphabet is finite, so by the pigeonhole principle two stages share an observation; take the least such later stage. Its observation occurs at an earlier stage whose repeat status is false, while its own repeat status is true, so the status is a function of the stage that fails to factor through the observation.
\end{proof}

\begin{theorem}[Minimum record cost of recurrence]\label{thm:repeat-depth-license}
  Let $q:\{0,\ldots,N-1\}\to Q$ be any observation map into an arbitrary
  set $Q$. Define
  \[
    d_q(n)=|\{j<n:q(j)=q(n)\}|,\qquad
    s_q(n)=\mathbf 1_{d_q(n)>0},\qquad
    K_q=\max_{a\in q(\{0,\ldots,N-1\})}|q^{-1}(a)|,
  \]
  where the empty maximum is zero. On a fiber of size $k$, the values of
  $d_q$ form the set $\{0,\ldots,k-1\}$. The least number of side-channel
  symbols needed to determine $d_q$ from the observation and that channel
  is $K_q$; the least fixed-length binary cost is
  $\lceil\log_2\max(1,K_q)\rceil$. The channel $n\mapsto d_q(n)$ attains
  these bounds. For $s_q$ the corresponding cost is zero if and only if
  $q$ is injective; otherwise it is one bit. Its maximum
  number of target values per observation fiber is zero for $N=0$,
  one for $N>0$ and injective $q$, and two otherwise.
  Every fixed finite side alphabet fails to determine repeat depth at
  some horizon, already for the constant observation map.
\end{theorem}
\begin{proof}
  List each observation fiber in increasing index order. Its $j$th index
  has $j$ earlier members, so the depth values enumerate that
  fiber. Any resolving channel must distinguish all of them, and the
  depth itself supplies a code with $K_q$ symbols. A binary code has
  enough words if and only if $K_q\le2^b$, giving the stated bit cost.
  On a singleton fiber the status is always zero; every larger fiber
  contains its first occurrence with status zero and a later occurrence
  with status one. This proves all three multiplicities, including the
  empty horizon. For a constant observer $K_q=N$, so an alphabet with
  $r$ symbols fails at horizon $r+1$.
\end{proof}

\section{The free recursor as a realization}\label{sec:recursor}

The laws of Sections~\ref{sec:observer} to~\ref{sec:dynamics} hold at the scopes stated in their hypotheses. This section and the next four present a worked realization with positive and negative controls: the free primitive-recursion system of the introduction.

Write $R_0$ for root application of the two displayed recursor rules and
$R_{\mathrm{ctx}}$ for their closure under all one-hole contexts. The free
term algebra has constructors $0,S,G,F$; the formalization also admits
inert unary and binary constructors, with congruence in every argument
and no rewrite rule at their root.

\begin{theorem}[Normalization of the free recursor]\label{thm:free-context-normalization}
  The relations $R_0$ and $R_{\mathrm{ctx}}$ terminate and are confluent on
  every term. There is an executable function $N$ such that
  $t\,R_{\mathrm{ctx}}^*\,N(t)$, $N(t)$ is normal, and
  \[
    u\mathbin{\downarrow}_{R_{\mathrm{ctx}}}v
    \quad\Longleftrightarrow\quad N(u)=N(v).
  \]
  Thus equality of computed normal forms decides joinability.
\end{theorem}
\begin{proof}
  Assign weight zero to $0$, add one at each unary constructor, sum the
  arguments and add one at each binary constructor, and put
  \[
    w(F(b,s,n))=w(b)+(w(s)+2)(w(n)+1).
  \]
  Both root rules strictly decrease $w$, as does their application inside
  any context. Normalize constructor arguments by recursion, then evaluate
  $F(b,s,0)$ as $b$ and $F(b,s,S(n))$ as $G(s,F(b,s,n))$ until the
  counter is neither $0$ nor a successor. This defines $N$. Induction proves
  reachability and invariance of $N$ under a step; strict weight decrease
  proves that $N(t)$ is normal. Two reducts therefore reduce to the same
  $N(t)$. The root relation is deterministic and inherits termination.
\end{proof}

\begin{proposition}[Additive whole-term aggregation fails on the duplicating step]
  \label{prop:additive-fails-on-dup}
  Let $\mu$ be any additive whole-term measure that assigns nonnegative weight to visible term structure and requires strict decrease across the step-duplicating rule
  \[
    F(x,y,S(n)) \to G(y,F(x,y,n)).
  \]
  Then $\mu$ fails to orient the rule uniformly over arbitrary substitutions for the step argument.
\end{proposition}

\begin{proof}
  The left-hand side contains one occurrence of the step argument, and the right-hand side contains two: one in the wrapper and one in the recursive call. After the shared context cancels, strict descent requires the fixed decrease of one successor to exceed the measure of an extra copy of the step argument. In an additive whole-term measure this forces an inequality of the form
  \[
    \mu(S) > \mu(y).
  \]
  But $y$ is a variable and can be instantiated by terms whose measure exceeds any given bound. Every fixed constructor-side contribution is exceeded by some substitution. Hence additive whole-term orientation fails here.
\end{proof}

The system terminates. The obstruction is the incompatibility between uniform additive descent on the visible term and an unbounded duplicated substitution.

\begin{proposition}[Trace of the recursor]\label{prop:trace-law}
  Fix ground terms $a,b$ that are normal forms and $k\in\mathbb N$. Write $S^k(0)$ for the $k$-fold successor and define
  \[
    t_0 := F(a,b,S^k(0)).
  \]
  Then every maximal rewrite sequence from $t_0$ has the shape
  \[
    t_i = G^i\!\bigl(b,\,F(a,b,S^{k-i}(0))\bigr)\quad(0\le i\le k),
  \]
  followed by the final step
  \[
    t_{k+1}=G^k(b,a),
  \]
  and terminates in $k+1$ steps. Along this trace, if $\mathrm{ctr}(t)$ is the visible successor-height of the third argument of the single active $F$ occurrence and $\mathrm{pay}(t)$ is the number of payload slots occupied by $b$, then
  \[
    \mathrm{ctr}(t_i)=k-i,\qquad \mathrm{pay}(t_i)=i+1\quad(0\le i\le k).
  \]
  So the counter decreases by one per step while payload multiplicity increases by one per step.
\end{proposition}

\begin{proof}
  Proceed by induction on $i$. The base case is immediate. For the induction step, assume
  \[
    t_i=G^i\!\bigl(b,F(a,b,S^{k-i}(0))\bigr),\quad i<k.
  \]
  The only redex is the inner $F(a,b,S^{k-i}(0))$, so one step of Rule~2 gives
  \[
    t_{i+1}=G^i\!\bigl(b,G(b,F(a,b,S^{k-i-1}(0)))\bigr)
    =G^{i+1}\!\bigl(b,F(a,b,S^{k-i-1}(0))\bigr).
  \]
  At $i=k$ we have $t_k=G^k(b,F(a,b,0))$, and one application of Rule~1 yields $t_{k+1}=G^k(b,a)$. The formulas for $\mathrm{ctr}$ and $\mathrm{pay}$ are read off the trace shape.
\end{proof}

\begin{proposition}[Derivational length and time-mass decoupling]\label{prop:exact-runtime-mass-rate}
  Every maximal derivation from $t_0$ has the single endpoint $t_{k+1}$ and length $k+1$, whatever the identities and structural weights of $a$ and $b$. Normalize the weight of each successor constructor to one. Let $\alpha=|a|$, $\beta=|b|$, $\gamma=|G|$, $\phi=|F|$, and let $\zeta$ be the zero-leaf weight. Set
  \[
    w:=\gamma+\beta,
    \qquad
    c:=\phi+\alpha+\beta+\zeta.
  \]
  Here $c$ includes the active payload occurrence in the live $F$-site. For $0\le i\le k$ the weighted syntax size is
  \[
    |t_i|=iw+(k-i)+c,
    \qquad
    |t_{k+1}|=kw+\alpha.
  \]
  Every recursive step has the constant mass-production law
  \[
    |t_{i+1}|+1=|t_i|+w\qquad(i<k),
  \]
  and the terminal step removes $\phi+\beta+\zeta$ units. The runtime therefore depends only on the counter, while the term size grows in proportion to the payload size.
\end{proposition}

\begin{proof}
  The relation has one redex at each live state and none at the terminal record. Induction over the orbit gives the single derivation, of $k+1$ steps. The size identities follow by induction over the $G$-stack, after which the two per-step identities are subtraction-free rearrangements of the displayed closed form.
\end{proof}

\begin{proposition}[One step removes a counter layer and adds a payload slot]\label{prop:exchange-rate}
  Each application of the duplicating step
  \[
    F(x,y,S(n)) \to G(y,F(x,y,n))
  \]
  does two things:
  \begin{enumerate}[nosep]
    \item it removes one $S$-layer from the counter;
    \item it creates one payload slot, the first argument of the new $G$-wrapper.
  \end{enumerate}
  The variable $y$ occurs once on the left-hand side and twice on the right-hand side. The new occurrence lies in the wrapper slot, which the residual counter comparison ignores.
\end{proposition}

\begin{proof}
  Read off the rule: the source has one $S$-layer and one $y$-occurrence in $F$, and the target drops the outer $S$-layer and has two $y$-occurrences, one in the new $G$-wrapper and one in the recursive $F$.
\end{proof}

\begin{proposition}[Offset conservation law]\label{prop:offset-conservation}
  For every live stage $t_i$ with $0\le i\le k$,
  \[
    \#b(t_i)-\#G(t_i)=1,
    \qquad
    \#b(t_i)=i+1,
    \qquad
    \#G(t_i)=i.
  \]
  After the base contraction, the terminal record $t_{k+1}$ satisfies
  \[
    \#b(t_{k+1})=\#G(t_{k+1})=k.
  \]
  Thus the offset-one identity applies to the live recursive stages; the terminal record has equal payload and wrapper counts.
\end{proposition}

\begin{proof}
  Along the trace of Proposition~\ref{prop:trace-law}, the $i$ completed frames and the active site contain the payload, so $\#b(t_i)=i+1$ and $\#G(t_i)=i$ on the live stages. At the base step the active site disappears and the $k$ frames remain, which gives the terminal counts. Both computations are inductions over the trace.
\end{proof}

\begin{proposition}[Reachable-increment conservation]\label{prop:conserved-lattice}
  Let $R\subseteq S\times S$ be any relation, let $s_0\in S$, let
  $O:S\to\mathbb Z^d$, and fix $a\in\mathbb Z^d$ and an offset $c\in\mathbb Z$.
  The affine reading
  \[
    \rho(s)=c+\sum_{j<d}a_j O_j(s)
  \]
  is constant on every state reachable from $s_0$ if and only if every reachable
  one-step transition $s\,R\,t$ satisfies
  \[
    \sum_{j<d}a_j\bigl(O_j(t)-O_j(s)\bigr)=0.
  \]
  For a supplied finite enumeration of $S$ and a decidable relation table, this
  criterion is executable: the finite checker returns a reachable
  violating edge when one exists, and returns none if and only if all reachable
  increments are annihilated.

  For the $r$-ary recursor orbit at live stage $0\le i\le k$, take the
  three coordinates to be remaining counter, live payload multiplicity, and
  wrapper count.  They are
  \[
    (k-i,\;ri+1,\;i).
  \]
  Hence, for $k>0$, the integer functional
  $\alpha\,\#S+\beta\,\#b+\gamma\,\#G$ is constant on the orbit if and only if
  \[
    \alpha=r\beta+\gamma.
  \]
  At $k=0$ every coefficient triple is conserved, since the orbit has no transition.  The binary case
  $r=1$ recovers $\alpha=\beta+\gamma$ and in particular the two familiar
  invariants $\#b-\#G$ and $\#S+\#G$.
\end{proposition}

\begin{proof}
  Subtracting two adjacent affine readings cancels the offset and leaves the dot
  product of $a$ with the observed increment.  If the reading is constant on the
  reachable component, every reachable edge therefore has zero dot product.
  In the other direction, induction along a finite $R$-path from $s_0$ telescopes the zero
  increments and gives the source reading at every reachable state.  On a finite
  explicit carrier, exhaustive search over represented reachable edges decides
  the same condition and returns a violating edge if and only if the universal
  condition fails.  For the $r$-ary orbit, the observation map
  identifies the three coordinates with $(k-i,ri+1,i)$; one step changes them by
  $(-1,r,1)$, so the reading changes by $-\alpha+r\beta+\gamma$.  Positive depth
  contains a step and therefore forces the displayed balance; at depth zero
  there is no transition to constrain the coefficients.
\end{proof}

\begin{definition}[Wrapper cell and wrapper-cell weight]\label{def:wrapper-cell}
  At step $i$, each wrapper layer is a \emph{wrapper cell} of total structural size
  \[
    w:=|G|+|b|.
  \]
\end{definition}

\begin{definition}[Wrapper stack and diagonal subset]\label{def:wrapper-stack}
  Let $B$ be the set of payload terms of the relevant sort. The \emph{wrapper stack} at step $i$ is
  \[
    v_i := [b_1,\ldots,b_i]\in B^i,
  \]
  where every component equals the same payload term $b$. Thus $v_i$ lies in the diagonal subset
  \[
    \Delta_i := \{(c,\ldots,c):c\in B\}\subset B^i.
  \]
\end{definition}

\begin{proposition}[Permutation symmetry of payload labels]\label{prop:gauge-symmetry}
  Label the $i+1$ payload occurrences in $t_i$ and let $\mathrm{Sym}_{i+1}$ permute those labels, including which label is attached to the active $F$-site. Then:
  \begin{enumerate}[nosep,leftmargin=1.6em]
    \item the constant payload tuple on the diagonal is fixed by every permutation in $\mathrm{Sym}_{i+1}$;
    \item direct additive observers are permutation-invariant on this orbit but remain \emph{multiplicity-sensitive}: on the constant tuple they return $(i+1)|b|$, and therefore grow strictly with the number of payload positions whenever $|b|>0$;
    \item dependency-pair extraction removes the surrounding wrapper context, and the subsequent counter projection selects the coordinate fixed by every payload-label permutation.
  \end{enumerate}
\end{proposition}

\begin{proof}
  All $i+1$ occurrences are the same ground term $b$, so relabelling them leaves the constant tuple unchanged. Additive mass is invariant under relabelling but still counts multiplicity: on the diagonal it is $(i+1)|b|$. Dependency-pair extraction removes the wrapper context; counter projection on the extracted pair retains only the counter.
\end{proof}

\begin{proposition}[The counter is the permutation-invariant decreasing coordinate]
  \label{prop:S-retained}
  Let $S^{K-i}(0)$ denote the counter term inside the active $F$-site at step~$i$
  of the trace. Then:
  \begin{enumerate}[nosep,leftmargin=1.6em]
    \item \textbf{Permutation invariance.} $S^{K-i}(0)$ is fixed by every permutation in
      $\mathrm{Sym}_{i+1}$: this group acts only on $b$-positions, and
      $S^{K-i}(0)$ is free of $b$-occurrences.
    \item \textbf{Strict subterm descent.} At each step,
      $S^{K-i}(0) \triangleright S^{K-i-1}(0)$: the next-step counter is a strict
      subterm of the counter at that stage.
    \item \textbf{Reconstructibility on successor inputs.} The map $S(n)\mapsto n$ is
      injective on successor terms: its source counter is reconstructed from $n$ by
      applying $S$. This reconstructibility claim applies only to the counter map.
    \item \textbf{Uniqueness.} $S^{K-i}(0)$ is the only component of the triple
      (F-counter, G-stack, $b$-copies) that is both permutation-invariant and
      strictly descending. The G-stack multiplicity $i$ increases, and the payload-multiplicity coordinate also increases.
  \end{enumerate}
  The dependency-pair transformation removes the surrounding $G$-context but retains the marked recursive call and all three arguments. Counter projection selects its third argument and ignores $x$ and $y$. Strict subterm descent on the projected $S$-terms proves chain-freeness; the dependency-pair theorem of Arts and Giesl~\cite{artsgiesl2000} transports that result to termination of the source TRS.
\end{proposition}

\begin{proof}
  \emph{(1)} The action of $\mathrm{Sym}_{i+1}$ permutes the $i+1$ $b$-positions
  in $t_i$. The counter $S^{K-i}(0)$ is a ground term over $\{S, 0\}$, free of
  $b$-occurrences. Every permutation fixes it.

  \emph{(2)} $S^{K-i}(0) = S(S^{K-i-1}(0))$, so $S^{K-i-1}(0)$ is a proper
  subterm of $S^{K-i}(0)$ by the subterm relation.

  \emph{(3)} The rule $F(x,y,S(n))\to G(y,F(x,y,n))$ replaces the source counter $S(n)$ with $n$ in the recursive call. Given $n$, its predecessor source within this counter map is reconstructed by prepending one $S$-application.

  \emph{(4)} G-multiplicity $\#G(t_i) = i$ strictly increases, and the
  payload-multiplicity coordinate $\#b(t_i)=i+1$ also strictly increases. By
  Proposition~\ref{prop:gauge-symmetry}, payload relabelling leaves the
  retained control coordinate fixed. Only $S^{K-i}(0)$ both decreases and is
  permutation-invariant.
\end{proof}

\begin{corollary}[The counter is a sufficient statistic for residual work]\label{cor:counter-exact-sufficient}
  From every live state $t_i$, the remaining maximal derivation has
  $\mathrm{ctr}(t_i)+1$ steps and terminates at $t_{k+1}$. The counter is therefore permutation-invariant and decreasing, and it determines the residual work.
\end{corollary}

\begin{proposition}[Licenses on the recursor]\label{prop:recursor-license-instances}
  On the live stages $0\le i\le k$ of the trace of Proposition~\ref{prop:trace-law}, the counter observer $i\mapsto\mathrm{ctr}(t_i)$ licenses the residual work $i\mapsto\mathrm{ctr}(t_i)+1$; for $k>0$ the constant observer fails to license the payload multiplicity $i\mapsto\mathrm{pay}(t_i)$, the stages $0$ and $1$ being a collision; and on one trace the counter observer is injective, so it licenses every target. On all recursor terms, the dependency-pair rank licenses the step count of the extracted-call machine and fails to license the step argument, which the terms $F(0,0,S(0))$ and $F(0,S(0),S(0))$ of equal rank show; the rank is therefore not injective. The rank licenses termination of the extracted calls, as does every observer, since the call relation is well founded. After adding a self-loop at every term whose step argument satisfies a predicate $\mathcal P$, the rank licenses termination if and only if $\mathcal P$ is empty; for $\mathcal P=\{S(0)\}$ the terms $F(b,0,S(n))$ and $F(b,S(0),S(n))$ have equal positive rank and opposite termination status, while every extracted call is retained.
\end{proposition}
\begin{proof}
  Residual work is a function of the counter, and on one trace the counter determines the stage, which gives the first three statements; stages $0$ and $1$ have payload multiplicities $1$ and $2$. The machine's step count is the successor-prefix length of the counter, which is the rank, so the rank licenses it, and the two displayed terms share the rank and differ in the step argument. A well-founded relation has a constant termination verdict, which every observer licenses. With the loops added, a term terminates if and only if its step argument avoids $\mathcal P$: the loop is a self-step, and the retained calls descend in rank. The term $F(b,0,S(n))$ terminates and $F(b,S(0),S(n))$ loops, at equal rank.
\end{proof}

\begin{remark}[Asymmetry of the reconstructible and persistent coordinates]\label{rem:reversibility-asymmetry}
  Each recursive firing removes one $S$-layer, reconstructible from the reduct counter by applying $S$, and creates one persistent $G$-frame, which no rule removes. Thus $S$-height decreases while $G$-multiplicity increases, and their sum remains $K$: the conservation law is $\mathrm{ctr}(t_i)+\#G(t_i)=K$. The dependency-pair transformation removes the accumulated $G$-context; counter projection retains the descending counter and ignores the payload coordinates under the dependency-pair soundness theorem.
\end{remark}

\subsection{Live computation and terminal record}

\begin{definition}[Terminal record map for the duplicator]\label{def:terminal-record-map}
  For the primitive self-duplicating recursor, define the terminal record map
  \[
    \mathcal R_K(X,Y):=G^K(Y,X).
  \]
  It records the terminal output after $K$ recursive unfoldings.
\end{definition}

\begin{proposition}[Live computation versus terminal record at the duplicator]\label{prop:live-vs-record}
  Along the trace of the primitive self-duplicating recursor,
  \[
    t_i=G^i\!\bigl(Y,F(X,Y,S^{K-i}(0))\bigr)
    \qquad (0\le i\le K),
  \]
  the following hold:
  \begin{enumerate}[leftmargin=1.4em]
    \item the single active occurrence of $F$ is the live computation site and is the only locus at which further recursive rule firings occur;
    \item the visible multiplicity of $G$-frames is $i$ at stage $i$ and therefore records the completed progress;
    \item the terminal normal form is
      \[
        t_{K+1}=G^K(Y,X)=\mathcal R_K(X,Y),
      \]
      in which no $F$-occurrence remains and the frame count $\#_G(t_{K+1})=K$ records the depth.
  \end{enumerate}
  Hence the duplicator separates live computation from terminal record: the active site disappears at termination, and the accumulated frames record the completed depth.
\end{proposition}

\begin{proof}
  The trace law (Proposition~\ref{prop:trace-law}) gives
  \[
    t_i=G^i\!\bigl(Y,F(X,Y,S^{K-i}(0))\bigr)
    \quad(0\le i\le K),
    \qquad
    t_{K+1}=G^K(Y,X).
  \]
  At every nonterminal stage there is a single active $F$-occurrence, the innermost recursive site. The number of visible $G$-frames is $i$ by inspection of the trace expression. At termination the base rule removes the last active $F$-site and leaves the record form $G^K(Y,X)$, whose $G$-multiplicity is $K$. So the final state records the depth, and the active site that produced the frames is gone.
\end{proof}

\begin{proposition}[Terminal-record completeness and the zero-depth boundary]\label{prop:terminal-record-complete-quant}
  Assume the base and payload slots carry distinct sorts, so that base terms and $G$-frames belong to distinct sorts. Then for every $K\ge 1$ the terminal record $G^K(Y,X)$ admits a decoder that recovers the base identifier, payload identifier, and depth $K$, and the positive-depth terminal-record map is injective in all three data. Sort separation is necessary: if the base slot may be frame headed, then $G^K(Y,G(Y,X'))=G^{K+1}(Y,X')$, so the depth and base term cannot both be recovered. At $K=0$ the record is $X$ alone, so the unused payload identifier is absent. If $\beta=|Y|$, the same computation also satisfies
  \[
    K^2\beta\le 2\,\mathrm{Con}(K,Y).
  \]
  Thus the quadratic payload mass is compatible with decoding the terminal record. The conditional-description interpretation of the decoder is stated in Remark~\ref{rem:kolmogorov-record}.
\end{proposition}

\begin{proof}
  The decoder peels the $G$-frames, requires every frame to contain the same payload identifier, counts their number, and reads the terminal base identifier. At positive depth this returns $(X,Y,K)$; frames with differing payloads are rejected. The mass inequality is the lower integer envelope of the triangular formula for $\mathrm{Con}(K,Y)$.
\end{proof}

\begin{corollary}[Least example of live-state and record separation]\label{cor:min-live-record}
  Combining Proposition~\ref{prop:live-vs-record} with Theorem~\ref{thm:min-instance}, the primitive self-duplicating recursor is the minimal structurally complete instance in the analyzed family at which the separation between live computation and terminal record is required for the proof. Direct whole-term proof languages measure the whole evolving term, while the decreasing quantity is the residual counter.
\end{corollary}

\begin{proof}
  Proposition~\ref{prop:live-vs-record} separates the active $F$-site from the persistent $G$-record. Theorem~\ref{thm:min-instance} identifies the duplicator as the only structurally complete member of the analyzed family without a direct whole-term certificate. The result follows.
\end{proof}
\section{The direct grammar as a license instance}\label{sec:oi}

Definition~\ref{def:op-incomplete} and Proposition~\ref{prop:grammar-license} apply to the closed direct-measure grammar of the introduction.

\begin{theorem}[Operational inexpressibility of the direct grammar]\label{thm:canonical-instance}
  Let $S_{\mathrm{DA}}$ be the closed reflected grammar generated by the forms counter, payload, constant, addition, multiplication, maximum, and scalar multiplication. A grammar expression is a candidate only when its interpretation uniformly orients the fixed duplicating carrier. Let $a_{\mathrm{dup}}$ be the primitive-recursion duplicator with rules $F(x,y,0)\to x$ and $F(x,y,S(n))\to G(y,F(x,y,n))$, and let $\pi_y$ return the step-argument slot of $F$. Then $S_{\mathrm{DA}}$ is operationally inexpressible for $a_{\mathrm{dup}}$ at $\pi_y$.
\end{theorem}

\begin{proof}
  The edge $(1,0)\to(0,1)$ shows that $\pi_y$ is nondegenerate. Proceed by induction on the seven grammar forms. A counter expression ignores $\pi_y$. A payload expression reads $\pi_y$ but increases on the displayed edge. Constants ignore $\pi_y$. For addition, multiplication by a nonnegative expression, maximum, and scalar multiplication, the induction hypotheses preserve the following dichotomy: either the result is independent of the payload, or some substitution makes the duplicated payload contribution at least as large on the right as on the left. In the latter case strict orientation fails. Thus every derivable expression either ignores $\pi_y$ or fails to orient the duplicating carrier, which is Definition~\ref{def:op-incomplete}(ii). Together with the displayed edge, this proves the claim.
\end{proof}

\begin{corollary}[Transfer across denotational subgrammars]
  \label{cor:universal-oi}
  Let $E$ be any certificate syntax with an interpretation
  \[
    \operatorname{ev}_E:E\longrightarrow(\mathbb N\times\mathbb N\to\mathbb N).
  \]
  If every denotation $\operatorname{ev}_E(x)$ is the denotation of an expression in $S_{\mathrm{DA}}$, then no $x\in E$ both depends on the payload coordinate and uniformly orients the duplicating carrier. Thus every proof system whose derivable certificate language is a denotational subgrammar of $S_{\mathrm{DA}}$ inherits Theorem~\ref{thm:canonical-instance}.

  Dropping the realization premise admits the following counterexample. The function
  \[
    \mu_*(c,p)=
    \begin{cases}
      2c+1,&p=0,\\
      2c,&p>0
    \end{cases}
  \]
  depends on the payload coordinate and uniformly orients every duplicating edge, but no expression of $S_{\mathrm{DA}}$ denotes it. The unrestricted assertion is false. Nonlinear polynomial interpretations and path orders supply construction certificates outside the grammar; the analyzed $\mathcal W_2$ projection family supplies confession certificates.
\end{corollary}

\begin{proof}
  Choose $e\in S_{\mathrm{DA}}$ with $\operatorname{ev}_E(x)=\operatorname{ev}(e)$. Theorem~\ref{thm:canonical-instance} excludes simultaneous payload dependence and orientation for $e$, hence for $x$. The realization premise is necessary: $\mu_*$ satisfies $\mu_*(c+1,p)>\mu_*(c,p+L)$ for every $L>0$, while $\mu_*(0,0)\ne\mu_*(0,1)$. If $\mu_*$ were denoted by a grammar expression, the first property would force payload blindness by Theorem~\ref{thm:canonical-instance}, contradicting the second.
\end{proof}

\begin{remark}[Rec$\Delta$ core and typed systems]
  \label{rem:recdelta-typed-baseorder}
  The barrier and dependency-pair projection already occur in the four-constructor, two-rule Rec$\Delta$ fragment~\cite{rahnamaOrientation}. The additive and affine barriers persist in typed and many-sorted presentations whenever the step sort contains terms of unbounded size. After dependency-pair extraction, the counter supplies a linear base order. The certificate-order increase therefore records the representation change needed to expose the descending coordinate.
\end{remark}

\begin{theorem}[Structural minimality of the duplicator]\label{thm:min-instance}
  Consider the six rule systems on the recursor symbol $F$ that include or omit the base rule $F(x,y,0)\to x$ and have no step rule, the linear step rule $F(x,y,S(n))\to F(x,y,n)$, or the duplicating step rule $F(x,y,S(n))\to G(y,F(x,y,n))$, and call a member structurally complete when it has the base rule and a step rule. The duplicating complete member is the only structurally complete member on which no additive whole-term measure decreases along every rule. Every member lies in one of three classes: the duplicating complete member, which has no additive direct certificate but has a natural-valued rank decreasing along every rule and a rank decreasing along every recursive call; the linear complete member, which has an additive direct certificate; and the four members that are not structurally complete.
\end{theorem}

\begin{proof}
  One additive measure decreases along the base rule and along the linear step rule, so every member without the duplicating step has an additive direct certificate. No additive measure decreases along the duplicating step. A quadratic rank decreases along the base rule and the duplicating step, and the dependency-pair projection decreases along the recursive call of either step rule. The classification follows by exhaustion over the six members; removing duplication from the duplicating complete member gives the linear complete member, and deleting its base or step rule leaves a member that is not structurally complete.
\end{proof}

\begin{remark}[Strongly connected components]
  \label{rem:scc-survival}
  The obstruction also occurs in a finite cyclic SCC whose composite cycle duplicates the payload: the additive and affine barriers apply to the composite source and target~\cite{rahnamaOrientation}. They also apply to a two-node SCC in which each local rule preserves payload multiplicity but the composite cycle exposes two wrapper-visible copies of one synchronized payload, with the affine case requiring wrapper dominance. Thus the same boundary occurs when duplication is distributed across a cycle.
\end{remark}

\subsection{Non-derivability of the projection from the rewrite signature}\label{subsec:non-derivability}

The dependency-pair projection lies outside the rewrite signature in the following sense. The \emph{Concrete Witness System} (CWS) is the seven-constructor, eight-rule calculus of the orientation analysis~\cite{rahnamaOrientation}; its signature is
\[
  \Sigma=\{\mathrm{void},\mathrm{delta},\mathrm{integrate},\mathrm{merge},\mathrm{app},\mathrm{recDelta},\mathrm{eqW}\},
\]
and every $\Sigma$-homomorphism whose recursor operation is constant in its third argument identifies two terms that the projection separates.

\begin{theorem}[The dependency-pair projection lies outside the signature]\label{thm:dp-not-signature-derivable}
  Let $S$ be a $\Sigma$-algebra over the CWS signature with $S.\mathrm{recDelta}\,x\,y\,z=S.\mathrm{recDelta}\,x\,y\,z'$ for all $x,y,z,z'$. Then every $\Sigma$-homomorphism $P\colon\mathrm{RecursorTerm}\to S.\text{carrier}$ assigns the same image to the displayed pair
  \[
    (\mathrm{recDelta}\,\mathrm{void}\,\mathrm{void}\,\mathrm{void},\;
    \mathrm{recDelta}\,\mathrm{void}\,\mathrm{void}\,(\mathrm{delta}\,\mathrm{void})).
  \]
  Every such homomorphism therefore identifies this pair, which the dependency-pair projection separates.
\end{theorem}

\begin{proof}
  Let $P$ be a $\Sigma$-homomorphism. By the substitution-invariance principle for the CWS signature, $P$ agrees with the fold from $\mathrm{RecursorTerm}$ into $S$ on every term, so the value of $P$ on either displayed term factors as $S.\mathrm{recDelta}\,(P\,\mathrm{void})\,(P\,\mathrm{void})\,(P\,z)$ for the appropriate $z$ slot. The hypothesis on $S.\mathrm{recDelta}$ forces both images to coincide.
\end{proof}

The proof uses substitution invariance of the free recursor algebra: every homomorphism from it is the fold. Proposition~\ref{prop:signature-hypothesis} shows that the constancy hypothesis on the counter slot is necessary.

\begin{proposition}[Scope of the signature theorem]\label{prop:signature-hypothesis}
  On the fragment of the signature with the symbols $\mathrm{void}$, $\mathrm{delta}$, and $\mathrm{recDelta}$, every algebra whose recursor operation is constant in its third argument identifies the pair of Theorem~\ref{thm:dp-not-signature-derivable} under its fold, and the counter-height algebra, which returns its third argument and adds one at each $\mathrm{delta}$, separates the pair; that algebra fails the constancy hypothesis at the arguments $(0,0,0)$ and $(0,0,1)$. The theorem is therefore relative to the class of homomorphisms constant in the counter slot.
\end{proposition}
\begin{proof}
  The fold of a constant-third-argument algebra sends both terms to the operation applied to the images of $\mathrm{void}$, which agree by constancy. The counter-height fold returns $0$ on the first term and $1$ on the second.
\end{proof}

Both signature obstructions of this paper, the projection here and the disequality guard of Section~\ref{sec:sibling}, are instances of one criterion about term definability over an arbitrary finitary signature.

\begin{theorem}[Homomorphic inexpressibility]\label{thm:homomorphic-inexpressibility}
  Let $\Sigma$ be a signature with finite arities, let $A$ and $B$ be $\Sigma$-algebras, and let $h:A\to B$ be a $\Sigma$-homomorphism, so that evaluation of terms commutes with $h$.
  \begin{enumerate}[label=(\alph*),leftmargin=1.6em]
    \item If $h$ identifies two elements that a map $\pi:A\to Z$ separates, then $\pi$ fails to factor through $h$: no map $d:B\to Z$ satisfies $d\circ h=\pi$.
    \item Let $h:A\to A$ be an endomorphism that fixes an element $\top$, and let $P$ be a binary predicate on $A$ with $P(a,b)$ and $\neg P(h(a),h(b))$ for some $a,b$. Then no term $t$ in two variables over $\Sigma$ satisfies $t(x,y)=\top$ if and only if $P(x,y)$ for all $x,y\in A$.
    \item An endomorphism that fixes $\top$ and identifies two distinct elements shows that disequality has no such term definition.
  \end{enumerate}
\end{theorem}
\begin{proof}
  (a) If $d\circ h=\pi$, then $\pi(a)=d(h(a))=d(h(b))=\pi(b)$. (b) Suppose $t$ defines $P$. Then $t(a,b)=\top$, and naturality of evaluation gives $t(h(a),h(b))=h(t(a,b))=h(\top)=\top$, so $P(h(a),h(b))$ holds, a contradiction. (c) is (b) for the disequality predicate.
\end{proof}

\begin{corollary}[Factorization form of the projection theorem]\label{cor:homomorphic-instances}
  For every $\Sigma$-algebra $S$ whose recursor operation is constant in its third argument, the fold $\mathrm{RecursorTerm}\to S$ is a $\Sigma$-homomorphism in the sense of Theorem~\ref{thm:homomorphic-inexpressibility} that identifies the pair of Theorem~\ref{thm:dp-not-signature-derivable}, and the counter projection, which returns $1$ on a term $\mathrm{recDelta}\,x\,y\,(\mathrm{delta}\,z)$ and $0$ on every other term, separates that pair. Hence the counter projection fails to factor through the fold, and through every $\Sigma$-homomorphism into such an algebra.
\end{corollary}
\begin{proof}
  A map that satisfies the seven preservation equations, one per constructor, is a $\Sigma$-homomorphism in the sense of the theorem; Theorem~\ref{thm:dp-not-signature-derivable} gives the identification, the projection values are $0$ and $1$, and part~(a) of Theorem~\ref{thm:homomorphic-inexpressibility} applies.
\end{proof}
\section{Record emission forces duplication}\label{sec:arch}

In the many-sorted first-order syntax of Definition~\ref{def:record-emitter}, a right-hand side that contains both a generator-bearing record frame and a generator-bearing recursive site contains two occurrences of the generator. Along the trace, every stage after the first recursive firing contains one occurrence in a frame slot and one at the active site.

\begin{definition}[Record-emitting base/step/counter schema]\label{def:record-emitter}
  A \emph{record-emitting base/step/counter schema} is a first-order rewrite system with:
  \begin{enumerate}[nosep,leftmargin=1.6em]
    \item a counter sort $C$ with constructors $Z \in C$ and $S : C \to C$;
    \item a generator sort $Y$ with one distinguished generator token $y$ in the rule context and no signature symbol producing another $Y$-sorted term;
    \item a carrier sort $T$;
    \item a distinguished active-site symbol $F : T \times Y \times C \to T$;
    \item a frame constructor $G : Y \times T \to T$;
    \item a base rule $F(x, y, Z) \to x$;
    \item a step rule of the form $F(x, y, S(n)) \to \rho(x, y, n)$ for some first-order right-hand-side term $\rho$ whose free variables are among $\{x, y, n\}$.
  \end{enumerate}
\end{definition}

\begin{definition}[Record emission and generator preservation]\label{def:record-emission}
  Let $\mathcal R$ be the step rule of a schema satisfying Definition~\ref{def:record-emitter}. Say $\mathcal R$ \emph{emits a new record frame} if its right-hand side $\rho(x,y,n)$ contains an occurrence of the frame constructor $G$; the left-hand side $F(x,y,S(n))$ contains none, so every such occurrence is new. Say $\mathcal R$ \emph{preserves the generator} if $\rho(x,y,n)$ contains an active-site occurrence headed by $F$.
\end{definition}

\begin{theorem}[Generator duplication]\label{thm:record-emission-necessity}
  In the many-sorted first-order syntax of Definition~\ref{def:record-emitter}, let $\mathcal R: F(x,y,S(n)) \to \rho(x,y,n)$ be a step rule whose right-hand side $\rho(x,y,n)$ contains both:
  \begin{enumerate}[nosep,leftmargin=1.6em]
    \item a $G$-frame occurrence, and
    \item an active-site occurrence headed by $F$.
  \end{enumerate}
  Then $\rho(x,y,n)$ contains at least two distinct occurrences of the generator variable $y$: one inside the first-argument slot of a $G$-frame and one at the generator position of an active-site subterm headed by $F$.
\end{theorem}

\begin{proof}
  By Definition~\ref{def:record-emitter}, every $Y$-sorted position in $\rho(x,y,n)$ contains the distinguished token $y$. A $G$-frame occurrence therefore contributes one occurrence at the frame slot, and an active-site occurrence contributes one at the generator slot of the active-site term. These positions are distinct, so the right-hand side contains at least two distinct occurrences of $y$.
\end{proof}

\begin{corollary}[Record emission forces duplication or generator erasure]\label{cor:no-record-without-duplication}
  Within the same minimal many-sorted syntax, a right-hand side avoids duplicating the generator only by failing at least one of:
  \begin{enumerate}[nosep,leftmargin=1.6em]
    \item \emph{Record emission in the positional syntax.} The right-hand side omits every $G$-frame occurrence.
    \item \emph{Recursive-generator preservation in the positional syntax.} The right-hand side omits every active-site occurrence.
  \end{enumerate}
\end{corollary}

\begin{proof}
  Direct from Theorem~\ref{thm:record-emission-necessity}: any right-hand side containing both a frame occurrence and an active-site occurrence contains two distinct occurrences of the generator variable.
\end{proof}

\begin{corollary}[Uniqueness of the shared generator variable]\label{cor:y-bridge}
  In the syntax of Definition~\ref{def:record-emitter}, the duplicated token $y$ alone carries live-generator continuity and record ownership. The frame slot has sort $Y$, which excludes the base variable $x$ of sort $T$ and the counter variable $n$ of sort $C$, and the signature provides no second $Y$-sorted term. Thus Theorem~\ref{thm:record-emission-necessity} identifies $y$ at both the emitted record and the continuing computation.
\end{corollary}

\begin{proof}
  The frame constructor has type $G:Y\times T\to T$, so its first argument is $Y$-sorted. The active-site constructor has type $F:T\times Y\times C\to T$, so its generator position is also $Y$-sorted. The signature restriction leaves $y$ as the sole inhabitant available at either position.
\end{proof}

\begin{proposition}[Record emission along the trace]\label{prop:record-emission-dynamic}
  In the syntax of Definition~\ref{def:record-emitter}, let $\sigma_i$ be the state after $i$ firings from counter height $k$: $i$ record frames around the active site, whose counter is $k-i$. For $i<k$ one firing of the step rule under the accumulated frames sends $\sigma_i$ to $\sigma_{i+1}$. For $1\le i\le k$ the state $\sigma_i$ contains an emitted frame and preserves the generator, so by Theorem~\ref{thm:record-emission-necessity} it contains the generator at two distinct positions, one in a frame slot and one at the active site. The sequence has at least one firing whenever $k\ge1$.
\end{proposition}
\begin{proof}
  A firing consumes one counter layer and emits one frame, and congruence under the accumulated frames places it at the active site. At every stage $i\ge1$ the outermost symbol of $\sigma_i$ is a frame, and the active site survives under the frames, so the positional theorem applies to $\sigma_i$.
\end{proof}

\begin{remark}[Storage forms for generator retention]\label{rem:two-canonical-forms}
  The retained generator can be stored in either of the following forms:
  \begin{enumerate}[nosep,leftmargin=1.6em]
    \item \textbf{Syntactic retention.} The step rule's right-hand side contains at least one occurrence of $y$ outside the recursive call, as in Rule~2 of the self-duplicating primitive recursor.
    \item \textbf{Externalized trace.} The rewrite relation is extended to act on configurations $\langle t, \pi \rangle$ where $\pi$ is a history list each of whose tokens records the rule and the position applied at the corresponding step together with the bindings required for reversal, including the bindings of erased variables, in the sense of the trace-annotation reversibilization for TRSs of Nishida, Palacios, and Vidal~\cite{nishidaPalaciosVidal2018}, with the categorical framework for reversible computation of Abramsky~\cite{abramsky2005}. The trace tokens must record $y$ whenever the rule's right-hand side omits it.
  \end{enumerate}
  Theorem~\ref{thm:record-emission-necessity} applies to the first form. The second changes the state carrier and is covered by Proposition~\ref{prop:record-emission-beyond}(2). Any faithful retention mechanism therefore stores the generator either in the term or in the enlarged state; the two cases are classified by these results.
\end{remark}

\begin{proposition}[Record emission in three regimes]\label{prop:record-emission-beyond}
  The two-occurrence conclusion of Theorem~\ref{thm:record-emission-necessity} extends to three further regimes.
  \begin{enumerate}[nosep,leftmargin=1.6em]
    \item \emph{Many-sorted first order.} For a family of record-emitting right-hand sides indexed by a finite sort set, every sort whose component emits a new record frame and preserves the recursive generator has two distinct generator positions, one in a frame slot and one at an active site.
    \item \emph{Externalized trace.} For a faithful record emitter, whose decoder recovers the progress index from the emitted observation at every terminal record, distinct progress depths $i\ne j$ below a common bound yield distinct externalized observations.
    \item \emph{Semantic regimes.} For an interpretation of right-hand sides into an observable type together with a regime predicate, under the coherence condition that regime membership entails frame emission and generator preservation, every right-hand side whose interpretation lies in the regime has the two distinct generator positions.
  \end{enumerate}
  A three-way classifier places each record emitter in one of the single-sort, many-sorted, or externalized-trace regimes and in one only.
\end{proposition}

\begin{proof}
  Items (1) and (3) apply Theorem~\ref{thm:record-emission-necessity} to the sort component and to the interpreted right-hand side after the coherence condition supplies its two hypotheses. For (2), equal observations at two depths would decode to equal indices. The classifier follows by case analysis on the three regimes.
\end{proof}

\begin{remark}[Sharing]\label{rem:sharing-below-theorem}
  Theorem~\ref{thm:record-emission-necessity} concerns abstract terms. Term graphs may represent the two $y$-occurrences by one shared node while keeping both references. Kennaway, Klop, Sleep, and de Vries~\cite{kennawayKlopSleepDeVries1994} prove that for orthogonal systems term-graph rewriting simulates term rewriting through the unravelling of graphs into terms, and the recursor is orthogonal. Paolini, Piccolo, and Roversi~\cite{paoliniPiccoloRoversi2016} give an information-preserving simulation of primitive recursive functions.
\end{remark}

\begin{remark}[Kolmogorov upper bound]\label{rem:kolmogorov-record}
  For fixed constructor $G$ and $m\ge2$, the Kolmogorov complexity of the record at step $m$ of the trace, relative to a fixed universal prefix machine and under the standard $O(1)$ additive-constant convention~\cite{liVitanyi2019}, satisfies
  \begin{align*}
    K(\mathrm{record}_m)
    &\le K(x,y,m)+O(1) \\
    &\le K(x)+K(y)+K(m)+O(1),\\
    K(m)
    &\le \log_2 m + 2\log_2\log_2 m + O(1).
  \end{align*}
  and, with $x$, $y$, and $m$ supplied as conditional inputs,
  \[
    K(\mathrm{record}_m \mid x,y,m) = O(1).
  \]
  The upper bound follows from a fixed program that assembles the record from $x$, $y$, and $m$.
\end{remark}

\begin{remark}[Encoding invariance]\label{rem:kolmogorov-invariance}
  Remark~\ref{rem:kolmogorov-record} gives a coding bound distinct from Theorem~\ref{thm:record-emission-necessity}. For any fixed computable encoding, the record is assembled from descriptions of the base, generator, and index; changing the universal prefix machine changes the bound by an additive constant. Term-graph sharing can reduce physical memory while preserving the abstract-term occurrence statement, subject to the simulation of term rewriting by term-graph rewriting for orthogonal systems~\cite{kennawayKlopSleepDeVries1994}.
\end{remark}

\begin{remark}[Scope of the duplication theorem]\label{rem:duplication-is-feature}
  Theorem~\ref{thm:record-emission-necessity} shows that the two occurrences are forced once the many-sorted right-hand side is required to contain both the specified frame slot and the specified active-site slot. Section~\ref{sec:oi} proves that additive whole-term aggregation fails to orient this positional pattern. The theorem's domain is the first-order syntax of Definition~\ref{def:record-emitter}; Remark~\ref{rem:linear-logic-scope} treats carrier changes, sharing, and fixed-point types.
\end{remark}

\begin{remark}[Record emission and certificate order]\label{rem:layer-crossing-bridge}
  Theorem~\ref{thm:record-emission-necessity} supplies the rewriting-side input to the six-slot comparison of Section~\ref{sec:weight}. In the direct grammar, the forced occurrence pattern creates the operational-inexpressibility result of Theorem~\ref{thm:canonical-instance}; the dependency-pair theorem licenses transformation to the recursive-call relation.

  Theorem~\ref{thm:structural-identity} proves that the dependency-pair profile and a compiled G\"odel profile for Robinson arithmetic both realize the six listed stages. Its dependency-pair component is the natural-number descent of Theorem~\ref{thm:ag-rca}.

  Corollary~\ref{cor:no-record-without-duplication} locates the direct failure at the minimal positional syntax requiring both frame and active-site occurrences.
\end{remark}

\subsection{The \texorpdfstring{$r$}{r}-ary duplicator family}\label{subsec:rary-duplicator}

The unary wrapper is one member of a family. Replace the recursive rule by
\[
  F(x,y,S(n))\to G([y,\ldots,y]_r,F(x,y,n)),
\]
where each new frame contains a list of $r$ payload copies.

\begin{proposition}[Duplication-order scaling laws]\label{prop:rary-duplicator-laws}
  For the $r$-ary family, every maximal derivation still has length $k+1$, whatever the value of $r$ and the payload identity. At live stage $i$,
  \[
    \#_b(t_i)=ri+1.
  \]
  If $\beta=|b|$ and $\mathrm{Con}_r$ sums payload mass over all live states, then
  \[
    2\mathrm{Con}_r(k,\beta)
    =\beta\bigl(rk(k+1)+2(k+1)\bigr),
  \]
  with the envelopes
  \[
    r\beta k^2\le2\mathrm{Con}_r(k,\beta)
    \le r\beta(k+1)^2+2(k+1)\beta.
  \]
  Thus the omitted payload mass grows in proportion to the duplication order, while the runtime and the projected counter are independent of $r$. In the $r$-ary positional syntax, $m$ emitted frames together with one preserved active site contain $rm+1$ generator occurrences, which is $m+1$ for $r=1$.

  For every $r\in\mathbb N$, including $r=0$, the live-orbit action, frame mass, and payload mass equal their finite stage sums; the action has its closed form; the mass partition equality holds; the crossover bound, crossover minimality, and dominance hold whenever $rw>0$; the inefficiency coefficient is unbounded whenever $rw>0$; and the payload-count and projection-bit laws hold at every stage. At $r=1$ these laws equal the binary laws term by term. For $r>0$, the frame quantities at wrapper weight $w$ equal the binary frame quantities at weight $rw$, while the live payload mass satisfies the correction identity
  \[
    \mathrm{Con}_1(k,rw)=\mathrm{Con}_r(k,w)+(r-1)(k+1)w.
  \]
  At $r=0$ the frame mass is zero, and the positivity-dependent clauses hold through their explicit premises.
\end{proposition}

\begin{proof}
  The orbit proof is the unary induction with a replicated frame list. Each step adds $r$ payload leaves and consumes one counter constructor. Summing $ri+1$ over $0\le i\le k$ gives the displayed mass formula, and the same finite-sum calculation gives the live-orbit equalities, action formula, and partition. Euclidean division by $rw+1$ gives the crossover bounds and minimality; the lower bound yields dominance. The quadratic numerator and logarithmic denominator give unbounded inefficiency when $rw>0$. The payload-count induction gives the payload-count law. Substitution $r=1$ gives the binary laws. For $r>0$, expanding both live payload sums gives the correction term $(r-1)(k+1)w$. Substitution $r=0$ gives zero frame mass and discharges the remaining clauses through their stated premises. The $rm+1$ occurrence count follows by induction; the unary helper is its $r=1$ specialization.
\end{proof}

\begin{theorem}[Nonuniform information spent by the $r$-ary role license]
  \label{thm:rary-license-information}
  Fix $r\ge1$.  Give the active occurrence probability $a\in[0,1]$ and let
  $\phi$ be any nonnegative normalized probability law on the $r$ frame
  indices.  The induced role law is
  \[
    p(\mathrm{active})=a,
    \qquad p(\mathrm{frame}_i)=(1-a)\phi_i.
  \]
  The dependency-pair discriminator on the $r$-ary contractum is
  the partition ``active versus frame''.  Writing entropy in bits,
  \[
    H(\mathrm{role})=h_2(a)+(1-a)H(\phi).
  \]
  The partition supplies $h_2(a)$ bits and leaves
  $(1-a)H(\phi)$ bits of frame-index identity unresolved; the role itself
  leaves zero.  At $a=1$ the residual vanishes; at $a=0$ the spent active/frame
  information is zero and the residual is the entire frame entropy.

  Under the uniform law on all $r+1$ roles,
  $a=1/(r+1)$ and $\phi$ is uniform on the frames, so this specializes to
  \[
    h_2\!\left(\frac1{r+1}\right)
    =\log_2(r+1)-\frac{r}{r+1}\log_2 r
  \]
  spent bits and $\frac{r}{r+1}\log_2 r$ residual bits.  For $r=1$, every
  normalized one-frame law has zero frame entropy, so the active/frame
  partition resolves all role uncertainty for every active mass.  The $r=0$
  boundary is stated on its own: with the single active role at mass one the
  role entropy is zero, without requiring a normalized law on the empty frame
  type.
\end{theorem}

\begin{proof}
  Split the finite role carrier into the active cell and the frame cell.  The
  active conditional law is a point mass.  Conditional on the frame cell, the
  law is $\phi$ whenever that cell has positive mass; zero-mass defaults
  do not change the joint law.  The finite entropy chain rule therefore gives
  $H(\mathrm{role})=h_2(a)+(1-a)H(\phi)$, while entropy of the partition itself
  is $h_2(a)$.  The occurrence equivalence proves that the
  $r$-ary dependency-pair channel is this active/frame partition.  Substituting
  the uniform active mass and the uniform frame law gives the uniform
  formula.  A
  one-element frame carrier has zero entropy under every normalized law; the
  zero-frame boundary is a point mass and is proved without an empty-carrier
  normalization premise.
\end{proof}

\begin{proposition}[Role exchange at countable support]\label{prop:countable-role-exchange}
  For laws $\nu$ and $\rho$ on a countable role alphabet $C$, put
  \[
    D(\nu\,\|\,\rho)=\sum_{c\in C}\nu(c)\ln\frac{\nu(c)}{\rho(c)},
  \]
  a sum over the alphabet with natural logarithms. On a finite alphabet the sum is the finite sum, and against the uniform law it equals $\ln|C|-H(\nu)$, the residual of the role exchange under the uniform reference; it vanishes at the uniform law and only there. If $\nu\ge0$, $\rho>0$, both laws sum to the same total, and the series defining $D(\nu\,\|\,\rho)$ is summable, then $D(\nu\,\|\,\rho)\ge0$; and $D(\rho\,\|\,\rho)=0$ for every $\rho$.
\end{proposition}
\begin{proof}
  The finite case is the definition. Against the uniform law the logarithm splits into $\ln\nu(c)+\ln|C|$, which gives $\ln|C|-H(\nu)$, and the finite role-exchange theorem identifies the zero set. Termwise, $-\nu(c)\ln(\nu(c)/\rho(c))\le\rho(c)-\nu(c)$ by $\ln u\le u-1$; summing under the summability hypothesis and equal totals gives the sign. Each summand of $D(\rho\,\|\,\rho)$ is zero.
\end{proof}

The summability hypothesis is part of the statement: on a countable alphabet the divergence series can diverge, and the finite theorem is the countable one at the uniform reference.

\section{Dependency-pair extraction and occurrence roles}\label{sec:license}

We first use only the free two-rule relation $R_0$. Its transformed-call
relation $D_0$ has the rule $F(b,s,S(n))\to F(b,s,n)$.
The Arts and Giesl theorem~\cite{artsgiesl2000} supplies the general
dependency-pair soundness theorem; Theorem~\ref{thm:free-context-normalization}
also proves termination of the free source.

\subsection{Root extraction}

\begin{theorem}[Free extraction and all root extensions]\label{thm:free-root-extraction}
  Let $E(a,u)$ return $F(b,s,n)$ if and only if
  $a=F(b,s,S(n))$ and $u=G(s,F(b,s,n))$, and return no call otherwise.
  Then
  \[
    E(a,u)=\mathsf{some}(v)\quad\Longleftrightarrow\quad
    a\,R_0\,u\ \land\ a\,D_0\,v.
  \]
  The relation $D_0$ is deterministic, terminating, and confluent. Its
  execution from $F(b,s,n)$ has length $\operatorname{sp}(n)$ and ends at
  $F(b,s,\operatorname{strip}(n))$, where $\operatorname{strip}$ removes
  the initial successor prefix. For any added root relation $A$, extraction
  remains complete for the same $D_0$ on $R_0\cup A$ if and only if no
  edge in $A\setminus R_0$ starts at a state with a $D_0$ call.
\end{theorem}
\begin{proof}
  The two root-rule cases prove the extraction equivalence.
  The prefix length drops by one on each call; induction gives the length
  and endpoint. For an extension, completeness forces an edge with a call
  to have the successor rule's output, hence already to lie in $R_0$.
  In the other direction, a new edge with no call adds no extraction requirement.
  The criterion concerns preservation of this extractor.
\end{proof}

The free relation supplies the source, the output, and the extracted call: the extractor returns the recursive callee if and only if the step is a successor root step. The successor-prefix count ranks the residual relation $D_0$.

\subsection{The duplicated generator as two occurrences of one value}

In $u=G(s,F(b,s,n))$, using zero-based child indices, the generator
occurs at paths $[0]$ and $[1,1]$. Both positions contain $s$.
Their containing branches, at $[0]$ and $[1]$, are $s$ and $F(b,s,n)$.
The \emph{role channel} compares the containing branch with the
extracted callee. It compares neither generator occurrence with itself.

\begin{theorem}[Role selection from the free output]\label{thm:free-role-selection}
  For all free terms $b,s,n$, the computed channel is false on the frame
  and true on the active branch. For any type $C$, a channel
  $c:\{\mathrm{frame},\mathrm{active}\}\to C$ decodes this decision if
  and only if $c(\mathrm{frame})\neq c(\mathrm{active})$.
  No channel factoring through the common generator value decodes it.
  For any probability distribution $\nu$ on the two roles, the computed
  decision has entropy $H(\nu)$, posterior entropy zero, and gain
  $H(\nu)$. On the binary record fork, its overproduction gap is
  $1-H(\nu)$; the uniform distribution gives gain one and gap zero.
\end{theorem}
\begin{proof}
  The callee contains $s$ as a proper subterm, so $s\neq F(b,s,n)$.
  The branch comparison therefore separates the roles. A decoder exists
  if and only if their channel values differ. For a role $r$ and decision
  $d$, the joint mass is $\nu(r)$ if the computed decision of $r$ is $d$,
  and zero otherwise. Each decision cell is a singleton. Hence joint mass
  equals decision mass times the point-mass posterior, including zero-mass
  cells; division gives Bayes' formula at positive-mass cells.
  This calculation concerns the record fork; $R_0$ itself is confluent.
\end{proof}

Theorem~\ref{thm:free-role-selection} therefore separates occurrences: the two positions hold the same generator value, and the computed branch decision distinguishes their roles, so it does not factor through that value. Free dependency-pair extraction uses this distinction.

\subsection{The one-bit role separator}

For a finite one-step relation $R$ and a source $x$, let $H_0^R(x)=\log_2\mu_R(x)$ be the Hartley entropy of the terminal support, the number of normal forms reachable from $x$. For an evidence channel from a target $X$ through an observable $W$ to answers, let the licensed gain be the conditional mutual information $I(X;C\mid W)$ in bits. Define the \emph{overproduction gap} by $\Omega=H_0^R(x)-I(X;C\mid W)$ and a \emph{boundary event} by $\Omega>0$~\cite{rahnamaDistinction}. On the three-state fork $\mathsf{Fork}_3$ with two terminal outputs, $H_0=1$.

\begin{theorem}[One role bit closes the support gap]\label{thm:one-bit}
  On the fork $\mathsf{Fork}_3$ over the two-point role alphabet, every channel constant across the roles has $\Omega=1$. The role channel has licensed gain one bit and $\Omega=0$. A channel factoring through the value coordinate is constant because the two occurrences have the same value. Thus the one-bit role separator closes the one-bit terminal-support gap in this calibrated fork model.
\end{theorem}

\begin{proof}
  The fork has two terminals, so $H_0^R=1$. A channel constant on the two roles has licensed gain zero, hence $\Omega=1$. The role channel takes two values and has licensed gain one, hence $\Omega=0$. The two occurrences share the value coordinate, so a channel that factors through that coordinate is constant on the roles.
\end{proof}

Theorem~\ref{thm:rary-license-information} takes the same active-versus-frame partition as its $r$-ary discriminator, and Proposition~\ref{prop:role-split-law} gives the separator in general form.

\begin{proposition}[Role-split law]\label{prop:role-split-law}
  Let $v:D\to V$ be a value map and $p,q\in D$. One of three cases holds: $p=q$; $v(p)=v(q)$ and $p\neq q$, a \emph{role split}; or $v(p)\neq v(q)$. In a role split every map that factors through $v$ takes the same value at $p$ and $q$, and some Boolean map separates them; the two conditions together characterize the role split, and every map that separates $p$ from $q$ fails to factor through $v$. On the two occurrences of one payload term $s$ in the successor output $G(s,F(b,s,n))$, the value projection is onto the payload terms, the frame occurrence and the active occurrence form a role split, and the activity predicate separates them and fails to factor through the value.
\end{proposition}
\begin{proof}
  The trichotomy is a case split on $p=q$ and on $v(p)=v(q)$. A map $g\circ v$ takes equal values on a $v$-identified pair; the indicator of $p$ separates $p$ from $q$; and a separator of a $v$-identified pair cannot be of the form $g\circ v$. In the other direction, refusal by every value-factored map applied to $v$ itself gives $v(p)=v(q)$, and a separator gives $p\neq q$. On the recursor instance the two occurrences carry the same term $s$ and differ in role, every payload term occurs, and the activity predicate takes different values at the two occurrences.
\end{proof}

\begin{proposition}[Scheduled and evidence-relative forms of the gap]\label{prop:gap-forms}
  Let $R$ be a finite relation that is locally normalizing at $x$, and let $\sigma$ be a probability law on the normal forms reachable from $x$, a \emph{scheduler}. Put $\Omega_\sigma=H(\sigma)-I(X;C\mid W)$ in bits. Then $\Omega_\sigma\le\Omega$, with equality for the uniform scheduler. If $H(\sigma)$ equals the conditional entropy of the target given the observer, then $\Omega_\sigma$ is the residual conditional entropy of the target given observer and channel, in bits, and it vanishes when that residual vanishes. On $\mathsf{Fork}_3$ the uniform scheduler gives $\Omega_\sigma=1$ against the constant channel and $\Omega_\sigma=0$ against the role channel. For a fixed relation and source, replacing one evidence channel by another changes $\Omega$ by the difference of licensed gains: $\Omega_1=\Omega_0-(I_1-I_0)$, so $\Omega$ is antitone in the licensed gain, a strict increase of the gain lowers $\Omega$, and the passage from the constant channel to the role channel on $\mathsf{Fork}_3$ lowers $\Omega$ by one bit.
\end{proposition}
\begin{proof}
  A law on $\mu_R(x)$ points has entropy at most $\log_2\mu_R(x)=H_0^R(x)$, with equality for the uniform law, and $\Omega$ subtracts the same licensed gain. Under the alignment hypothesis the two direct entropies cancel and the residual remains. The fork values are computed, and the chain rule is the definition of $\Omega$ read at two channels.
\end{proof}

\subsection{Certificate language and decidability}

The recursor terminates and is confluent, with a computable normalizer (Theorem~\ref{thm:free-context-normalization}), while sound termination certificates first appear after one change of representation away from the direct whole-term language (Proposition~\ref{prop:boundary-kappa} and Theorem~\ref{thm:free-witness-depths}). Operational inexpressibility therefore concerns the certificate language. Budget exhaustion is relative to a chosen budget; Proposition~\ref{prop:confession-crossing} gives a finite crossing point for the recursor.

\subsection{Mass profiles and projection}

Equality between the mass profile of $F(a,b,S^k(0))$ and a self-embedding orbit depends on the constructor costs. Under $c_G=2c_{\delta}$ the profiles agree; under unit costs they agree only at index zero. Counter projection separates the terminating recursor from the self-embedding relation.

\subsection{Structural identity under direct measure}\label{subsec:structural-identity-under-direct-measure}

A direct whole-term measure assigns a nonnegative weight to visible term structure and asks whether the weight strictly decreases across each rule firing. A circular reference is a rule shape whose every channel coordinate stays equal or rises, the schematic case being $A\to B$, $B\to A$, where the rewrite generates a cycle whose weight remains constant or grows under any constructor-additive interpretation. The step-duplicating recursor has a strictly decreasing counter $S(n)\to n$ on the right-hand side, but the recursive call is wrapped in $G(y,\cdot)$, so the wrapped payload accumulates one copy per step along the trace. By Proposition~\ref{prop:trace-law}, the payload count at step $i$ is $i+1$ and the counter is $k-i$.

\begin{proposition}[Calibrated mass identity on the recursor execution]\label{prop:structural-identity-under-direct-measure}
  Fix a weighted constructor-cost algebra $D$ with natural costs $c_{\delta},c_{G},c_F$ and equations
  \[
    D.\mu(\mathrm{delta}\,t)=D.\mu(t)+c_{\delta},\quad
    D.\mu(\mathrm{recDelta}\,b\,s\,u)=D.\mu(u)+c_F,\quad
    D.\mu(\mathrm{app}\,s\,t)=D.\mu(t)+c_G.
  \]
  For the finite execution from $\mathrm{recDelta}\,b\,s\,(\mathrm{delta}^k\,\mathrm{void})$, compare the orbit generated from the same state by $t\mapsto\mathrm{delta}\,t$. If $c_G=2c_{\delta}$, their masses agree at every displayed index. In the other direction, when $k>0$, equality of the displayed profiles implies $c_G=2c_{\delta}$. The costs $(c_{\delta},c_G,c_F)=(1,2,1)$ satisfy the equation. With unit costs $(1,1,1)$, equality fails at every positive index.
\end{proposition}

\begin{proof}
  At stage $i$, the recursor execution has mass
  \[
    D.\mu(\mathrm{void})+(k-i)c_{\delta}+c_F+i c_G,
  \]
  while the self-embedding orbit from the shared initial state has mass
  \[
    D.\mu(\mathrm{void})+k c_{\delta}+c_F+i c_{\delta}.
  \]
  The displayed equality follows from $c_G=2c_{\delta}$. At positive $i$, cancellation gives the converse.
\end{proof}

The statement is limited to costs satisfying $c_G=2c_{\delta}$. Unit costs refute positive-index equality with the self-embedding orbit. A nonterminating two-cycle has the same constant unit-cost profile as the recursor execution; this observation equivalence concerns the two-cycle relation. The direct-measure barriers of~\cite{rahnamaOrientation} concern orientation instead of equality of arbitrary mass profiles.

\begin{proposition}[Whole-term indistinguishability and counter projection]
  \label{prop:whole-term-indistinguishability}
  Under the factor-two constructor-cost calibration, every observer factoring through the finite mass profile assigns the fixed-input recursor execution and the same-start self-embedding prefix the same value. Under unit costs, the corresponding observer theorem compares the execution with the nonterminating two-cycle. Counter projection orients the transformed call.
\end{proposition}

\begin{proof}
  Under the factor-two equation, Proposition~\ref{prop:structural-identity-under-direct-measure} gives equal profiles, so every factor through that profile gives equal observer values. The same argument applies to the unit-cost profile shared with the two-cycle. Counter projection retains $k-i$ and omits the payload carrier.
\end{proof}

\subsection{Dependency pairs as coordinate projection}\label{subsec:dp-as-coordinate-projection}

The Arts and Giesl soundness theorem~\cite{artsgiesl2000} justifies replacing the wrapped term by its decreasing counter coordinate. The direct whole-term language derives no such projection from the rewrite signature. This separates the terminating recursor from the circular-reference mass profile: the recursor has a decreasing counter coordinate to project to, and the circular reference has none.

This distinction separates two classes of termination methods. Construction methods (polynomial interpretations, path orders) add a global comparison object to the proof. Confession methods (dependency pairs, counter projection, size-change termination, argument filtering) prove a smaller residual problem after a justified projection removes a dimension at which direct whole-term measures are operationally inexpressible. The four confession methods share one projection rank (Remark~\ref{rem:W2-mechanized}).

\subsection{Two family comparisons}\label{subsec:closure}

The preceding proposition concerns a fixed-input execution. A relation
comparison can instead preserve its progress index.

\begin{theorem}[Progress-preserving quotient]\label{thm:progress-quotient}
  Let $o:I\to X$ and $p:I\to Y$ be injective, and let $J$ be any relation
  on $I$. Give each image the relation induced by $J$ on indices.
  The map $o(i)\mapsto p(i)$ is a relation isomorphism. On the disjoint
  union of the images, identify two states if and only if their indices
  agree. The quotient, with edges induced from either image, is relation
  isomorphic to $(I,J)$. It has at least two points when $I$ has two
  distinct elements.
\end{theorem}
\begin{proof}
  Injectivity defines inverse index maps on both images. The common index
  map descends to the quotient and has inverse $i\mapsto[o(i)]$.
  An edge between classes has representatives related by $J$ if and only
  if their indices are related by $J$.
\end{proof}

\begin{theorem}[Recursive calls reverse wrapper expansion]\label{thm:progress-dp-direction}
  Put $r_n=F(b,s,S^n(0))$ and $v_n=G^n(p,a)$ for arbitrary free terms
  $b,s,p,a$. The relation $D_0$ satisfies $r_m\,D_0\,r_n$ if
  and only if $m=n+1$. Hence $r_n\mapsto v_n$ identifies
  \emph{reversed} calls with wrapper expansion $t\mapsto G(p,t)$.
  For every $N$, the map $r_{N-i}\mapsto v_i$, $0\leq i\leq N$,
  identifies forward calls with expansion restricted to that finite
  prefix. No map, even a noninjective one, preserves every wrapper
  expansion as a forward $D_0$ step.
\end{theorem}
\begin{proof}
  Constructor injectivity proves the two index equations and the
  isomorphisms. A map preserving all expansion steps would map the
  infinite expansion sequence to an infinite $D_0$ sequence, contradicting
  strict descent of its natural rank. The finite prefix omits the outgoing
  expansion edge at $v_N$.
\end{proof}

A second comparison uses a one-point quotient of orbit
functions. It records linear-growth agreement and is independent of the
relation isomorphisms above.

\begin{definition}[Licensed quotient on the orbit-function space]\label{def:lq-on-orbits}
  A \emph{licensed quotient}, abbreviated $\mathrm{LQ}$, on $\mathbb N\to\mathrm{Trace}$ is a tuple $(G,\,\mathrm{action},\,Q,\,\mathrm{proj},\,\mathrm{license})$ consisting of a group $G$ acting on $\mathbb N\to\mathrm{Trace}$, a quotient carrier $Q$, a projection $\mathrm{proj}\colon(\mathbb N\to\mathrm{Trace})\to Q$ satisfying $\mathrm{proj}(g\cdot f)=\mathrm{proj}(f)$, and a license object. The quotient $\mathrm{LQ}_\mathrm{TRS}$ takes $G$ trivial, $Q=\{\ast\}$, and $\mathrm{proj}(f)=\ast$. It therefore has one equivalence class.
\end{definition}

\begin{definition}[Licensed-quotient linear-growth agreement]\label{def:trs-equiv-lq}
  Fix a licensed quotient $\mathrm{LQ}$ on $\mathbb N\to\mathrm{Trace}$ and a direct measure system $D$. Two orbit functions $o_1,o_2\colon \mathbb N\to\mathrm{Trace}$ satisfy \emph{licensed-quotient linear-growth agreement} (written $o_1\sim_{\mathrm{LQ},D} o_2$) if:
  \begin{enumerate}[nosep,leftmargin=1.4em]
    \item $\mathrm{LQ}.\mathrm{proj}(o_1)=\mathrm{LQ}.\mathrm{proj}(o_2)$;
    \item $n\mapsto D.\mu(o_1\,n)$ satisfies the linear-growth predicate of Proposition~\ref{prop:structural-identity-under-direct-measure};
    \item $n\mapsto D.\mu(o_2\,n)$ satisfies the same linear-growth predicate.
  \end{enumerate}
  The first conjunct is projection equality; the other two state that both profiles belong to the linear-growth class.
\end{definition}

\begin{theorem}[Index correspondence under the one-point quotient]\label{thm:trs-iso-mod-lq}
  Let $b,s,A,B\in\mathrm{Trace}$ and let $D$ be a direct measure proof system whose interpretation satisfies the standard constructor-cost equations
  \begin{align*}
    D.\mu(\mathrm{delta}\,t) &= D.\mu(t)+1,\\
    D.\mu(\mathrm{recDelta}\,b'\,s'\,u) &= D.\mu(u)+1,\\
    D.\mu(\mathrm{merge}\,x\,y) &= D.\mu(x)+D.\mu(y)+1.
  \end{align*}
  Then the counter-indexed recursor input family and the circular-reference family satisfy the three-field agreement predicate under the one-point licensed quotient relative to $D$:
  \[
    \mathrm{CounterIndexedRecursorFamily}(b,s)\;\sim_{\mathrm{LQ}_\mathrm{TRS},D}\;\mathrm{CircularReferenceOrbit}(A,B),
  \]
  and the two indexed systems admit the index-preserving bijection of Proposition~\ref{prop:orbit-iso}. The bijection concerns the two family images and leaves the underlying rewrite relations unchanged.
\end{theorem}

\begin{proof}
  Every orbit maps to $\ast$, which proves the first conjunct. The mass formulas for the two families are affine in the index, which proves the other two. Proposition~\ref{prop:orbit-iso} supplies the bijection. The quotient identifies every pair of functions, so the first conjunct holds for all orbit functions.
\end{proof}

\begin{proposition}[Orbit-system isomorphism]\label{prop:orbit-iso}
  Both families are injective in their index: $\mathrm{CounterIndexedRecursorFamily}(b,s,m)=\mathrm{CounterIndexedRecursorFamily}(b,s,n)$ implies $m=n$, and likewise for $\mathrm{CircularReferenceOrbit}(A,B)$. Hence the index-preserving state map
  \[
    \varphi\colon \mathrm{CounterIndexedRecursorFamily}(b,s,n)\;\longmapsto\;\mathrm{CircularReferenceOrbit}(A,B,n)
  \]
  is a bijection between the two family images, commutes with their index successors, and has the index-preserving map in the opposite direction as inverse. It leaves the underlying rewrite relations unchanged.
\end{proposition}

\begin{proof}
  Injectivity of the recursor orbit reduces, by injectivity of the $\mathrm{recDelta}$ constructor, to injectivity of the counter trace $n\mapsto\mathrm{delta}^n(\mathrm{void})$, which follows by induction with constructor discrimination. Injectivity of the circular-reference orbit follows from strict growth of structural size along the merge chain: each step adds $|A|+1>0$, so distinct indices give terms of distinct sizes. With both index maps injective, define $\varphi$ on orbit states by transporting the index; index injectivity is what makes this well defined. If $x=o_1(n)$ and $y=o_1(n+1)$ are consecutive recursor-orbit states, then $\varphi(x)=o_2(n)$ and $\varphi(y)=o_2(n+1)$ are consecutive circular-orbit states, so $\varphi$ commutes with the orbit step; the symmetric argument applies to the inverse transport, and the two composites are the identity on orbit states by the same index computation.
\end{proof}

\begin{definition}[Linear-growth-invariant functional]\label{def:entropy-measure}
  A \emph{linear-growth-invariant functional} on orbit-mass profiles is a pair $(E.\mathrm{discardedInfo},\,E.\mathrm{respects\_linear\_growth})$ where:
  \begin{itemize}[nosep,leftmargin=1.4em]
    \item $E.\mathrm{discardedInfo}\colon(\mathbb N\to\mathbb N)\to\mathbb N$ is the functional applied to a mass profile;
    \item $E.\mathrm{respects\_linear\_growth}$ records that for any two mass profiles $f,g\colon\mathbb N\to\mathbb N$ satisfying the linear-growth predicate, $E.\mathrm{discardedInfo}(f)=E.\mathrm{discardedInfo}(g)$.
  \end{itemize}
  The value assigned to an orbit $o\colon\mathbb N\to\mathrm{Trace}$ is $\mathrm{DI}_{E,D}(o):=E.\mathrm{discardedInfo}(n\mapsto D.\mu(o\,n))$.
\end{definition}

The invariance field requires equal values on every two profiles in the linear-growth class. Since both families lie in that class, their values agree.

\begin{definition}[Agreement under stipulated linear-growth invariance]\label{def:info-equiv-mod-dp}
  Two orbits $o_1,o_2\colon\mathbb N\to\mathrm{Trace}$ \emph{agree under stipulated linear-growth invariance} if, for every $E$ satisfying Definition~\ref{def:entropy-measure}, every direct measure system $D$ satisfying the constructor-cost equations, and every pair of linear-growth certificates, $\mathrm{DI}_{E,D}(o_1)=\mathrm{DI}_{E,D}(o_2)$. Its data consist of mass profiles and the stated invariance property.
\end{definition}

\begin{theorem}[Functional agreement under linear-growth invariance]\label{thm:info-equiv-mod-dp}
  Let $b,s,A,B\in\mathrm{Trace}$. The counter-indexed recursor input family and the circular-reference family agree in the restricted sense of Definition~\ref{def:info-equiv-mod-dp}: the quantified functional is constant on the shared linear-growth class by hypothesis.
\end{theorem}

\begin{proof}
  Apply $E.\mathrm{respects\_linear\_growth}$ to the two profile certificates and unfold $\mathrm{DI}_{E,D}$. The factor-two theorem concerns different orbit objects; the invariance hypothesis remains independent.
\end{proof}

\subsection{Scope of the family comparisons}\label{subsec:joint-closure}

Theorem~\ref{thm:trs-iso-mod-lq} combines equality under a constant projection, membership in a linear-growth class, and a bijection between two indexed family images. Theorem~\ref{thm:info-equiv-mod-dp} applies a functional assumed to be constant on that class. Both claims concern projected families; the underlying rewrite systems remain unchanged.

Proposition~\ref{prop:structural-identity-under-direct-measure} gives mass equality under the factor-two calibration and refutes unit-cost equality against the same-start self-embedding orbit. The function-space results give an index-preserving correspondence between the two families under a constant quotient and equality under an invariance field. By the signature theorem, the counter-reading function lies outside the class of homomorphisms constant in the counter slot.
\section{Costs and certificates}\label{sec:price}

This section uses three quantities. The structural ratio $H_{\mathrm{proof}}$ is defined in Proposition~\ref{prop:proof-entropy}. The Shannon entropy $H_{\mathrm{pos}}$ belongs to the coding model of Definition~\ref{def:gauge-entropy}. Prefix-free Kolmogorov complexity $K$ is used only for the upper bound in Remark~\ref{rem:kolmogorov-record}~\cite{liVitanyi2019}.

\begin{remark}[Counter and payload growth]\label{rem:inverse-growth}
  At step~$i$ of the trace from $F(a,b,S^k(0))$:
  \begin{itemize}[nosep]
    \item the counter has dropped from $k$ to $k-i$ (linear decrease);
    \item the number of distinct payload slots occupied by~$b$ has grown from $1$ to $i+1$ (linear increase in the number of wrapper-context positions containing~$b$).
  \end{itemize}
  The counter and the payload multiplicity therefore move in opposite directions. In an arbitrary substitution where $b$ is itself a large term of measure~$m$, the total whole-term size contributed by the payload copies after $i$ steps is at least $(i+1)\cdot m$, while the successor-side contribution has dropped by~$i$. For any additive or affine measure, the payload contribution dominates for large enough~$m$.

  Dependency-pair projection omits this growing payload dimension from the termination comparison. The omitted component is unbounded in trace depth and payload size. Soundness of the transformation, instead of decrease of that component, licenses the termination conclusion.
\end{remark}

\begin{proposition}[Quadratic growth of the omitted payload mass]\label{prop:confession-dominance}
  Fix the primitive recursion duplicator with ground terms $a,b$ and counter $S^k(0)$. Along the trace, define:
  \begin{itemize}[nosep]
    \item the \emph{residual proof work} $\mathrm{Res}(k)$: the number of strict subterm comparisons performed by the dependency-pair descent across all $k$ recursive steps, so $\mathrm{Res}(k)=k$;
    \item the \emph{omitted payload mass} $\mathrm{Con}(k,b)$: the cumulative size of the payload omitted from the residual comparison across all $k$ steps, counting the $i+1$ payload-bearing positions at stage $i$ (the completed-wrapper count is $i$, and Proposition~\ref{prop:exact-mass-partition} separates the two),
      \[
        \mathrm{Con}(k,b)=\sum_{i=0}^{k}(i+1)|b|=\frac{(k+1)(k+2)}{2}|b|.
      \]
  \end{itemize}
  Then
  \[
    \frac{\mathrm{Con}(k,b)}{\mathrm{Res}(k)}=\frac{(k+1)(k+2)}{2k}|b|\sim \frac{k|b|}{2}
    \qquad\text{as }k\to\infty.
  \]
  The omitted mass grows without bound relative to the residual work: for fixed $|b|\ge 1$, the ratio exceeds every constant.
\end{proposition}

\begin{proof}
  The residual proof work is $k$ by direct count. The omitted mass at step~$i$ is $(i+1)|b|$ by Proposition~\ref{prop:trace-law}. Summing from $i=0$ to $k$ gives the displayed formula.
\end{proof}

Charge the transformed proof a license price $L$ once and the residual cost thereafter. This linear cost crosses the quadratic cost of the direct representation at a computable depth.

\begin{proposition}[Weak crossing, strict crossing, and ties]\label{prop:confession-crossing}
  Write the doubled direct cost and the doubled transformed cost as
  \[
    C_k=(k+1)(k+2)|b|,
    \qquad E_k=2L+2k.
  \]
  For $|b|>0$, let $k_{\mathrm w}$ be the least depth with $E_k\le C_k$ and
  let $k_{\mathrm s}$ be the least depth with $E_k<C_k$.  Then both thresholds
  exist and
  \[
    k_{\mathrm w}\le k_{\mathrm s}\le k_{\mathrm w}+1,
    \qquad k_{\mathrm s}\le L.
  \]
  For $k<k_{\mathrm w}$, the direct cost is smaller: $C_k<E_k$.  At
  $k_{\mathrm w}$ either $C_{k_{\mathrm w}}>E_{k_{\mathrm w}}$, in which case
  $k_{\mathrm s}=k_{\mathrm w}$, or the two costs tie, in which case
  $k_{\mathrm s}=k_{\mathrm w}+1$.  A positive payload admits at most one tie.
  For every $k\ge k_{\mathrm s}$, the transformed cost is smaller, and the strict
  inequality persists because the signed difference $C_k-E_k$ increases strictly
  with depth.

  For $|b|=0$ there is no strict crossing.  Equality occurs at
  $L=k=0$ and nowhere else; at every other zero-payload point the direct cost is smaller.
  The bounded searches return the weak and strict
  thresholds for positive payload and return the stated zero-payload outcomes.
\end{proposition}

\begin{proof}
  The finite searches and leastness statements are the weak- and strict-crossing
  theorems of the cost model.  Positive payload makes the signed doubled-cost
  difference strictly increasing.  The weak threshold cannot occur after the
  strict threshold, and strict increase makes the strict threshold either the
  weak threshold itself or its immediate successor; equality at the weak
  threshold is equivalent to the latter case.  The strict-crossing
  bound of the cost model gives $k_{\mathrm s}\le L$.  Strict monotonicity also proves uniqueness
  of a tie and persistence after $k_{\mathrm s}$.  With zero payload, $C_k=0$
  and $E_k=2L+2k$, so equality is equivalent to $L=k=0$ and strict crossing is
  impossible.
\end{proof}

\begin{proposition}[Single-trace and repeated-prefix workloads]\label{prop:workload-semantics}
  Let $p=|b|$ and count in doubled cost units.  If each depth is processed once
  through horizon $h$, the total direct charge is
  \[
    C_{\mathrm{trace}}(p,h)=(h+1)(h+2)p.
  \]
  If instead every prefix from depth zero through $h$ is recomputed from its
  start, then
  \[
    3C_{\mathrm{prefix}}(p,h)
    =(h+1)(h+2)(h+3)p.
  \]
  For $p>0$ and horizon $h+1$, repeated-prefix recomputation costs strictly more
  than processing the final trace once.  For $p=0$ both workloads have zero
  cost.  Thus the cubic law concerns repeated prefix recomputation.
\end{proposition}

\begin{proof}
  The one-trace quantity is the sum of the stage charges
  $2(j+1)p$ for $0\le j\le h$ and equals $C_h$ of
  Proposition~\ref{prop:confession-crossing}.  Summing that one-trace cost again over every prefix gives the second
  identity.  Both closed forms follow by finite-sum induction.  For positive payload, the repeated-prefix recurrence is the
  new final-prefix cost plus a positive earlier-prefix sum, yielding the strict
  comparison.  Substitution $p=0$ gives the boundary case.
\end{proof}

\begin{proposition}[Budget threshold]\label{prop:budget-thresholds}
  For positive payload, call a budget $B$ \emph{strictly derived} when
  $B\ge k_{\mathrm s}$.  At every permitted depth above this threshold,
  the direct cost exceeds the transformed cost, and $k_{\mathrm s}$ is
  the least threshold with that property.  The weak threshold
  $k_{\mathrm w}$ records the first depth at which the direct cost reaches the transformed cost;
  it is one less than $k_{\mathrm s}$ if and only if the costs tie there.  For zero payload there is no
  strict crossing; the only equality case is $L=k=0$.
\end{proposition}
\begin{proof}The strict threshold is the least depth above the tie, and the weak threshold records it one below when the costs tie there; for zero payload there is no strict threshold.
\end{proof}

Both thresholds depend on the stated cost model and on the license price.

\begin{remark}[Marginal cost and stopping]\label{rem:marginal-regimes}
  The accumulated curves have marginal increments
  \[
    \mathrm{Res}(k+1)-\mathrm{Res}(k)=1,
    \qquad
    \mathrm{Con}(k+1,b)-\mathrm{Con}(k,b)=(k+2)|b|.
  \]
  If the value of one further step is bounded by $V_{\max}$ and unit cost is at least one, continuation stops by depth $V_{\max}$. If the licensed gain is zero, it stops at depth zero.

  These are resource halts relative to the stated cost model.

  With a continuation value the zero-gain halt is conditional. In a finite-branching sequential model whose continue action stays inside one observer fiber, zero current gain and positive unit cost force a stop when every reachable continuation is valued at most as an immediate stop; a one-state model with zero current gain, unit cost one, and continuation value two makes one more step optimal. The static comparison above is the special case in which the only continuation is valued at the stopping value, and the dynamic policy agrees with the static one there.
\end{remark}

\begin{theorem}[Stopping for recursor-depth extension]
  \label{thm:recursor-stopping}
  Fix any natural unit price $u$, payload size $p$, and duplication
  arity $r$, including zero. Extending the selected depth from $k$ to
  $k+1$ appends a live state of the longer recursor execution with
  $r(k+1)+1$ payload copies. Its marginal price is
  $d_k=up\bigl(r(k+1)+1\bigr)$, so the cumulative price satisfies
  $C(k+1)-C(k)=d_k$.
  Let $B,S:\mathbb N\to\mathbb R$ be arbitrary continuation benefits
  and stopping rewards. At depth $k$ the decision is either to stop
  for $S(k)$ or to pay $d_k$, receive $B(k)$, and advance to $k+1$.
  For every finite horizon $n$, the optimal value is
  \[
    V_n(k)=\max_{0\leq t\leq n}
    \left\{S(k+t)+\sum_{i=0}^{t-1}\bigl(B(k+i)-d_{k+i}\bigr)\right\}.
  \]
  A stopping time and a policy attain this maximum. Immediate stopping
  is optimal at every finite horizon if and only if every finite prefix
  return in braces is at most $S(k)$.
\end{theorem}
\begin{proof}
  The live orbit's payload-count formula gives $d_k$ and the cumulative
  difference. Put $V_0(k)=S(k)$. The single continuation gives
  \[
    V_{n+1}(k)=\max\{S(k),B(k)-d_k+V_n(k+1)\}.
  \]
  Induction separates the stopping time zero from the positive stopping
  times and proves the displayed maximum formula. Its finite domain
  contains a time that attains the maximum; continuing that many times and then
  stopping realizes its return. All policies are bounded by the same
  recurrence. Last, use horizon $n=t$ for necessity of the prefix
  criterion and the maximum formula for sufficiency.
\end{proof}

\begin{corollary}[Least stopping depth and degenerate costs]
  \label{cor:recursor-stopping-cutoff}
  For constant benefit $B(k)=B\in\mathbb N$ and $S(k)=0$, with $u,p,r>0$,
  immediate stopping is optimal at every finite horizon if and only if
  \[
    k\geq k_*:=
    \max\left\{
      \left\lceil\frac{\max\{\lceil B/(up)\rceil-1,0\}}{r}\right\rceil-1,0
    \right\}.
  \]
  For $r=1$ this is $k_*=\max\{\lceil B/(up)\rceil-2,0\}$.
  At arity zero the condition is $B\leq up$ at every depth.
  If $up=0$ and $B>0$, one continuation has positive value.
\end{corollary}
\begin{proof}
  For positive $up$ and $r$, the marginal prices increase with depth.
  Thus every prefix return is nonpositive if and only if $B\leq d_k$;
  otherwise the one-step return is positive. Two ceiling divisions
  solve this inequality over the natural numbers. At arity zero the
  price is the constant $up$; at zero price the one-step return is $B$.
\end{proof}

The state in Theorem~\ref{thm:recursor-stopping} is a choice of input
depth. Each such input has an execution of $k+1$ rewrite steps;
the decision transition compares depths of different inputs. The
prefix criterion includes future stopping rewards: for $B(k)=0$,
$S(k)=3k$, and $u=p=r=1$, continuing from depth zero costs two and
returns three, although the current benefit is zero.

\begin{theorem}[Finite-horizon Snell envelope]\label{thm:snell-envelope}
  In a finite-horizon stopping model with stopping values $g$, a finite set of legal actions at each state, with finite outcome laws
  and rewards, $V_0=g$, and $V_{n+1}=\max\bigl(g,\max_a\mathbb E_a[r+V_n]\bigr)$:
  \begin{enumerate}[label=(\alph*),leftmargin=1.6em]
    \item $(V_n)$ dominates $g$ and satisfies $V_{n+1}\ge\mathbb E_a[r+V_n]$ for every legal action $a$, and
      every sequence with these two properties dominates $(V_n)$;
    \item $V_n(s)$ is the greatest expected return of an adaptive policy of horizon $n$ from $s$;
    \item the policy that stops where $V_n=g$ and otherwise takes an action attaining $V_n$
      attains $V_n$.
  \end{enumerate}
\end{theorem}
\begin{proof}
  (a) The Bellman value dominates stopping and each legal continuation, and induction on the horizon shows that any sequence with both properties dominates it. (b) By induction on the horizon, a policy that stops earns $g\le V_n$, and a policy that first takes $a$ earns at most $\mathbb E_a[r+V_{n-1}]\le V_n$; the rule of (c) attains $V_n$. (c) The envelope rule attains the value at every horizon, and it stops where the value meets the stopping value.
\end{proof}

With one legal action this is the finite-horizon envelope of Snell~\cite{snell1952,chowRobbinsSiegmund1971}. The recursor model of Theorem~\ref{thm:recursor-stopping} has one action, so its Bellman value is the Snell envelope of its stopping rewards, and the price $d_k$ of the live extension is its only datum specific to the recursor.

\begin{proposition}[Zero-deficit stopping]\label{prop:zero-deficit-stopping}
  Consider a finite decision problem: a nonempty finite set of actions, a finite target $X$, an observer value $W$ with weights $\mu$, a channel value $C$ with conditional law $\nu(c\mid w)$ summing to one, a conditional target law $r(x\mid w,c)$, and a loss $\ell(a,x)\in\mathbb N$. The stop branch acts on $W$ alone at the least expected loss; the compute branch pays a cost and acts on $(W,C)$. If the channel is uninformative, $r(\cdot\mid w,c)=s(\cdot\mid w)$ for every $c$, then the licensed-channel deficit is zero, the value of computation, the loss reduction before cost, is zero, and every positive cost makes stopping the optimal policy. In the two-action two-state model with mismatch loss $2$, uniform prior, and unit cost, the copy channel has deficit zero and optimal policy stop, and the revealing channel has positive deficit and optimal policy compute.
\end{proposition}
\begin{proof}
  Mixing out an uninformative channel leaves the expected loss of every action in every $(w,c)$ cell equal to its expected loss given $w$, so the least losses agree cell by cell and the two branches have equal loss before cost; a positive cost then breaks the tie toward stopping. The instance is computed.
\end{proof}

\begin{proposition}[Licensed deliberation is well founded]\label{prop:reflective-pairs}
  Let a deliberation state consist of a state label, a query label, a finite multiset of requirements over $\mathbb N$, and a budget in $\mathbb N$. Call a transition \emph{licensed} when the budget decreases, or the budget is unchanged and the requirement multiset decreases in the Dershowitz and Manna order. Every licensed transition decreases the lexicographic rank (budget, requirements); the licensed relation is well founded; no state has a licensed transition to itself; no licensed path returns to its source; and a transition between two states of equal budget and equal requirements, whatever their labels, is unlicensed. A three-state deliberation that discharges a requirement and then spends a unit of budget reaches a terminal state.
\end{proposition}
\begin{proof}
  The lexicographic product of the well-founded orders on $\mathbb N$ and on finite multisets is well founded, and each licensed transition descends in it. Irreflexivity and the absence of cycles follow, since a cycle would descend around to its start. Equal ranks have no strict descent. The instance is computed.
\end{proof}

The identity transition is the deliberative form of the zero-deficit channel of Proposition~\ref{prop:zero-deficit-stopping}: both are refused.

\begin{proposition}[Payload size]\label{prop:payload-ratio}
  The payload size $|b|$ and counter height $k$ are independent parameters. Unlike derivational-complexity bounds stated in the total size of the start term~\cite{hofbauer1992mpo,weiermann1995lpo,hofbauerlautemann1989rta,moser2009habil}, the following ratio isolates the duplicated subterm; for $k\ge1$,
  \[
    \frac{\mathrm{Con}(k,b)}{\mathrm{Res}(k)^2}
    =\frac{(k+1)(k+2)|b|}{2k^2}
    \;\xrightarrow{k\to\infty}\;\frac{|b|}{2}.
  \]
  For $|b|=0$ the ratio is zero; for $|b|\ge1$ its limit scales in proportion to the payload size.
\end{proposition}
\begin{proof}The closed form gives the ratio as $\beta(k+1)(k+2)/(2k^2)$ for $k\ge1$; dividing by $k^2$ and letting $k\to\infty$ gives the displayed limit, and the zero and positive payload cases follow from the same formula.
\end{proof}

\begin{proposition}[Trace action and the second quadratic invariant]\label{prop:trace-action-law}
  With $w$ and $c$ as in Proposition~\ref{prop:exact-runtime-mass-rate}, define the live-trace action
  \[
    A(k):=\sum_{i=0}^{k}|t_i|.
  \]
  Then
  \[
    2A(k)=k(k+1)(w+1)+2(k+1)c.
  \]
  In particular, $A(k)/k^2\to (w+1)/2$; by Proposition~\ref{prop:payload-ratio}, the payload-only ratio $\mathrm{Con}(k,b)/\mathrm{Res}(k)^2$ has the limit $|b|/2$.
\end{proposition}

\begin{proof}
  Insert $|t_i|=iw+(k-i)+c$ and apply the two triangular sums
  $\sum_{i=0}^{k}i=\sum_{i=0}^{k}(k-i)=k(k+1)/2$.
\end{proof}

\begin{proposition}[Proof-entropy monotonicity]\label{prop:proof-entropy}
  Define the \emph{proof entropy} of the term~$t_i$ in the trace as the fraction of its size occupied by payload copies in completed wrappers, which the dependency-pair comparison omits:
  \[
    H_{\mathrm{proof}}(t_i):=\frac{\text{completed-wrapper payload size in }t_i}{\text{size of }t_i}.
  \]
  The completed wrappers at step~$i$ carry payload size $i|b|$. Let $c_\ast$ denote the fixed structural contribution of the active $F$-symbol, the base value~$a$, the active payload occurrence, and the zero-side constructor overhead. Thus $c_\ast=c_0+|b|$ for a payload-independent constant $c_0$. Then
  \[
    |t_i|=i(|G|+|b|)+(k-i)+c_\ast,
  \]
  so
  \[
    H_{\mathrm{proof}}(t_i)=\frac{i|b|}{i(|G|+|b|)+(k-i)+c_\ast}.
  \]
  Then for $|b|\ge 1$:
  \[
    H_{\mathrm{proof}}(t_0)=0,
    \qquad
    H_{\mathrm{proof}}(t_i)\to \frac{i}{i+1}\quad\text{as }|b|\to\infty\text{ for fixed }i\ge 1.
  \]
  $H_{\mathrm{proof}}$ is nondecreasing for $0\le i<k$, that is across every live stage of the trace.
\end{proposition}

\begin{proof}
  At step~$0$ the $G$-wrapper stack is empty and the completed-wrapper payload size is zero. For $i\ge 1$, set
  \[
    D_i:=i(|G|+|b|)+(k-i)+c_\ast = k+c_\ast+i(|G|+|b|-1).
  \]
  Then
  \[
    H_{\mathrm{proof}}(t_{i+1})-H_{\mathrm{proof}}(t_i)
    =\frac{|b|((i+1)D_i-iD_{i+1})}{D_iD_{i+1}}
    =\frac{|b|(k+c_\ast)}{D_iD_{i+1}}\ge 0.
  \]
  So $H_{\mathrm{proof}}$ is nondecreasing for $0\le i<k$. For fixed $i$, the numerator has dominant term $i|b|$ and the denominator has dominant term $(i+1)|b|$, which gives the stated limit.
\end{proof}

\begin{remark}[Interpretation of $H_{\mathrm{proof}}$]\label{rem:proof-entropy-interpretation}
  The quantity $H_{\mathrm{proof}}$ is the fraction of term size occupied by payload copies in completed wrappers. The dependency-pair comparison omits these copies. This is a structural ratio; the Shannon quantity is $H_{\mathrm{pos}}$. The payload value remains determinate.
\end{remark}

\begin{proposition}[Total omitted cell mass]\label{prop:confession-dominance-full}
  Let $w=|G|+|b|$ be the wrapper-cell weight. Define the \emph{homogeneous-cell envelope} by assigning one cell budget $w$ to each of the $i+1$ payload-bearing positions at live stage $i$. Its cumulative mass is
  \[
    \mathrm{Con}_{\mathrm{total}}(k,w)=\sum_{i=0}^{k}(i+1)w=\frac{(k+1)(k+2)}{2}w,
  \]
  and
  \[
    \frac{\mathrm{Con}_{\mathrm{total}}(k,w)}{\mathrm{Res}(k)}\sim \frac{kw}{2}\quad\text{as }k\to\infty.
  \]
\end{proposition}

\begin{proof}
  The coding convention pads the active payload position with one wrapper-symbol budget so that every payload-bearing position has the same cell weight $w$. There are $i+1$ such positions at stage $i$, and summing gives the formula. The envelope is a budget instead of a symbol count: the completed-wrapper count in the syntax is $i$, and Proposition~\ref{prop:exact-mass-partition} separates the two.
\end{proof}

\begin{proposition}[Mass partition and the asymptotic omitted fraction]\label{prop:exact-mass-partition}
  Let $A(k)$ be the trace action of Proposition~\ref{prop:trace-action-law}, and let
  \[
    \mathrm{Con}_{\mathrm{cell}}(k,w):=\sum_{i=0}^{k}iw=\frac{k(k+1)}{2}w
  \]
  count the completed wrapper-cell mass, setting aside the active payload occurrence. Then the division-free partition is
  \[
    w\,2A(k)=(w+1)\,2\mathrm{Con}_{\mathrm{cell}}(k,w)+2w(k+1)c.
  \]
  Hence $\mathrm{Con}_{\mathrm{cell}}(k,w)/A(k)\to w/(w+1)$. Here $\mathrm{Con}_{\mathrm{cell}}$ counts completed wrapper-cell mass; $\mathrm{Con}_{\mathrm{total}}$ of Proposition~\ref{prop:confession-dominance-full} also budgets a padded active payload-bearing cell at each stage.
\end{proposition}

\begin{proof}
  Substitute the closed forms $2A(k)=k(k+1)(w+1)+2(k+1)c$ and $2\mathrm{Con}_{\mathrm{cell}}=wk(k+1)$ and expand. Division by $k^2$ with fixed $w$ and $c$ gives the limit.
\end{proof}

\begin{proposition}[Crossover or majority-loss law]\label{prop:crossover-law}
  Assume $w\ge1$. The first integer stage at which completed wrapper-cell mass reaches the retained live-state mass is
  \[
    i_\ast=\left\lceil\frac{k+c}{w+1}\right\rceil
    =\frac{k+c+w}{w+1}
  \]
  with natural-number division. Whenever $i_\ast\le k$,
  \[
    i_\ast w\ge(k-i_\ast)+c,
  \]
  and every $j<i_\ast$ satisfies $jw<(k-j)+c$. The integer envelope
  \[
    k+c\le i_\ast(w+1)\le k+c+w
  \]
  implies $i_\ast/k\to1/(w+1)$ for fixed $w$ and $c$.
\end{proposition}

\begin{proof}
  The majority inequality is equivalent to $i(w+1)\ge k+c$. Natural-number ceiling division gives the displayed least solution and its two-sided envelope.
\end{proof}

\begin{proposition}[Four payload observables]\label{prop:norm-mismatch}
  Let $i\ge 1$ and $|b|>0$, and consider the constant wrapper vector $v_i=(b,\ldots,b)\in\Delta_i\subset B^i$. Four observables must be distinguished:
  \begin{itemize}[nosep,leftmargin=1.4em]
    \item the number of distinct components is $\rho_\Delta(v_i)=1$;
    \item the standard $\ell^0$ support count is $\|v_i\|_0=i$;
    \item the maximum component size is $\|v_i\|_\infty=|b|$;
    \item the additive component size is $\|v_i\|_1=i|b|$.
  \end{itemize}
  Thus the number of distinct components, the support size, the maximum size, and the additive mass are different coordinates. Direct additive aggregation reads the last coordinate. The dependency-pair transformation removes the wrapper context; counter projection then ignores the payload arguments and retains the counter.
\end{proposition}

\begin{proof}
  The support count is $i$ because all $i$ components have positive size. The constant tuple has one distinct component; the maximum and sum follow from the constant tuple. Dependency-pair extraction removes the wrappers before projection selects the counter.
\end{proof}

\begin{proposition}[Growth of the payload dimension]\label{prop:payload-dimension}
  The orbit is a one-parameter curve with an injective index, while its live payload carrier has $i+1$ coordinates at stage $i$. Thus the trace index remains one-dimensional while the number of payload positions grows by one per stage. Proposition~\ref{prop:norm-mismatch} separates the corresponding distinct-component, support, maximum, and additive coordinates.
\end{proposition}
\begin{proof}The orbit index is injective and stage $i$ carries $i+1$ live payload coordinates; Proposition~\ref{prop:norm-mismatch} separates the coordinate types.
\end{proof}

\subsection{Hidden progress and its channels}

\begin{proposition}[Information content of the hidden-progress coordinate]\label{prop:hidden-progress-entropy}
  Let $K\in\mathbb N$ be fixed and let $j$ denote the progress index along the trace $t_0,t_1,\dots,t_K$. Under a uniform prior on $j\in\{0,1,\dots,K\}$, the Shannon entropy of the progress coordinate is
  \[
    H(j) = \log_2(K+1).
  \]
  Fix a prefix-free binary encoding of the finite constructor alphabet, and measure a payload by the bit length of its encoding. An observer can determine $j$ through four channels, with the following output-description costs:
  \begin{enumerate}[nosep,leftmargin=1.6em]
    \item \emph{History observation}: store the state sequence $t_0,\dots,t_j$. Worst-case cost over $j\le K$: $\Theta(K^2(|b|+1))$ bits.
    \item \emph{Structural counting}: count $\#_G(t_j)$ on the state $t_j$. Cost: $\lceil\log_2(K+1)\rceil$ bits.
    \item \emph{Origin comparison}: compare the stage counter $S^{K-j}(0)$ with the initial $S^K(0)$. Cost: $\lceil\log_2(K+1)\rceil$ bits.
    \item \emph{Parallel clocking}: maintain an external counter of rule firings. Cost: $\lceil\log_2(K+1)\rceil$ bits.
  \end{enumerate}
  Call an injective map from the $K+1$ progress states into the $2^b$ binary words of length $b$ a lossless $b$-bit progress code. Every lossless progress code satisfies
  \[
    b\ge \lceil\log_2(K+1)\rceil.
  \]
  Channels (2) to (4) attain this universal lower bound for every $K$. This statement concerns the encoded output. Channel (1) stores the trace instead of only the recovered index.
\end{proposition}

\begin{proof}
  Under uniform prior $\pi(j)=1/(K+1)$, the Shannon entropy is $H(j)=-\sum_{j=0}^{K}\pi(j)\log_2\pi(j)=\log_2(K+1)$.

  \emph{Channel costs.} (1) Under the fixed constructor encoding, each intermediate state has bit length $\Theta(j(|b|+1)+(K-j)+1)$, so storing the length-$K+1$ trace costs $\Theta(K^2(|b|+1))$ bits in the worst case. (2) $\#_G(t_j)=j$ by Proposition~\ref{prop:trace-law}. (3) The counter-height difference is $j$. (4) The clock stores $j$. In cases (2) to (4), the value lies in $\{0,\dots,K\}$ and has fixed-length code size $\lceil\log_2(K+1)\rceil$ bits. For an arbitrary lossless $b$-bit code, injectivity gives $K+1\le2^b$ and hence $b\ge\lceil\log_2(K+1)\rceil$. The three displayed channels are injective and use equality-width codes, proving optimality.
\end{proof}

\begin{corollary}[Mutual information of the terminal record]\label{cor:meta-trace-mi}
  Let $K_{\max}>0$, let $K$ be uniform on $\{1,\dots,K_{\max}\}$, and fix the base and payload so that the terminal record varies only with $K$. Then
  \[
    I(K;T_{K+1}) \;=\; \log_2 K_{\max},
  \]
  because $K=\#_G(T_{K+1})$ by Proposition~\ref{prop:live-vs-record}(3). In the finite probability model, the observation likelihood is the deterministic channel that counts the frames of the terminal record, its evidence mass is $1/K_{\max}$, and Bayes' rule yields the point posterior at the observed depth.
\end{corollary}

\begin{proof}
  The prior entropy is $H(K)=\log_2 K_{\max}$. The terminal frame-count channel is the identity on the zero-based encoding of $\{1,\dots,K_{\max}\}$. Its Bayes posterior is therefore a point mass, so $H(K\mid T_{K+1})=0$ and $I(K;T_{K+1})=H(K)$.
\end{proof}

\begin{proposition}[Channel costs]
  \label{prop:channel-cost-gap}
  Channels (2) to (4) encode the progress coordinate in $\lceil\log_2(K+1)\rceil$ bits, while channel (1) stores the trace at cost $\Theta(K^2(|b|+1))$. Proposition~\ref{prop:confession-dominance} has the same quadratic dependence on trace depth because it sums payload-bearing positions over all stages.
\end{proposition}
\begin{proof}Channels (2) to (4) attain the least width of a faithful code, $\lceil\log_2(K+1)\rceil$ bits, while channel (1) stores the trace at cost $w(K+1)(K+2)/2\ge K(K+1)/2$ for payload width $w\ge1$.
\end{proof}

\subsection{Coding model for payload positions}\label{subsec:shannon-validator}

The proof-theoretic results of Sections~\ref{sec:oi} and~\ref{sec:witness} are independent of information theory. This subsection introduces a coding model that compares the repeated payload mass with the description length of a selected payload position.

\begin{definition}[Position entropy]\label{def:gauge-entropy}
  Fix step $i$ of the trace. Under the uniform coding convention on the $i+1$ payload-bearing positions, let $X_i$ be the random variable naming which position is singled out as the active locus. The \emph{position entropy} is
  \[
    H_{\mathrm{pos}}(i):=H(X_i)=\log_2(i+1).
  \]
\end{definition}

The randomness is in the chosen label; the payload tuple itself is constant and has zero entropy when treated as a deterministic value.

\begin{remark}[Two distinct entropy objects]
  \label{rem:two-entropy-objects}
  The quantity $H_{\mathrm{pos}}(i)$ is a Shannon entropy under the coding model of Definition~\ref{def:gauge-entropy}. The proof-entropy fraction $H_{\mathrm{proof}}(t_i)$ is a structural ratio inside the term algebra. They measure different objects.
\end{remark}

\begin{definition}[Inefficiency coefficient]\label{def:inefficiency}
  For $k\ge 1$, let
  \[
    \eta(k,w):=\frac{\mathrm{Con}_{\mathrm{total}}(k,w)}{\ln 2\cdot H_{\mathrm{pos}}(k)}
    =\frac{(k+1)(k+2)w}{2\ln(k+1)}.
  \]
  The restriction to $k\ge 1$ keeps the position entropy in the denominator positive.
  This coefficient compares the payload mass that the direct whole-term observer reads with the information content of the position label. It measures the direct representation; transformed proofs have separate costs.
\end{definition}

\begin{proposition}[Divergence of the inefficiency coefficient]\label{prop:inefficiency-classification}
  Fix $w\ge 1$. Then for every $N\in\mathbb N$,
  \[
    \eta(2N+1,w)\ge N.
  \]
  In particular, the inefficiency coefficient is unbounded along
  the trace.
\end{proposition}

\begin{proof}
  By Definition~\ref{def:inefficiency},
  \[
    \eta(k,w)=\frac{(k+1)(k+2)w}{2\ln(k+1)}.
  \]
  For $k\ge 1$ we have $\ln(k+1)\le k$, hence
  \[
    \eta(k,w)\ge \frac{(k+1)w}{2}.
  \]
  Substituting $k=2N+1$ gives $\eta(2N+1,w)\ge (N+1)w\ge N$ when $w\ge 1$.
\end{proof}

\begin{proposition}[Asymptotic rate of inefficiency divergence]\label{prop:inefficiency-rate}
  Fix $w\ge 1$. Then for all $k\ge 1$:
  \begin{enumerate}[nosep,leftmargin=1.6em]
    \item \emph{Closed form.} $\eta(k,w)=\dfrac{(k+1)(k+2)w}{2\ln(k+1)}$.
    \item \emph{Quadratic growth with a logarithmic divisor.} $\eta(k,w)=\Theta\!\left(\dfrac{k^2 w}{\ln(k+1)}\right)$ as $k\to\infty$.
    \item \emph{Per-residual rate.} $\dfrac{\eta(k,w)}{\mathrm{Res}(k)}=\dfrac{(k+1)(k+2)w}{2k\ln(k+1)}=\Theta\!\left(\dfrac{kw}{\ln(k+1)}\right)$.
    \item \emph{Per-step rate.} $\dfrac{\eta(k,w)}{k+1}=\dfrac{(k+2)w}{2\ln(k+1)}=\Theta\!\left(\dfrac{kw}{\ln(k+1)}\right)$.
  \end{enumerate}
  Thus the absolute rate is $k^2/\ln(k+1)$, and the rate per residual step is $k/\ln(k+1)$.
\end{proposition}

\begin{proof}
  \emph{(1)} Direct substitution: $\eta(k,w)=\mathrm{Con}_{\mathrm{total}}(k,w)/(\ln 2\cdot H_{\mathrm{pos}}(k))$ with $\mathrm{Con}_{\mathrm{total}}(k,w)=\tfrac{(k+1)(k+2)}{2}w$ (Proposition~\ref{prop:confession-dominance-full}) and $H_{\mathrm{pos}}(k)=\log_2(k+1)=\ln(k+1)/\ln 2$.
  \emph{(2)} The displayed closed form has numerator of order $k^2w$ and denominator of order $\ln(k+1)$, giving $\eta(k,w)=\Theta(k^2w/\ln(k+1))$.
  \emph{(3)} Divide by $\mathrm{Res}(k)=k$.
  \emph{(4)} Divide by the trace length $k+1$; the same asymptotic calculation gives the stated per-step rate.
\end{proof}

\begin{proposition}[The sampled bound]
  \label{prop:inefficiency-odd-bound}
  For payload width $w\ge1$ and every $N$, the inequality $\eta(2N+1,w)\ge N$ holds; it is a lower bound along odd indices. Proposition~\ref{prop:inefficiency-rate} gives the asymptotic $\Theta(k^2w/\ln(k+1))$.
\end{proposition}
\begin{proof}Substituting the odd index $2N+1$ into the linear lower bound gives $N\le((2N+2)w)/2$ for $w\ge1$.
\end{proof}

\begin{proposition}[Marginal cost per bit of position information]\label{prop:marginal-gauge-cost}
  Let $\mathrm{Cum}(j,w)=wj(j+1)/2$ be the cumulative completed-cell mass through stage $j$, and for $m\ge1$ define
  \[
    \mathrm{bitCost}(m,w):=\mathrm{Cum}(2^m-1,w)-\mathrm{Cum}(2^{m-1}-1,w).
  \]
  Then
  \[
    2\,\mathrm{bitCost}(m,w)+w2^{m-1}=3w4^{m-1},
    \qquad
    \mathrm{bitCost}(m,w)\ge w4^{m-1}.
  \]
  By contrast, for $k\ge 1$ the projected counter requires $\mathrm{projBits}(k)=\lfloor\log_2 k\rfloor+1$ bits and one register cell per bit. The divergence of $\eta$ therefore has a marginal form: each additional bit of the position code costs a block of payload mass larger than the last by a constant factor.
\end{proposition}

\begin{proof}
  Apply the triangular identity to the two adjacent power-of-two cutoffs and simplify. The projection inequality is $k<2^{\mathrm{projBits}(k)}$.
\end{proof}

\begin{proposition}[Description-length linear gap under the encoding]\label{prop:shannon-uniqueness}
  Let
  \[
    M_i:=(i+1)(|b|+|G|)
  \]
  be the repeated-carrier envelope, and let
  \[
    L_{\mathrm{exp}}(i):=|b|+|G|+\mathrm{size}_2(i+1)+c_0
  \]
  be the description length in a fixed computable encoding, obtained from one payload description, one
  wrapper-symbol description, a binary-length code $\mathrm{size}_2(i+1)$ for
  the step index, and fixed glue overhead $c_0$. Then
  \[
    i(|b|+|G|)+L_{\mathrm{exp}}(i)=M_i+\mathrm{size}_2(i+1)+c_0,
  \]
  which rearranges to
  \[
    M_i-L_{\mathrm{exp}}(i)=i(|b|+|G|)-\mathrm{size}_2(i+1)-c_0.
  \]
  The difference is therefore linear in $i$ against a logarithmic indexing overhead, and it is
  positive above a stated bound: for cell weight $|b|+|G|\ge 2$, index $i\ge 2$, and glue
  overhead $c_0<i$, the repeated-carrier envelope strictly exceeds the description. At $i=0$
  the inequality runs the other way.
\end{proposition}

\begin{proof}
  Expand both sides from the definitions of $M_i$ and $L_{\mathrm{exp}}(i)$:
  \[
    i(|b|+|G|)+\bigl(|b|+|G|+\mathrm{size}_2(i+1)+c_0\bigr)
    =(i+1)(|b|+|G|)+\mathrm{size}_2(i+1)+c_0.
  \]
  Rearranging yields the second display.
\end{proof}

\begin{definition}[Payload vector and collapse map]\label{def:carrier-collapse}
  At step $i$, let
  \[
    u_i:=(b_0,\ldots,b_i)\in B^{i+1}
  \]
  be the \emph{payload vector} listing the payload occurrences stored in the wrapper stack together with the active payload occurrence inside the live $F$-subterm. On the trace every component equals the same value~$b$, so $u_i\in\Delta_{i+1}$. Define the \emph{collapse map}
  \[
    c_i:\Delta_{i+1}\to B,
    \qquad
    c_i(c,\ldots,c):=c.
  \]
  The map $c_i$ forgets the multiplicity and keeps the common value.
\end{definition}

\begin{proposition}[Factorization through the collapse maps]\label{prop:carrier-factorization}
  Let $\{O_i\}_{i\ge 0}$ be a family of payload observables with
  \[
    O_i:\Delta_{i+1}\to Z
  \]
  for some codomain~$Z$. The following are equivalent:
  \begin{enumerate}[nosep,leftmargin=1.6em]
    \item for every $c\in B$ and all $i,j\ge 0$,
      \[
        O_i(c,\ldots,c)=O_j(c,\ldots,c),
      \]
      so the observable ignores multiplicity along the diagonal;
    \item there exists one and only one map $\overline O:B\to Z$ such that
      \[
        O_i=\overline O\circ c_i
        \qquad\text{on }\Delta_{i+1}\text{ for every }i\ge 0.
      \]
  \end{enumerate}
  Hence an observable ignores multiplicity if and only if it factors through the collapse maps~$c_i$.
\end{proposition}

\begin{proof}
  If (2) holds, then for every $c$ and all $i,j$,
  \[
    O_i(c,\ldots,c)=\overline O(c_i(c,\ldots,c))=\overline O(c)=\overline O(c_j(c,\ldots,c))=O_j(c,\ldots,c).
  \]
  In the other direction, assume (1). Define $\overline O(c)$ to be the common value $O_i(c,\ldots,c)$. This is well defined by multiplicity-insensitivity, and then $O_i=\overline O\circ c_i$ on $\Delta_{i+1}$ for every $i$. Uniqueness is immediate because each $c_i$ is surjective.
\end{proof}

\noindent The collapse maps $c_i$ of Proposition~\ref{prop:carrier-factorization} are the target-kernel quotient of Proposition~\ref{prop:target-kernel-quotient} for the common value on the diagonal orbit, and Proposition~\ref{prop:carrier-factorization} is the license criterion of Proposition~\ref{prop:license-criterion} for the observer that returns this value. The counter observer licenses the residual-work target of the dependency-pair proof, and a constant observer fails to license the payload target. Proposition~\ref{prop:side-channel-minimum} gives the least side alphabet that repairs an unlicensed observer.

\begin{corollary}[Additive observables depend on multiplicity]\label{cor:additive-nonfactor}
  The family of direct additive observables
  \[
    O_i^{(1)}(u_i):=\|u_i\|_{\ell^1}=(i+1)|b|
  \]
  fails to factor through the collapse maps $c_i$, because for fixed value~$b$ it varies strictly with~$i$. By contrast, the value-only observable
  \[
    O_i^{(\mathrm{val})}(u_i):=|c_i(u_i)|=|b|
  \]
  does factor through $c_i$. Direct whole-term aggregation therefore reads the multiplicity itself.
\end{corollary}

\begin{proof}
  For fixed $b$, the value $(i+1)|b|$ depends on $i$, so Proposition~\ref{prop:carrier-factorization}(1) fails and $O_i^{(1)}$ stays outside the image of $c_i$. The value-only observable is, by definition, $\overline O\circ c_i$ with $\overline O(c)=|c|$.
\end{proof}

\begin{remark}[Scope of the coding model]
  \label{rem:what-shannon-validates}
  The grammar theorem and the coding model concern different objects: the grammar theorem excludes payload-dependent orientations in the direct grammar, and the coding model compares carrier multiplicity with the description length of a selected position.
\end{remark}

\begin{proposition}[Quotient factor law]\label{prop:quotient-factor-law}
  For a finite deterministic channel that factors through an observer's kernel, replacing the observer with its quotient factor preserves the pushforward distribution, the Shannon value, and the Hartley support.
\end{proposition}

\begin{proof}
  The pushforward through the observer equals the pushforward of the quotient factor, and the entropy and support of the distribution depend on it alone.
\end{proof}

\begin{proposition}[Output information and certificate information differ]
  \label{prop:verdict-zero-entropy}
  For the recursor relation, every sound Boolean termination verdict returns true. Under any normalized prior supported on such output functions, evaluation at the recursor induces a point mass and has zero Shannon entropy. The random variable here is the output instead of the termination certificate.
\end{proposition}
\begin{proof}A sound termination verdict returns true at the recursor, so the output law is a point mass there and has zero Shannon entropy; the two fixtures place the conclusion at the recursor.
\end{proof}

\subsection{Certificates}\label{subsec:certificates}

\begin{proposition}[Certificate length in the assembly model]\label{prop:ag-derivational}
  Let a finite TRS shape record $r$ rules, maximum rule size $\sigma$, and $n$ extracted pairs. Fix an additive base-order proof-length function $L_{\mathrm{base}}(n)=nL_{\mathrm{base}}(1)$. Define the assembly costs
  \[
    L_{\mathrm{construct}}:=r^2\sigma,\qquad
    L_{\mathrm{sound}}:=1,\qquad
    L_{\mathrm{AG}}:=L_{\mathrm{construct}}+L_{\mathrm{base}}(n)+L_{\mathrm{sound}}.
  \]
  If $r\ge1$ and $\sigma\ge1$, then
  \[
    L_{\mathrm{AG}}\le 2r^2\sigma+L_{\mathrm{base}}(n).
  \]
  The displayed accounting model defines the bound. Signature cardinality cannot replace $\sigma$, since rules over a fixed signature have unbounded size. Arts and Giesl~\cite{artsgiesl2000} supply the transformation theorem; the displayed equations define the costs.
\end{proposition}

\begin{proof}
  By definition, $L_{\mathrm{AG}}=r^2\sigma+L_{\mathrm{base}}(n)+1$. Since $r^2\sigma\ge1$, the final unit is at most a second copy of $r^2\sigma$, giving the bound. Maximum rule size is necessary because rule size is unbounded at fixed rule count and signature.
\end{proof}

\begin{corollary}[AG proof-length on the step-duplicating recursor]\label{cor:ag-recursor}
  For the step-duplicating recursor $F(x,y,Z)\to x$, $F(x,y,S(n))\to G(y,F(x,y,n))$, set $r=2$, let $\sigma$ be the maximum rule size, and set $n=1$. With $L_{\mathrm{base}}(1)=1$,
  \[
    L_{\mathrm{AG}}(\mathrm{recursor})=4\sigma+2.
  \]
  For the two rules, $\sigma$ is constant and the assembly charge is independent of input counter height $K$. Adding the residual work $\mathrm{Res}(K)=K$ gives modeled length $4\sigma+2+K$.
\end{corollary}

\begin{proof}
  Substitute $r=2$, $n=1$, and $L_{\mathrm{base}}(1)=1$ into Proposition~\ref{prop:ag-derivational}, then add $\mathrm{Res}(K)=K$.
\end{proof}

\begin{remark}[Consistency with Moser and Schnabl]
  \label{rem:consistency-moser-schnabl}
  Corollary~\ref{cor:ag-recursor} is consistent with the Moser and Schnabl preservation result~\cite{moserschnabl2011derivational}, which concerns derivational-complexity classes; the recursor step numbers and the constant charge are recorded in Proposition~\ref{prop:certificate-charge}.
\end{remark}

\begin{proposition}[Certificate-size consequence of the Arts and Giesl license]
  \label{prop:certificate-charge}
  Within the assembly model, the charge is $4\sigma+2+K$.
\end{proposition}
\begin{proof}The closed forms give the certificate length $4\sigma+2+K$ and a license overhead constant in the input; the recursor's $K+1$ steps on inputs of size $O(K)$ carry the charge $4\sigma+2$.
\end{proof}

\begin{theorem}[Strict compact certificates characterize root steps]\label{thm:strict-certificate}
  For the free two-rule root relation $R_0$, there is a self-delimiting compact
  binary certificate language with a strict whole-message checker such that, for
  every source $a$ and proposed right-hand side $u$,
  \[
    \exists z\in\{0,1\}^*\;\mathrm{Check}_{\mathrm{strict}}(a,u,z)=1
    \quad\Longleftrightarrow\quad a\,R_0\,u.
  \]
  The forward direction holds for every accepted bitstring, including inputs
  that no serializer generated; the reverse direction builds a certificate from
  the root step.  The prefix checker is a separate streaming interface and may
  leave a suffix, whereas the strict checker rejects every nonempty trailing
  suffix after a constructed certificate.

  Certificate words use a self-delimiting compact binary encoding:
  each natural field consists of a unary-coded binary width followed by the binary
  payload.  The serialized bit length is the sum of those field lengths,
  including repeated fields.  The compact decoder is instrumented bit by bit:
  its bit-inspection counter is produced by the same recursion that
  parses the unary width, payload bits, and word sequence.  On every input this
  counter is at most the input length; on a constructed certificate followed by an
  arbitrary suffix it equals the serialized certificate length, so the suffix
  is not inspected.  Word-format operations and same-input proof-checker operations
  have separate counters, each bounded on all inputs by the parsed word list,
  source-tree size, right-hand-side tree size, and path size; the statement keeps
  these counts in their own units.  The assembly model of
  Proposition~\ref{prop:ag-derivational} is a separate cost model.
\end{theorem}
\begin{proof}
  The strict checker first parses the compact certificate and requires an empty
  returned tail.  Any accepted decoded certificate satisfies the same-input
  structural checker, whose soundness theorem returns an $R_0$ root step.
  In the other direction, each zero or successor root step has an executable certificate
  constructor, its compact serialization round-trips, and the strict checker
  accepts it.  Prefix round-trip with arbitrary suffix proves the streaming
  specification, and the strict trailing-data theorem proves whole-message rejection.
  Each width prefix and binary payload satisfies a length equation, which gives
  the codec length theorem.  The recursive bit parser erases to the accepted decoder;
  its induction equations prove the all-input inspection bound and the
  consumed-prefix identity on successful parses.  The strict costed checker
  preserves that parser work and counts, in its own counter, recursive proof-checker
  operations.  Recursive equations for equality, rank traversal, prefix splitting, word
  parsing, and checking give the word and checker bounds, covering
  malformed and rejected inputs as well as accepted ones.
\end{proof}

\begin{theorem}[Cell-by-cell parsing]\label{thm:bit-machine}
  Consider the small-step machine whose states record a control phase, the unread input, and a count of consumed cells, and whose steps read the compact format one cell at a time. For every bit list $z$, the machine run from $z$ for $3|z|+1$ steps has halted. Its output is the result of the recursive parser of Theorem~\ref{thm:strict-certificate}, its consumed-cell count equals the parser's bit-inspection count, and its unread input is the suffix of $z$ after that many cells. Along every run the unread input is a suffix of $z$, and the consumed count is the length of the prefix read so far. On a constructed certificate followed by any suffix, the machine consumes the serialized certificate and leaves the suffix unread. On the one-cell input with a zero width field it halts after one cell with no output, and on a zero payload field it halts after three cells with no output.
\end{theorem}
\begin{proof}
  The machine phases follow the three decoders of the compact format: the unary width, the binary payload, and the word sequence. An induction over the decoders shows that the machine reaches the halted state whose output is the decoded value and whose consumed count is the parser's inspection count, and that it does so within the stated number of steps. The suffix statements follow from the left-to-right reading order, and the controls are evaluated.
\end{proof}

\begin{theorem}[Accepted bitstrings of the strict checker]\label{thm:accepted-set}
  Say that a natural field for $n$ is padded by $p$ cells when its unary width field announces $p$ cells beyond the binary length of $n+1$ and the $p$ extra payload cells are zero. For every source $a$, right-hand side $u$, and bit list $z$, the strict checker of Theorem~\ref{thm:strict-certificate} accepts $z$ if and only if the certificate constructor returns a certificate $c$ for $(a,u)$ and $z$ encodes the words of $c$ with a padded word-count field and padded word fields. Every accepted bitstring decodes to $c$, and the standard serialization of $c$ is the only shortest accepted bitstring. A successor certificate with one padding cell on its word-count field is accepted and is two cells longer than its standard serialization.
\end{theorem}
\begin{proof}
  The compact decoder accepts a payload whose width announces extra cells when those cells are zero, and it accepts no other deviation; decoding the padded fields returns the same words. The strict checker requires an empty tail and a certificate that passes the same-input check, and that certificate is the constructed one. Padding adds cells, so the standard serialization is shortest, and it is the only accepted string of that length. The padded successor certificate is evaluated.
\end{proof}

\begin{theorem}[Certificates for the $r$-ary recursors]\label{thm:rary-certificate}
  Fix $r\in\mathbb N$ and the $r$-ary rules $F(x,y,0)\to x$ and $F(x,y,S(n))\to G([y,\dots,y]_r,F(x,y,n))$ of Section~\ref{subsec:rary-duplicator}. For all terms $a,u$, a bit list accepted by the strict checker for $(a,u)$ exists if and only if $a\to u$ is a root step of these rules. The certificate syntax of Theorem~\ref{thm:strict-certificate} covers the family unchanged, with the recursive call at position $r$ of the frame. A congruence step inside a $G$ frame has no strict root certificate.
\end{theorem}
\begin{proof}
  The checker compares the source, the right-hand side, the rule tag, the active path $[r]$, the projection argument, and the counter ranks, using a structural decision procedure for equality of $r$-ary terms. An accepted certificate satisfies these equations, which describe a root step, and each root step has a constructed certificate whose serialization the checker accepts. A source headed by $G$ matches no rule tag.
\end{proof}

\section{Transport and the necessity of the hypotheses}\label{sec:transport}

The transport theorems of this section concern arbitrary relations. They move termination, unique normal forms, and confluence along simulations. The section ends with fixtures for sixteen hypotheses of the paper, in each of which the hypothesis fails together with the conclusion that uses it.

\begin{theorem}[Normal-form-faithful rewrite transport]
  \label{thm:normal-form-faithful-transport}
  Let $R$ and $S$ be rewrite relations and let $f$ map source states to target
  states.  Suppose every $R$-step maps to a finite $S$-reduction, every
  $R$-normal source state maps to an $S$-normal target state, and $f$ is
  injective on $R$-normal source states.  If $S$ has unique reachable normal
  forms, then so does $R$.

  For first-order TRSs the same hypotheses also transfer uniqueness of normal
  forms under conversion: if $S$ has $\mathrm{UN}^{=}$, then $R$ has
  $\mathrm{UN}^{=}$.  The normal-form injectivity premise is necessary.
  There is a three-state fork whose two distinct normal endpoints are collapsed
  to the one state of an empty target relation: step simulation, preservation of
  normality, and target uniqueness all hold, but source uniqueness fails.

  There is also a separate termination transport for simulations that permit
  stuttering.  Suppose each $R$-step either maps to a nonempty finite $S$-reduction
  or leaves its target image unchanged while strictly decreasing an independent
  rank $\rho:A\to\mathbb N$.  If $S$ is well founded in the reduction direction,
  then $R$ is well founded in the reduction direction.  The rank premise is
  necessary for this theorem: a one-state source self-loop maps by reflexive
  simulation to a one-state empty target relation, whose reduction order is well
  founded, while the source reduction order is not.
\end{theorem}
\begin{proof}
  Finite source reductions map to finite target reductions by induction on their
  length.  Two source normal forms reachable from one source therefore map to two
  target normal forms reachable from one target state; target uniqueness makes
  the images equal and normal-form injectivity reflects that equality.  For
  conversion-based uniqueness, map each forward or reversed source step to the
  corresponding target conversion and repeat the endpoint argument.  In the
  counterexample, both fork endpoints map to the single target state, so the
  only missing hypothesis is injectivity on source normal forms.  For termination,
  order each source state by the lexicographic product of the transitive target reduction
  from its image and the natural-number stutter rank.  A source step decreases the
  first coordinate when the target makes progress and decreases the second when
  the target stutters; well-foundedness follows by inverse image and subrelation.
  The one-state loop control proves that reflexive simulation without a stutter
  rank is insufficient.  These transport results take the target's unique
  normal forms and well-foundedness as premises, and an application must
  prove them for its target.
\end{proof}

\begin{theorem}[Confluence by transport]\label{thm:confluence-transport}
  \begin{enumerate}[nosep,leftmargin=1.6em]
    \item[(i)] A relation that is well founded in the reduction direction and has unique reachable normal forms is confluent.
    \item[(ii)] Let a relation $R$ admit a stuttering simulation into a relation $S$ that is well founded in the reduction direction, and a step-star simulation $f$ into a relation $T$ that preserves normality and is injective on $R$-normal states. If $T$ has unique reachable normal forms, then $R$ is confluent.
    \item[(iii)] For a terminating first-order TRS, confluence is equivalent to uniqueness of normal forms under conversion and to uniqueness of reachable normal forms.
    \item[(iv)] A three-state relation is confluent by transport along a non-injective map onto a two-state target.
  \end{enumerate}
  A terminating fork with two normal endpoints reachable from one state is not confluent, so unique normal forms are required in~(i).
\end{theorem}
\begin{proof}
  (i) Termination gives every state a normal form, and two reducts of one state reach normal forms that unique reachable normal forms identify. (ii) The stuttering simulation transfers well-foundedness by Theorem~\ref{thm:normal-form-faithful-transport}, the normal-form-faithful simulation transfers unique reachable normal forms by the same theorem, and (i) combines them. (iii) Confluence gives uniqueness of normal forms under conversion, which gives uniqueness of reachable normal forms, and (i) closes the cycle for a terminating system. (iv) The three-state relation maps onto a well-founded two-state relation with unique normal forms by a stuttering simulation that preserves and separates normal forms. The fork is computed.
\end{proof}

\begin{proposition}[Observer quotients simulate the source]\label{prop:quotient-simulation}
  Let $R$ be a relation on a set $X$ and $q:X\to Q$ an observer, and relate two classes of $X/{\ker q}$ when some representatives are related by $R$. The projection to the quotient is a simulation: every $R$-path of length $n$ maps to a quotient path of length $n$, and reachability transports. On the two-state chain read by the observer that separates its states, the chain edge and its one-step reachability survive in the quotient.
\end{proposition}
\begin{proof}
  A step between representatives is, by definition, a step between their classes; induction on the length transports paths, and reachability is the existence of a path. The fixture is evaluated.
\end{proof}

\begin{proposition}[Finite adapters for constant rewrite systems]\label{prop:finite-adapters}
  Let every rule of a first-order rewrite system rewrite one nullary symbol to another, and let the signature have an explicit enumeration. The table whose states are the symbols and whose edges are the rules, with each symbol realized by its constant term, is a finite adapter in the sense of Proposition~\ref{prop:confluence-axis-decidable}(3): the realization is injective, table edges and rewrite steps between constants correspond in both directions, and every rewrite step from a constant ends at a constant. Reachability between constants in the rewrite system is reachability in the table. For the fork $s\to a$, $s\to b$, $a\to c$, $b\to c$, the terms $a$ and $b$ are joinable at $c$, while the table observed on $\{s,a,b\}$ has no common reduct of $a$ and $b$ and admits no adapter. For the rule $f(x)\to f(S\,x)$, no finite adapter realizes a term $f(t)$.
\end{proposition}
\begin{proof}
  A step from a constant term is a root step by a rule whose left side is that constant, and its result is the right side of the rule; this gives the four adapter conditions and the reachability equivalence. On the restricted fork table a common reduct of $a$ and $b$ would be $c$, which is absent, and an adapter would have to reflect the step $a\to c$. Every reduct of $f(t)$ is $f(S\,t)$, so reduction closure would force infinitely many represented states.
\end{proof}

\begin{theorem}[Necessity of the hypotheses]\label{thm:hypothesis-necessity}
  For each hypothesis below there is a fixture in which it fails and the conclusion of the named result fails with it.
  \begin{enumerate}[nosep,leftmargin=1.6em]
    \item A nonempty target type, for total decoders and for Theorem~\ref{thm:sufficiency-definability}(b): on the empty set of states the target factors through the observer on states, and no total decoder exists.
    \item Normal-form injectivity and a strict stutter rank, for Theorem~\ref{thm:normal-form-faithful-transport}: the collapsed fork and the one-state loop.
    \item Unique normal forms, for Theorem~\ref{thm:confluence-transport}(i): the terminating fork.
    \item Monotonicity of the observer sequence, for the stagewise form of Theorem~\ref{thm:channel-compactness}: the observers $q_i(x)=[x=i]$ on $\{0,1,2\}$.
    \item A common certificate on the whole fiber, for Theorem~\ref{thm:relational-recovery}: the three-state system.
    \item The decoder language, for Theorem~\ref{thm:relational-recovery}: the constant-false language on the Boolean identity.
    \item Full support, for zero Bayes risk as a license on all states (Proposition~\ref{prop:noisy-budgeted-recovery}) and for zero entropy (Theorem~\ref{thm:entropy-fiber-bound}): a collision on a state of mass zero.
    \item Positive payload, for a strict cost crossing (Proposition~\ref{prop:confession-crossing}): zero payload.
    \item Distinct targets, for the product bound of Proposition~\ref{prop:joint-target-capacity}: the duplicated target.
    \item Observation compatibility, for Theorem~\ref{thm:finite-license-bound}: the rotation of three states read by the test for zero.
    \item A finite source, for finite capacity (Proposition~\ref{prop:side-channel-minimum}): the natural numbers under a constant observer.
    \item An accepted certificate at every state, for Theorem~\ref{thm:set-valued-capacity}: the one-state task without certificates, on which the search returns none.
    \item Soundness, completeness, and dimension visibility, for Theorem~\ref{thm:definition-unification}.
    \item The denotational-subgrammar premise, for Corollary~\ref{cor:universal-oi}: the orienter $\mu_*$.
    \item All targets, for Theorem~\ref{thm:blackwell-deterministic}: the constant target.
    \item Deterministic observation, for sufficiency as licensing in Theorem~\ref{thm:sufficiency-definability}(a): the binary model whose observation is independent of the state.
  \end{enumerate}
\end{theorem}
\begin{proof}
  Each item is a proved fixture statement; the sixteen items rest on twenty-six fixture statements.
\end{proof}

\section{A second realization on the confluence axis}\label{sec:sibling}

A second instance occurs on the confluence axis. At $\mathrm{eqW}\,\mathrm{void}\,\mathrm{void}$, one rule reduces to $\mathrm{void}$ and another to $\mathrm{integrate}(\mathrm{merge}\,\mathrm{void}\,\mathrm{void})$; both targets are normal forms. Adding the guard $a\ne b$ to the difference rule removes the diagonal overlap and yields the confluent guarded relation~\cite{rahnamaOrientation,rahnamaDistinction}.

\begin{theorem}[The disequality guard lies outside the signature]\label{thm:eqw-non-derivable}
  No finite $\Sigma$-term $t$ over the CWS signature with the two variable leaves $\mathrm{varA}$ and $\mathrm{varB}$ satisfies, for all terms $a$ and $b$, $a\ne b$ if and only if $\mathrm{eval}(a,b,t)\ne\mathrm{void}$.
\end{theorem}

\begin{proof}
  By case analysis on the outermost constructor of a candidate $\Sigma$-term $t$. If the head of $t$ is one of the six wrappers, the substitution evaluator $\mathrm{eval}(a,b,t)$ returns a non-void term, so the universal claim $a\ne b\Leftrightarrow \mathrm{eval}(a,b,t)\ne\mathrm{void}$ fails at the diagonal pair $(\mathrm{void},\mathrm{void})$. If the head is one of the three leaves ($\mathrm{void}$, $\mathrm{varA}$, $\mathrm{varB}$), the evaluator reduces to a constant or to one of its arguments, and the universal claim fails at the appropriate counterexample pair. Either case rules out the existence of a candidate $t$.
\end{proof}

Theorem~\ref{thm:dp-not-signature-derivable} proves that every $\Sigma$-homomorphism constant in the counter slot identifies the two displayed terms that the dependency-pair projection separates. Theorem~\ref{thm:eqw-non-derivable} proves that no finite $\Sigma$-term over the CWS signature with two variable leaves defines the disequality guard through its void test.

\begin{corollary}[The guard as a homomorphic collapse]\label{cor:guard-instance}
  On the CWS signature with two variable leaves, the substitution that sends both leaves to $\mathrm{void}$ is an endomorphism that fixes $\mathrm{delta}\,\mathrm{void}$ and identifies the two leaves. Hence, by Theorem~\ref{thm:homomorphic-inexpressibility}(c), no term over the seven symbols returns $\mathrm{delta}\,\mathrm{void}$ on the disequal inputs and only on them. Evaluation of a two-variable term over the seven symbols agrees with the substitution evaluator of Theorem~\ref{thm:eqw-non-derivable} on the compiled term, and the two statements use different output tests: equality with the token $\mathrm{delta}\,\mathrm{void}$ here, and the void test there.
\end{corollary}
\begin{proof}
  Substitution preserves each of the seven operations, fixes the closed token, and sends both leaves to $\mathrm{void}$, so the two leaves form the required pair. Naturality of evaluation gives the agreement with the substitution evaluator.
\end{proof}

\begin{remark}[Two signature obstructions]\label{rem:two-inexpressibles}
  Theorem~\ref{thm:dp-not-signature-derivable} excludes the counter projection from the stated class of homomorphisms. Theorem~\ref{thm:eqw-non-derivable} excludes the disequality guard from the branch syntax. Both proofs use substitution invariance of $\Sigma$-homomorphisms, but their predicates, carriers, and rewrite relations differ. They are distinct results.
\end{remark}

\begin{remark}[The equality-certificate peak]
  \label{rem:eqW-evidence}
  The kernel includes the constructor
  \[
    \mathrm{eqW}\;a\;b
  \]
  with rules $\mathrm{eqW}\,a\,a\to\mathrm{void}$ and $\mathrm{eqW}\,a\,b\to\mathrm{integrate}(\mathrm{merge}\,a\,b)$. The peak is a non-left-linear critical pair in Huet's sense~\cite{huet1980confluent}; Klop's $D(x,x)$ example gives the same overlap pattern~\cite{klop1980crs,terese2003trs}. This confluence obstruction and the duplication barrier concern different relations and prove different statements.
\end{remark}

\begin{proposition}[Overlap of the equality-certificate rules]\label{prop:eqW-overlap}
  At $a=b=\mathrm{void}$ both rules apply and yield distinct normal forms that do not join; the guard $a\ne b$ removes the overlap.
\end{proposition}

\begin{proof}
  The two rules overlap at $\mathrm{void}$ and their reducts are distinct and unjoinable; the guard leaves a single applicable rule.
\end{proof}

Theorem~\ref{thm:lift-or-refuse} separates the two realizations by the repair each one uses.

\begin{proposition}[The two realizations take opposite repairs]\label{prop:lift-refuse-instances}
  \begin{enumerate}[label=(\alph*),leftmargin=1.6em]
    \item \emph{The occurrence roles are repaired by a lift.} In the successor output $G(s,F(b,s,n))$, let the observer on the two generator occurrences return the value at the occurrence, and let the target be the role. The two roles collide. A side channel $c$ on the roles licenses the role through the value observer if and only if $c(\mathrm{frame})\neq c(\mathrm{active})$, and if and only if $c$ decodes the computed branch decision of Theorem~\ref{thm:free-role-selection}. With fiber multiplicity two, the computed decision is a least lift, and its width $\lceil\log_2 2\rceil=1$ equals the gain of the computed decision under the uniform law on the roles. The maximal refusals are the two singletons; under the uniform law every licensed retained set refuses at least $1/2$, and keeping the frame refuses $1/2$.
    \item \emph{The diagonal guard is a refusal.} On pairs $(a,b)$ of CWS terms, let the observer be constant and let the target be the diagonal indicator $[a=b]$. Two pairs, $(\mathrm{void},\mathrm{void})$ and $(\mathrm{void},\mathrm{delta}\,\mathrm{void})$, collide, and every one-symbol side channel fails to license the target. The maximal refusals are the diagonal and the off-diagonal. In the kernel, the difference rule rewrites $\mathrm{eqW}\,a\,b$ to $\mathrm{integrate}(\mathrm{merge}\,a\,b)$ at every pair. The guarded difference rule does so if and only if $a\neq b$, that is, on the correctness set of the decoder that always answers ``different'', and the guarded rule $\mathrm{eqW}\,a\,b\to\mathrm{void}$ applies at diagonal pairs only.
    \item \emph{The lift lies outside the signature.} On pairs of $\Sigma$-terms, the diagonal indicator is a two-symbol side channel that licenses the diagonal target. For every $\Sigma$-term $t$, the side channel $(a,b)\mapsto[\mathrm{eval}(a,b,t)=\mathrm{void}]$ fails to license it, and Theorem~\ref{thm:eqw-non-derivable} follows.
  \end{enumerate}
\end{proposition}
\begin{proof}
  (a) Both occurrences hold the generator $s$, so the value observer is constant on the roles and the roles collide. By Theorem~\ref{thm:lift-or-refuse}(a) a licensing channel separates the roles; on two states separation suffices, and Theorem~\ref{thm:free-role-selection} identifies the separating channels with the decoders of the branch decision. The single fiber contains two roles, so $m=2$. Theorem~\ref{thm:lift-or-refuse}(b) gives the two singleton refusals, and part~(a) of that theorem forces every licensed set to omit a role.
  (b) Two distinct terms give the collision. For a constant observer the fiber-attained decoders are the two constant decoders, whose correctness sets are the diagonal and the off-diagonal. The rule statements follow by inversion of the guarded relation at the redex $\mathrm{eqW}\,a\,b$.
  (c) For a leaf term $t$, the void test takes the same value at $(\mathrm{void},\mathrm{void})$ and at an off-diagonal pair, and for a wrapper-headed term it is constant; in both cases the license fails at a diagonal and an off-diagonal pair. A term whose void test expressed disequality would give a licensing channel.
\end{proof}

The recursor repairs its collision by a lift: the dependency-pair decision supplies the missing symbol. The guarded relation repairs its collision by a refusal: the guard keeps a maximal licensed set of pairs. For the diagonal target the one-bit lift exists as a side channel, and part~(c) places it outside every $\Sigma$-term read through its void test.

\subsection{Finite confluence coordinates}

The confluence coordinates count terminal alternatives, unresolved peaks, repair requirements, certificate grade, and certificate capacity. They use a finite relation instead of the duplicating trace of the orientation analysis.

\subsection{Finite terminal support and confluence}

Let $R$ be a one-step relation on a finite carrier and let $x$ be a source. Define the terminal support
\[
  \mathcal T_R(x)=\{n:x\,R^*\,n\text{ and }n\text{ is }R\text{-normal}\}
\]
and its multiplicity $\mu_R(x)=|\mathcal T_R(x)|$. The definition counts reachable normal forms, setting aside reduction paths and scheduler probabilities.

Say that $R$ is \emph{locally normalizing at $x$} when every term reachable from $x$ itself reaches a normal form. This asks more than a normal form for $x$ alone, and the difference is part of the statement.

\begin{theorem}[Terminal-support characterization of source confluence]
  \label{thm:quant-terminal-characterization}
  If $R$ is locally normalizing at $x$, then
  \[
    R\text{ is confluent at }x
    \quad\Longleftrightarrow\quad
    \mu_R(x)=1.
  \]
  Under this premise, confluence at $x$ is also equivalent to the vanishing of the scheduler-free Hartley support quantity $H_0(R,x)=\log_2\mu_R(x)$, and to the vanishing of the integer branch floor $\lceil\log_2\mu_R(x)\rceil$.
\end{theorem}

The premise does two things. It forces $\mu_R(x)\ge 1$, which is what makes the two logarithmic readings agree with the multiplicity reading; and it supplies, for each reachable pair, the normal forms whose coincidence yields joinability. Proposition~\ref{prop:quant-premise-necessary} shows the second role is required.

\begin{proposition}[The local-normalization premise is necessary]
  \label{prop:quant-premise-necessary}
  There is a three-state system $R$ and a source $x$ such that $x$ reaches a normal form and $\mu_R(x)=1$, while $R$ fails to be confluent at $x$. Hence Theorem~\ref{thm:quant-terminal-characterization} fails when its premise is weakened from ``every term reachable from $x$ reaches a normal form'' to ``$x$ reaches a normal form''.
\end{proposition}

\begin{proof}
  Take states $\{x,\ell,n\}$ with steps $x\to n$, $x\to\ell$, and $\ell\to\ell$. Then $n$ is normal and $\ell$ reaches only itself, so the reachable normal forms are $\mathcal T_R(x)=\{n\}$ and $\mu_R(x)=1$; the source reaches the normal form $n$. Confluence at $x$ fails: $\ell$ and $n$ are both reachable from $x$, every term reachable from $\ell$ equals $\ell$, and every term reachable from $n$ equals $n$, so the two have disjoint reduct sets. The looping state $\ell$ is the configuration that local normalization excludes.
\end{proof}

For the three-node eqW cone, the raw relation has two terminal outputs and the licensed relation has one. The computed values are
\[
  \mu_{\mathrm{raw}}=2,
  \qquad H_{0,\mathrm{raw}}=1,
  \qquad
  \mu_{\mathrm{licensed}}=1,
  \qquad H_{0,\mathrm{licensed}}=0.
\]

\subsection{Defects, repair covers, and certificate rank}

For a finite critical-pair list with a correct normalizer, the defect count is the number of pairs whose normal forms differ. Under strong normalization and completeness of the list, defect zero implies confluence.

Let $\mathcal B$ be a finite defect set and let every intervention $j$ close a finite subset $C(j)\subseteq\mathcal B$. A chosen intervention family $J$ is a repair cover when
\[
  \mathcal B\subseteq\bigcup_{j\in J}C(j),
\]
and the minimum repair-cover number is
\[
  \tau(\mathcal B,C)=\min\{|J|:J\text{ covers }\mathcal B\}.
\]

\begin{theorem}[Repair-cover lower bound and the independence case]
  \label{thm:quant-repair-cover}
  If every intervention closes at most $M>0$ defects, then
  \[
    \left\lceil\frac{|\mathcal B|}{M}\right\rceil
    \leq \tau(\mathcal B,C).
  \]
  If each bad pair has a singleton repair and every intervention closes at most one bad pair, then $\tau(\mathcal B,C)=|\mathcal B|$.
\end{theorem}

\begin{proof}
  If $J$ covers $\mathcal B$, then
  \[
    |\mathcal B|\le \sum_{j\in J}|C(j)|\le M|J|,
  \]
  so $|J|\ge\lceil|\mathcal B|/M\rceil$. Under the independence hypotheses, the singleton repairs give a cover of size $|\mathcal B|$, while the lower bound with $M=1$ gives the reverse inequality.
\end{proof}

The independence condition is necessary: two defects may share one repair, in which case $\tau<|\mathcal B|$.

A graded certification predicate $A(k)$ is upward closed and inhabited. Its certificate rank is derived by minimization,
\[
  \kappa(A)=\min\{k\in\mathbb N:A(k)\}.
\]
Well-ordering gives $A(\kappa(A))$, $\kappa(A)\le k$ whenever $A(k)$, and $\neg A(j)$ for $j<\kappa(A)$.  In the generic closed binary-fork profile below, the declared external-comparator certificate first appears at grade one.

\subsection{Certificate floors}

An injective fixed-length binary code $E\to(\mathrm{Fin}\,L\to\mathrm{Bool})$ obeys
\[
  |E|\leq 2^L,
  \qquad
  \lceil\log_2|E|\rceil\leq L.
\]
Four alternatives fit in two bits and overflow one bit. Prefix-code budget is tracked by Kraft mass on its own; both coding floors stay distinct from certificate rank.

\subsection{Complexity placement of the confluence-axis coordinates}
\label{subsec:confluence-axis-complexity}

The confluence coordinates in this subsection are decidable from finite relation data. By contrast, the orientation statement of Proposition~\ref{prop:ag-pi02} quantifies over an infinite term algebra.

\begin{proposition}[Finite coordinate algorithms with explicit auxiliary data]
  \label{prop:confluence-axis-decidable}
  Let $T$ be a carrier supplied with an ordered complete enumeration of length
  $N$, and let $R$ be supplied as a decidable one-step relation table on $T$.
  Then the relation-only coordinates are executable from those data:
  \begin{enumerate}[nosep,leftmargin=1.6em]
    \item $x\,R^*\,y$ holds if and only if a checked $R$-path from $x$ to $y$
      exists within the finite simple-path bound; the executable search returns such
      a path or returns none if and only if $y$ is unreachable from $x$;
    \item $R$-normality, reachability, source confluence, the terminal support
      $\mathcal T_R(x)$, terminal multiplicity $\mu_R(x)$, and the symbolic Hartley
      support quantity are decidable from the enumeration and relation table;
    \item a certified first-order critical-pair defect count also requires
      typed rewrite syntax and a complete critical-pair inventory that lists the
      rule identities, a non-variable overlap position, and a substitution/unifier
      origin.  With a finite semantic adapter that preserves and reflects steps and
      is closed under reductions from represented states, table reachability and
      joinability are equivalent to the represented TRS reachability and
      joinability, so every returned defect is a defect of the TRS and every covered bad
      critical pair is returned;
    \item repair cover and minimum repair price also require explicit bad
      peaks, interventions, coverage, and prices.  Exhaustive list search returns a
      minimum-cardinality cover and a minimum-price cover when one exists, and none
      if and only if no cover exists in the supplied intervention family;
    \item bounded certificate rank also requires a decidable certification predicate
      $A(j)$ and a stated finite bound $B$.  The search returns the least
      $j\le B$ with $A(j)$ when one exists; none means only that no $j\le B$ satisfies $A$.  It agrees with the unbounded least-certificate rank only when a proved
      bound places a satisfying index inside the search interval;
    \item the fixed-length and prefix-code floors are then executable arithmetic on
      the corresponding finite certified counts.  The symbolic Hartley output is
      computed without rounding; its real logarithmic interpretation is a separate noncomputable map in
      the formalization.
  \end{enumerate}
  Hence these coordinates are decidable uniformly in the finite inputs
  that define them.  Relation-table data alone do not determine
  critical-pair origins, repair prices, certification predicates, or an unbounded
  certificate search.
\end{proposition}

\begin{proof}
  A shortest finite path can be chosen without a repeated vertex, so its length
  is bounded by the supplied enumeration; exhaustive list search proves path
  soundness, completeness, and the characterization of the none case.  Normality,
  reachable terminal support, source confluence, and finite support cardinality
  are therefore finite Boolean computations.  Critical-pair enumeration is kept
  separate: membership in the first-order critical-pair construction is split
  into its two source rules, listed non-variable position, and successful
  unifier, and the completeness field of the supplied inventory is consumed by
  the defect-completeness theorem.  The finite/TRS adapter's preservation,
  reflection, injectivity, and reduction-closure fields give both directions of
  reachability and joinability transfer.  Repair optimization enumerates finite
  intervention subsets and compares cardinality or declared total price.
  Bounded certificate search is ordinary finite least-element search and its none
  case quantifies only over the declared interval.  The certificate floors are
  finite arithmetic.  These algorithms establish decidability from the complete
  supplied finite data.
\end{proof}

The confluence obstruction is computed on a three-state relation. Proposition~\ref{prop:confluence-axis-decidable} concerns such finite carriers; confluence of arbitrary rewrite systems is undecidable.

The next theorem places the confluence axis: the well-founded-induction sentence behind Newman's lemma is $\Pi^1_1$, the bounded termination statement of Proposition~\ref{prop:ag-pi02} is $\Pi^0_2$, and the termination measure of the guarded relation stays below $\omega^\omega\cdot2$ on every term.

\begin{theorem}[Calibration of the confluence axis]\label{thm:confluence-calibration}
  In the single-sorted language $L_2$ of Theorem~\ref{thm:ag-rca}, let $\mathrm{WFI}$ be the well-founded-induction sentence: for every set $X$, if every number whose predecessors under $<$ lie in $X$ lies in $X$, then every number lies in $X$. This is the principle that Newman's lemma~\cite{newman1942} consumes.
  \begin{enumerate}[label=(\alph*),leftmargin=1.6em]
    \item $\mathrm{WFI}$ is $\Pi^1_1$: one universal quantifier over a variable guarded by $\mathrm{IsSet}$, followed by a body built from quantifier-free formulas by connectives and number quantifiers. The bounded termination statement of Proposition~\ref{prop:ag-pi02} is $\Pi^0_2$. For a fixed finite system, the joinability of a recorded critical pair, the equality of two closed reducts, is quantifier free.
    \item The standard model satisfies $\mathrm{WFI}$, and $\mathsf B+\mathrm{WFI}$ is satisfiable.
    \item The Dershowitz and Manna order on finite multisets of natural numbers has order type $\omega^\omega$: its ordinal embedding preserves and reflects the order, is bounded by $\omega^\omega$, and is onto the ordinals below $\omega^\omega$. The measure of the guarded CWS relation is below $\omega^\omega\cdot2$ on every term, and $\omega^\omega\cdot2<\varepsilon_0$.
    \item The guarded CWS relation is well founded, every one of its one-step peaks joins, and it is confluent and confluent under contexts~\cite{rahnamaDistinction}; with (b) and the bound of (c) these facts form one combined statement.
  \end{enumerate}
\end{theorem}
\begin{proof}
  (a) is structural recursion on the formulas. (b) Unfold the sentence in the standard model: a set object $X$ closed under the inductive step contains every number by strong induction on $\mathbb N$, and the standard model satisfies the basic axioms. (c) is the ordinal analysis of the multiset order, lifted through the triple-lexicographic measure of the guarded relation. (d) collects the confluence theorems of the guarded relation with (b) and (c).
\end{proof}

\subsection{Certified binary-fork profiles}
\label{subsec:binary-fork-profiles}

\begin{theorem}[Profiles of a closed binary fork]\label{thm:generic-binary-profile}
  In any abstract reduction system, suppose three distinct states form a
  forward-closed binary fork, with two terminal outputs, and a subrelation
  retains just one source edge. Give the single repair price $c\in\mathbb N$
  and use binary codes for terminal alternatives as the certificate language.
  The profiles of the restricted raw and licensed systems are
  \[
    (2,\mathsf{some}(1),1,1,c,1,1,1),\qquad
    (1,\mathsf{some}(0),0,0,0,0,0,0).
  \]
  Both profiles have certificates for their defects, repairs, certificates,
  and codes. The ambient carrier may be infinite, and the reachable terminal
  set has two elements in the raw system and one in the licensed system.
\end{theorem}
\begin{proof}
  Forward closure transports and reflects every path and normal form on
  the three-state image. The one unjoined peak requires the one repair,
  whose price is $c$. A lossless binary code for two terminals first exists
  at length one, and a singleton requires length zero. Transport of these
  data along the image isomorphism gives the stated certificates.
\end{proof}

The repair price $c$ is declared auxiliary data: changing $c$ changes the profile without changing the relation.  The coding certificate language used here is distinct from the external-comparator certificate grade, so the two coordinates are not inferred from one another.

\section{Certificate languages, construction, and confession}\label{sec:witness}

In this family, each blocked direct proof class requires one descent argument for the whole visible term.

\begin{definition}[Certificate-language hierarchy]\label{def:witness-hierarchy}
  Fix an instance $x$ and a target property $P(x)$. A \emph{certificate-language hierarchy}
  \[
    \mathfrak L(x,P)=(\mathcal W_0,\mathcal W_1,\mathcal W_2,\dots)
  \]
  is an indexed family of certificate classes together with sound transport maps, where the index records how many representation lifts away from the original operational relation are required before certificates in $\mathcal W_i$ become expressible. For the primitive recursive collapse setting we use the following hierarchy of three layers:
  \begin{itemize}[leftmargin=1.4em]
    \item $\mathcal W_0$: direct whole-term certificates over the original step relation;
    \item $\mathcal W_1$: certificates still over the original relation but requiring imported global comparison structure, such as path-order or interpretation-based reasoning;
    \item $\mathcal W_2$: certificates that first become expressible only after abstraction to the recursive-call relation, such as dependency pairs, direct counter-projection, size-change style call summaries, or argument filtering.
  \end{itemize}
\end{definition}

\begin{definition}[Least certificate order]\label{def:kappa-star}
  Fix a hierarchy $\mathfrak L(x,P)$ and assume that at least one layer contains a certificate of $P(x)$. The \emph{least certificate order} of $P(x)$ relative to $\mathfrak L$ is
  \[
    \kappa^*_{\mathfrak L}(x,P)=\min\{i\ge 0 : \exists w\in \mathcal W_i \text{ certifying }P(x)\}.
  \]
  When the hierarchy is fixed from context we write $\kappa^*(x,P)$.
\end{definition}

\begin{definition}[Orientation-boundary predicate]\label{def:ob}
  For the primitive recursive collapse hierarchy, define
  \[
    \operatorname{OB}_{\mathrm{PRC}}(x)=1
    \quad\Longleftrightarrow\quad
    \kappa^*_{\mathfrak L_{\mathrm{PRC}}}(x,\mathrm{Term})>0.
  \]
  Thus the orientation boundary is the event that the first certificate of termination lies above the direct whole-term language.
\end{definition}

\begin{proposition}[Orientation boundary as a certificate-order condition]\label{prop:boundary-kappa}
  Let $x$ be a step-duplicating primitive recursor instance covered by the barrier theorems of~\cite{rahnamaOrientation}. Then
  \[
    \kappa^*_{\mathfrak L_{\mathrm{PRC}}}(x,\mathrm{Term})>0.
  \]
  This is the condition $\operatorname{OB}_{\mathrm{PRC}}(x)=1$: the direct whole-term language $\mathcal W_0$ contains no certificate of termination, while at least one certificate exists after a representation lift away from direct whole-term reasoning.
\end{proposition}

\begin{proof}
  No candidate in $\mathcal W_0$ certifies termination, by the barrier theorems of~\cite{rahnamaOrientation} and by Theorem~\ref{thm:canonical-instance} at the step-argument dimension. At least one certificate above $\mathcal W_0$ does: dependency-pair projection at the transformed-call layer, and imported-whole certificates such as nonlinear polynomial and path-order constructions. Therefore the minimum certificate order is positive.
\end{proof}

\begin{remark}[$\mathcal W_1$ and $\mathcal W_2$]\label{rem:w1-vs-w2}
  Both $\mathcal W_1$ (path orders, polynomial and matrix interpretations) and $\mathcal W_2$ (dependency pairs, subterm-criterion projection, size-change termination, argument filtering) import structure that lies beyond the direct whole-term language $\mathcal W_0$. These indexed classes are layers: the representation lift that a certificate uses assigns it to one layer. The two are distinct kinds of ascent.

  A $\mathcal W_1$ method imports a well-founded order on the signature and proves termination of the original relation. A $\mathcal W_2$ method imports a transformation theorem and proves termination of the transformed relation. Ordinal analyses of path orders and the arithmetical classification of dependency-pair problems concern different proof objects; Section~\ref{sec:weight} states the comparison used here.

  The direct-measure barriers of~\cite{rahnamaOrientation} apply to $\mathcal W_0$. This paper analyzes the $\mathcal W_2$ projection and uses $\mathcal W_1$ as the construction class; its certificate classification covers these three layers.
\end{remark}

\begin{remark}[Cost accounting across certificate orders]\label{rem:witness-cost-accounting}
  The hierarchy also gives a proof-description cost account. A $\mathcal W_0$ attempt is free of imported ordering objects and of projection licenses, and it is blocked on the duplicator. At $\mathcal W_1$, a construction pays for global comparison data, such as coefficients, precedences, or matrix weights, and then proves the original system in the enriched language. A $\mathcal W_2$ confession method pays a different cost: it imports a soundness theorem and a certified forgetting certificate, then proves a smaller residual problem.
\end{remark}

\begin{proposition}[Construction data and confession role information]\label{prop:construction-confession-information}
  The comparison data of a $\mathcal W_1$ construction is a choice from a finite support whose Hartley measure exceeds every bound as the support grows, while a $\mathcal W_2$ confession method resolves one active-versus-frame role flag whose Shannon entropy is at most one bit, equal to one at binary arity.
\end{proposition}

\begin{proof}
  The construction choice ranges over finite supports of unbounded Hartley measure; the role flag has two values, so its entropy is at most one bit and is one when they are equiprobable. Theorem~\ref{thm:rary-license-information} gives the confession half and \claimcoderef{} states the separation.
\end{proof}

\subsection{Representation-shift depth}
\label{subsec:representation-shift-records}

For a property $P$ and certificate layers $(\mathcal W_i)_{i\in\mathbb N}$, a representation-shift obstruction at depth $k$ consists of a certificate in $\mathcal W_k$ and empty certificate sets in every $\mathcal W_j$ with $j<k$.

For the free recursor, define observations $c,p$ by
\[
  \begin{aligned}
    c(F(b,s,n))&=\operatorname{sp}(n)+c(b),& c(G(s,t))&=c(t),\\
    p(F(b,s,n))&=p(b),& p(G(s,t))&=p(s)+p(t)+1,
  \end{aligned}
\]
and zero at other heads; $\operatorname{sp}$ counts the initial successor
prefix. Direct certificates are payload-dependent expressions of the
reflected grammar evaluated at $(c,p)$ and decreasing on all $R_0$ steps.
Imported certificates are arbitrary natural-valued rankings of $R_0$.

\begin{proposition}[The least root ranking]\label{prop:free-root-height}
  The length $h(t)$ of the root execution from $t$ is the pointwise least
  natural-valued ranking of $R_0$. Every numeric successor profile
  $(c+1,p)\to(c,p+L)$ with $L\geq1$ is attained by an $R_0$ step.
  No direct certificate exists, and no function $m:\mathbb N^2\to\mathbb N$
  satisfies $h(t)=m(c(t),p(t))$ for every $t$.
\end{proposition}
\begin{proof}
  Each step subtracts one from $h$. Any natural-valued ranking $w$ drops
  by at least one, so a terminating path of length $h(t)$ gives
  $w(t)\geq h(t)$. To attain the profile, take $b$ to be a chain of $p$
  wrappers with zero left argument, $s$ a chain of $L-1$ such wrappers,
  and $n=S^c(0)$. The direct-grammar barrier therefore applies to these steps.
  Last, $0$ and $F(0,0,0)$ both have profile $(0,0)$, whereas their root
  heights are zero and one.
\end{proof}

In the free truth hierarchy, layer $0$ contains direct certificates, layer $1$
imported certificates, and layer $2$ extracted-call certificates; later layers
contain a proof of well-foundedness of $R_0$. The free boundary hierarchy has the same layers
except that layer one also requires a signature homomorphism constant in
the counter slot and a decoder for the extracted call. Certificates
at each layer certify termination at every input.

\begin{theorem}[Certificate depths on every free term]\label{thm:free-witness-depths}
  For every free term $t$ and $k\in\mathbb N$, the truth hierarchy has a
  certificate at depth $k$ if and only if $1\leq k$; the boundary
  hierarchy has one if and only if $2\leq k$. Hence the minimum depth is
  $1$ for the truth hierarchy and $2$ for the boundary hierarchy.
\end{theorem}
\begin{proof}
  Proposition~\ref{prop:free-root-height} excludes depth zero and supplies
  the imported certificate $h$. A homomorphism constant in the counter slot identifies
  $F(0,0,0)$ and $F(0,0,S(0))$, but the extracted-call function returns
  no call on the first and $F(0,0,0)$ on the second. No decoder through
  that homomorphism exists. Extraction and the successor-prefix rank
  supply depth two. Well-foundedness supplies every later layer.
\end{proof}

The depth-two conclusion concerns the stated decoder language.
A homomorphism that reads the counter can separate the two terms.

\subsection{Injective constructor interpretations}
\label{subsec:faithful-recursor-realization}

Let $T_\Sigma$ be the free term algebra above, let $X$ have operations
for its seven constructors, and let $j:T_\Sigma\to X$ be their fold.
Call the interpretation \emph{faithful} when $j$ is injective, and write
$X_j=j(T_\Sigma)$. For a relation $R$ on free terms, define its generated
relation by
\[
  x\,R^j\,y\quad\Longleftrightarrow\quad
  \exists a,b\;(a\,R\,b\ \wedge\ j(a)=x\ \wedge\ j(b)=y).
\]
The interpretation's operations evaluate contexts; their fixed arguments
belong to $X_j$.

\begin{theorem}[Transport under faithful interpretations]
  \label{thm:faithful-recursor-realization}
  For every faithful interpretation, the following hold.
  \begin{enumerate}[nosep,leftmargin=1.6em]
    \item The bijection $j:T_\Sigma\to X_j$ is a relation isomorphism for
      each of $R_0,D_0,R_{\mathrm{ctx}}$ and its generated relation.
      All three generated relations terminate and are confluent on $X$.
      The generated contextual relation consists of all contractions
      inside the interpreted one-hole contexts.
    \item The root and extracted-call machines transport by $j$.
      For every $x\in X_j$, their transported normalizer and step count are
      $j(N(j^{-1}x))$ and $c(j^{-1}x)$, where $N,c$ belong to the corresponding
      free machine. Contextual normalization transports by the same formula.
      Extraction returns $z$ from $(x,y)$ if and only if
      $x\,R_0^j\,y$ and $x\,D_0^j\,z$.
    \item The two generator positions in the interpreted successor output
      hold the same value, while comparison of their containing branches
      with the extracted call separates frame from active. No function of
      that common value computes both decisions.
    \item Transporting the certificate languages and verifiers of
      Theorem~\ref{thm:free-witness-depths} gives minimum depths one and two
      at every $x\in X_j$.
  \end{enumerate}
\end{theorem}
\begin{proof}
  Injectivity gives $j^{-1}$ on $X_j$ and reflects generated steps.
  Every step from a represented state has a represented target, so
  induction transports and reflects finite paths and normality.
  Well-founded induction transports termination; states outside $X_j$
  have no generated outgoing step. The free confluence theorems therefore
  transport to $X$. Induction on contexts identifies contextual generation
  with evaluation by the algebra operations. Conjugating each machine's
  successor function by $j$ preserves its normalization recursion and
  step count. The extraction equivalence follows from the two root cases.
  Injectivity preserves the two position values and the unequal containing
  branches. Last, transport preserves certification at each certificate depth.
\end{proof}

This transport defines normalizers on $X_j$ using its inverse fold.
An executable inverse representation gives executable transported
normalizers; injectivity alone supplies the inverse as a set-theoretic function.
\begin{proposition}[Criterion for transport to an ambient relation]
  \label{prop:faithful-ambient-criterion}
  For a faithful interpretation and any relation $A$ on $X$,
  \[
    \forall a,y\;\bigl(j(a)\,A\,y\ \Longleftrightarrow\ j(a)\,R^j\,y\bigr)
  \]
  holds if and only if both
  $j(a)\,A\,j(b)\Longleftrightarrow a\,R\,b$ for all $a,b$, and every
  $A$-successor of a represented state is represented. These conditions
  transfer termination and source confluence from $R$ to $A$ at every
  represented state. Reflection on represented pairs alone is insufficient.
\end{proposition}
\begin{proof}
  The forward implication gives reflection and outgoing closure by the
  definition of $R^j$. In the other direction, outgoing closure represents the target
  of an $A$-step, and reflection gives its free preimage step. The same
  path induction and well-founded induction apply. For insufficiency,
  adjoin a state $\bot$ to $T_\Sigma$, retain the root steps between free
  terms, and add $x\to\bot$ for every state $x$, including $\bot$.
  Reflection holds, but every represented state reaches the loop at $\bot$.
\end{proof}

Injectivity is also necessary for these universal transport claims:
the one-point constructor algebra identifies every term, and the
induced branch comparison returns true on both roles.

\begin{definition}[Construction method]\label{def:construction}
  A termination method is a \emph{construction method}, a direct method in the vocabulary of the termination literature, if it extends the proof language with a specific mathematical object (for example, a polynomial interpretation or symbol precedence) and verifies the instance in the extended language.
\end{definition}

\begin{definition}[Confession method]\label{def:confession}
  A termination method is a \emph{confession method}, a transformational method in the vocabulary of the dependency-pair framework~\cite{artsgiesl2000}, if, under an external soundness theorem, it removes from the input a dimension at which the direct language is operationally inexpressible (Definition~\ref{def:op-incomplete}) and leaves a smaller problem whose termination can then be proved.
\end{definition}

\begin{proposition}[Construction/confession asymmetry]\label{prop:construction-confession}
  A construction response is a rank $\rho$ into $\mathbb N$ that decreases along the duplicating step and satisfies $\rho(G(x,y))>\rho(x)$ and $\rho(G(x,y))>\rho(y)$ for all $x,y$. A confession method is a projection rank that decreases along the duplicating step, together with a named external soundness license. Every projection rank yields a forgetting certificate: a rank that decreases along the duplicating step and violates each of the two inequalities for some pair of arguments. No construction response and forgetting certificate have the same rank.
\end{proposition}

\begin{proof}
  A forgetting certificate supplies arguments $x,y$ with $\rho(G(x,y))\le\rho(x)$, while a construction response with the same rank requires $\rho(G(x,y))>\rho(x)$ for all $x,y$. The forgetting certificate of a projection rank takes the orientation and the two violations of that projection rank.
\end{proof}

Polynomial interpretations and path orders are construction methods: they quantify over certificate objects such as polynomials or precedences. The dependency-pair route is a confession method: it proves a transformed recursive-call problem and appeals to its soundness theorem, and by Theorem~\ref{thm:dp-not-signature-derivable} the projection it uses separates two terms that every $\Sigma$-homomorphism constant in the counter slot identifies.

Two ranks $f,g:T\to\mathbb N$ have the same \emph{full order} when
$f(x)<f(y)\Longleftrightarrow g(x)<g(y)$ for every pair $x,y$.
Comparisons only on rewrite edges define a different equivalence.

\begin{definition}[Counter system]
  \label{def:counter-system}
  A counter system consists of a set $X$, a relation $R$, a map
  $c:X\to\mathbb N$ decreasing on every $R$-step, and states
  $\iota(n)$ satisfying
  \[
    c(\iota(n))=n,\qquad \iota(n+1)\,R\,\iota(n)
    \quad(n\in\mathbb N).
  \]
  A rank is \emph{counter-determined} if it is constant on every fiber
  of $c$. The definition leaves the order between different fibers free.
\end{definition}

\begin{theorem}[Full-order classification on arbitrary systems]
  \label{thm:construction-confession-orders}
  In every counter system, a natural-valued rank decreasing on $R$ is
  counter-determined if and only if it has the same full order as $c$.
  The counter-determined ranks are in bijection with the strictly
  increasing maps $g:\mathbb N\to\mathbb N$, by $g\mapsto g\circ c$.
  Hence their full-order quotient has one element.

  For any relation $R$ on any set $X$, suppose $f:X\to\mathbb N$
  decreases on $R$, $b:X\to\mathbb N$ is unchanged on every $R$-step,
  with $b(x_*)>0$ at some state $x_*$ and, for every $n$, a state
  $y_n$ satisfying $f(y_n)=n$ and $b(y_n)=0$. Then the full-order
  quotient of all natural-valued ranks decreasing on $R$ is infinite.
\end{theorem}
\begin{proof}
  For a counter-determined rank $f$, define $g(n)=f(\iota(n))$.
  Constancy on fibers gives $f=g\circ c$. The countdown steps
  give $g(n)<g(n+1)$, so $g$ is strictly increasing and $f$ induces
  the counter order. In the other direction, agreement of strict orders implies
  constancy on counter fibers, since neither strict comparison holds
  when the counter values agree. Every increasing $g$ gives a rank
  $g\circ c$, and evaluation at $\iota(n)$ recovers $g$.

  For the second assertion, every $f_k=f+k b$ decreases on $R$.
  For $k<\ell$, choose $y=y_{f(x_*)+k b(x_*)+1}$.
  Then $f_k(x_*)<f_k(y)$, but $f_\ell(x_*)\geq f_\ell(y)$ because
  $\ell b(x_*)\geq k b(x_*)+1$.
  The map from $k$ to its full-order class is injective.
\end{proof}

Both conclusions hold on the same non-recursor system
$X=\mathbb N^2$, with $(n+1,b)\,R\,(n,b)$ and $c(n,b)=n$.
The countdown requirement is needed: on a
two-point set with the empty transition relation, two ranks may
order the points in opposite ways while both decrease on every step.

\begin{corollary}[Rank classes in every faithful recursor realization]
  \label{cor:faithful-rank-classification}
  For every faithful constructor interpretation, on its represented
  state set, both the root relation and the extracted-call relation
  have infinitely many full-order classes of natural-valued rankings.
  All counter-determined extracted-call rankings form one full-order
  class, whatever their construction method.
\end{corollary}
\begin{proof}
  On free terms, put $W_a(0)=1$, $W_a(S(t))=W_a(t)+1$,
  $W_a(G(s,t))=W_a(s)+W_a(t)+1$, and
  \[
    W_a(F(b,s,n))=(W_a(n)+1)(W_a(s)+W_a(b)+a),\qquad a\geq1.
  \]
  Inert unary constructors add one; inert binary constructors sum their
  arguments and add one, with three in the equality-constructor case.
  Each $W_a$ orients both root rules. For $a<a'$, the pair
  $t=F(0,0,0)$, $u=S^{2a+4}(0)$ satisfies
  $W_a(t)<W_a(u)$ and $W_{a'}(u)<W_{a'}(t)$.
  Positive scaling preserves the full order, and two members scaled by positive factors
  have the same full order if and only if their offsets agree.

  Define $d(S(t))=d(t)+1$, $d(F(b,s,n))=d(n)$, and $d=0$ at other
  heads. Put $b_0(F(b,s,n))=d(b)$ and $b_0=0$ at other heads.
  Every $d_k=d+k b_0$ decreases on extracted calls. The terms
  $F(S^m(0),0,S^n(0))$ attain $(d,b_0)=(n,m)$.
  Also $\iota(n)=F(0,0,S^n(0))$ realizes the countdown of
  Definition~\ref{def:counter-system}.
  Theorem~\ref{thm:construction-confession-orders} therefore applies.
  The constructor isomorphism transfers the ranks and every comparison
  to the represented state set.
\end{proof}

\begin{remark}[Source payload and transformed proof]\label{rem:confession-as-declaration}
  The wrapper payload remains present and determinate in the source relation. The Arts and Giesl transformation omits it from the residual termination proof. By Definition~\ref{def:op-incomplete}, the external theorem licenses a move unavailable to the base grammar. Section~\ref{sec:weight} records a six-slot correspondence with reflective ascent.
\end{remark}

\begin{remark}[Projection and axiom extension]\label{rem:projection-vs-axiom}
  Adding an axiom enlarges derivability and defines a new proof language. The dependency-pair route instead extracts a recursive-call relation, proves chain-freeness of that relation, and invokes an external theorem to recover source termination. Construction adds comparison data to a proof of the original relation; confession changes the intermediate proof goal.
\end{remark}

\begin{proposition}[The forgetting certificate holds for every value of the step argument]\label{prop:confession-not-verdict-y}
  Let $x$ be a step-duplicating instance and let $\pi_y$ be its step-argument dimension. Each of the four confession methods, dependency pairs, counter projection, size-change summaries, and argument filtering, yields a certified forgetting certificate: a rank that decreases along the duplicating step for every value of the step argument and violates wrapper sensitivity for some pair of arguments. The certificate therefore excludes no value of $\pi_y$. The dimension remains in the object system: at live stage $i$ the step argument occupies $i+1$ positions of the trace, and the terminal record of a run of depth $K$ has $K$ frames.
\end{proposition}

\begin{proof}
  The forgetting certificate of each route records the rank of that route together with its orientation and its two violations, and the orientation field quantifies over every step argument, so no value of $\pi_y$ is excluded. On live stages the position count is the carrier-dimension law of the trace, and the frame count of the terminal record is Proposition~\ref{prop:live-vs-record}(3); Theorem~\ref{thm:record-emission-necessity} and Corollary~\ref{cor:y-bridge} show that the same generator continues to appear in the live trace and in the emitted record.
\end{proof}

Under the named external soundness license of each route, the termination verdict transfers from the residual problem, which omits $\pi_y$; for dependency pairs the license is the Arts and Giesl theorem~\cite{artsgiesl2000}.

\begin{remark}[Two license types]\label{rem:two-license-types}
  The certificate-language hierarchy of Definition~\ref{def:witness-hierarchy} places the two license types at different layers. $\mathcal W_1$ construction methods import a well-founded order, such as an interpretation or path order, and use it on the original relation. $\mathcal W_2$ confession methods import a transformation theorem and prove a recursive-call problem. Theorem~\ref{thm:record-emission-necessity} isolates the duplication pattern that blocks the direct grammar. Theorem~\ref{thm:ag-rca} assigns order type $\omega$ to the recursor's one-thread descent.
\end{remark}

\begin{remark}[Linear logic and first-order scope]\label{rem:linear-logic-scope}
  The step rule $F(x,y,S(n)) \to G(y,F(x,y,n))$ has one occurrence of $y$ in the source and two in the target. In Girard's linear logic~\cite{girard1987linearlogic}, a rule that duplicates a hypothesis is contraction, which the exponential modality admits and the linear core excludes; this is cited background, and the statement used here is the occurrence count. Abramsky~\cite{abramsky2005} gives linear-logic background for reversible computation.

  Chardonnet, Saurin, and Valiron~\cite{chardonnetSaurinValiron2023} obtain primitive recursion without exponentials or contraction through circular proofs in $\mu$MALL. Their typed fixed-point regime is outside Definition~\ref{def:record-emitter}, which uses a first-order many-sorted syntax without a fixed-point type former.
\end{remark}

\begin{remark}[The $\mathcal W_2$ family]\label{rem:W2-mechanized}
  Dependency pairs, counter projection, size-change termination, and argument filtering share the projection rank used here, but have different soundness theorems~\cite{rahnamaOrientation}. Polynomial and multiset-path-order certificates inhabit the construction layer. Thus both response classes are nonempty.
\end{remark}
\section{Proof-theoretic comparisons}\label{sec:weight}

\subsection{Metatheoretic strength of the Arts and Giesl license}\label{subsec:ag-strength}

\begin{theorem}[Arts and Giesl equivalence]\label{thm:ag-equivalence}
  Let $R$ be a finite first-order rewrite system whose right-hand sides use only variables of their
  left-hand sides. The following are equivalent: $R$ terminates; the minimal dependency-chain relation
  of $R$ is well founded; there is no infinite minimal dependency chain. A term is strongly normalizing
  if and only if the lengths of its reductions are bounded, and for each bound the absence of a
  reduction of that length is decidable. Without the variable condition both statements fail: the
  rule $c\to x$ has no dependency pairs, does not terminate, and has a reduct for every variable.
\end{theorem}
\begin{proof}
  From termination to well-foundedness, a minimal chain step from a call to a call is a reduction of
  the first followed by passage to a subterm, and this combined relation is well founded on strongly
  normalizing terms. The converse is the argument of Arts and Giesl through minimal non-terminating
  terms~\cite{artsgiesl2000}. The one-step reducts of a term under a finite system with the variable
  condition form a finite computable set, and a relation with finite successor sets is well founded at
  a state if and only if the runs from that state have bounded length.
\end{proof}

The standard ordinal analysis of Peano arithmetic is calibrated at $\varepsilon_0$, whereas the projected measure used by this fixed recursor is natural-number descent. These calibrate different proof objects. The comparison concerns the fixed recursor measure and the stated dependency-pair transport theorem.

\begin{proposition}[Bounded presentation of dependency-pair termination]\label{prop:ag-pi02}
  Let $R$ be a finite TRS whose right-hand sides use only variables of their left-hand sides. Strong normalization of $R$ admits the bounded presentation
  \[
    \mathrm{SN}_b(R)\;:\;\forall t\,\exists n\;D(t,n),
  \]
  where $D(t,n)$ is the decidable predicate ``every reduction sequence issuing from $t$ has length below $n$.'' Then:
  \begin{enumerate}[nosep,leftmargin=1.6em]
    \item $\mathrm{SN}_b(R)$ is a $\Pi^0_2$ sentence, and $\mathrm{SN}(R)\leftrightarrow \mathrm{SN}_b(R)$ is provable in $\mathrm{WKL}_0$;
    \item by Theorem~\ref{thm:ag-equivalence}, the chain-freeness statement for the initial dependency-pair problem of $R$ is equivalent to the bounded presentation $\forall t\,\exists n\,D(t,n)$;
    \item for the recursor's singleton projected pair, the displayed sentence $\forall c\,\exists n\,(c<n)$ is $\Pi^0_2$, its standard-model satisfaction is equivalent to bounded strong normalization of the extracted pair relation, and the finite-chain matrix reads: no chain of length $n$ starts at counter $c$ if and only if $c<n$. Lean proves these equivalences, and Theorem~\ref{thm:ag-rca} derives the predecessor-descent sentence from a finite number-relativized arithmetic theory with one bounded induction sentence. The placement in $\mathrm{RCA}_0$ uses the cited metatheory, and Theorem~\ref{thm:ag-equivalence} proves the Arts and Giesl equivalence for finite systems with the variable condition; its general form is cited~\cite{artsgiesl2000}.
  \end{enumerate}
\end{proposition}

\begin{proof}
  (1) The one-step reducts of a term under a finite system with the variable condition form a finite computable set, so $D(t,n)$ is decidable by bounded exploration and $\mathrm{SN}_b(R)$ has the displayed $\forall\exists$ form; the coding and the $\mathrm{WKL}_0$ equivalence are recorded in Remark~\ref{rem:ag-pi02-coding}. (2) applies Theorem~\ref{thm:ag-equivalence}. For (3), the extracted pair has the single successor $c+1\to c$ and no successor from $0$. Induction gives the chain matrix $D(c,n)\leftrightarrow c<n$, and interpreting the displayed first-order formula in the standard two-sorted model gives $\forall c\,\exists n\,(c<n)$.
\end{proof}

\begin{remark}[Coding of the bounded presentation]\label{rem:ag-pi02-coding}
  After coding terms by numbers, $\mathrm{SN}_b(R)$ is a $\Pi^0_2$ sentence, and the equivalence with strong normalization is provable in $\mathrm{WKL}_0$ over the base theory~\cite{simpson2009sosoa}.
\end{remark}

\begin{remark}[The initial dependency-pair problem versus arbitrary DP problems]
  \label{rem:canonical-vs-arbitrary-dp}
  Endrullis, Geuvers, and Zantema prove that arbitrary dependency-pair problems lacking the minimality flag are $\Pi^1_1$-complete, while arbitrary problems with minimality are $\Pi^0_2$-complete~\cite{endrullisGeuversZantema2009}. They also record the equivalence $\mathrm{SN}(\mathrm{DP}(R)_{\mathrm{top}}/R)\leftrightarrow\mathrm{SN}(R)$ for the initial problem and its minimal version. Proposition~\ref{prop:ag-pi02} concerns the bounded statement of strong normalization transported across this equivalence, and its bounded presentation uses this background and Theorem~\ref{thm:ag-equivalence}.
\end{remark}

\begin{theorem}[Natural-number calibration of the recursor's projected dependency pair]\label{thm:ag-rca}
  The recursor's extracted problem is the singleton pair $F^\sharp(x,y,S(n))\to F^\sharp(x,y,n)$. Projection to the third argument maps every pair step to strict descent on $\mathbb N$. Hence the pair relation is well founded. The intrinsic well-founded rank of every residual term equals its counter, and the supremum of the successor ranks is $\omega$. Thus ordinal height $\omega$ belongs to the residual relation itself, and the chosen measure codomain reflects it. Theorem~\ref{thm:ag-equivalence} is the dependency-pair equivalence between well-foundedness of minimal chains and termination.

  On the arithmetic side, let $L_2$ be the single-sorted language of second-order arithmetic with the predicate $\mathrm{IsSet}$, and read every quantifier below as ranging over numbers. Let $\mathsf B$ be the finite theory of the eight basic axioms of Simpson~\cite{simpson2009sosoa} for $0$, $S$, $+$, $\cdot$ and $<$, and let $\mathsf I$ be the induction sentence of the bounded formula
  \[
    \vartheta(x)\;:=\;x=0\;\vee\;\exists z<x\,\bigl(\neg\mathrm{IsSet}(z)\wedge x=S(z)\wedge z<x\bigr).
  \]
  Write $\mathrm{ZS}$ for the zero-or-successor sentence $\forall x\,\bigl(x=0\vee\exists z\,(x=S(z)\wedge z<x)\bigr)$, $\mathrm{PD}$ for the predecessor-descent sentence $\forall x\,\exists y\,(y<x\vee x=0)$, and $\mathsf I_{\neg<}$ for the induction sentence of $\neg(x<x)$. In the sound first-order calculus of the formalization:
  \begin{enumerate}[nosep,leftmargin=1.6em]
    \item $\mathsf B\vdash\forall x\,(x<S(x))$;
    \item $\mathsf B+\mathsf I\vdash\mathrm{ZS}$ and $\mathsf B+\mathsf I\vdash\mathrm{PD}$;
    \item every $L_2$ theory that contains $\mathsf B$ and $\mathsf I$ derives $\mathrm{PD}$, and so does the extension of $\mathsf B$ by the induction sentences of all bounded formulas in one free variable;
    \item $\mathsf B\nvdash\mathrm{ZS}$;
    \item $\mathsf B+\mathsf I+\mathsf I_{\neg<}\nvdash\forall x\,\neg(x<x)$.
  \end{enumerate}
\end{theorem}

\begin{proof}
  The projection $\pi$ gives $\pi(F^\sharp(x,y,S(n)))=S(n)\rhd n=\pi(F^\sharp(x,y,n))$. An infinite pair chain would yield an infinite descending sequence of natural numbers, contrary to induction. The residual relation has one successor from counter $n+1$ and no successor from counter $0$. Induction therefore computes the intrinsic rank at counter $n$ as the ordinal $n$. Taking the supremum of successor ranks over all counters gives $\omega$.

  For (1), instantiate the axiom $x<S(y)\leftrightarrow(x<y\vee x=y)$ at $y=x$ and use $x=x$. For (2), $\vartheta(0)$ holds by $0=0$, and $\vartheta(S(k))$ follows from the hypothesis that $k$ is a number by taking $z=k$, with (1) for the bound $k<S(k)$. The conclusion of $\mathsf I$ gives $\mathrm{ZS}$, and the two cases of $\mathrm{ZS}$ give the element $y$ of $\mathrm{PD}$: $y=x$ when $x=0$, and $y=z$ when $x=S(z)$. For (3), a theory that contains $\mathsf B$ and $\mathsf I$ contains every premise of the derivation in (2). Items (4) and (5) follow from soundness of the calculus and two finite structures. The first has elements $z_0$, $n$, $s$, interprets $0$ as $z_0$, has $n$ as its only number, and sets $S(n)=s$, $n+s=s$, $n+y=n$ for $y\in\{z_0,n\}$, $n\cdot z_0=z_0$, $n\cdot y=n$ for $y\in\{n,s\}$, and $x<y$ for the single pair $(n,s)$; it satisfies $\mathsf B$, and $n$ differs from $0$ and from $S(n)$, so $\mathrm{ZS}$ fails. The second has elements $n$ and $s$, interprets $0$ as the set object $s$, makes $S$ the identity, $x+y=x$, $x\cdot y=y$, and $<$ the equality relation; it satisfies $\mathsf B$ and $\mathsf I$, it satisfies $\mathsf I_{\neg<}$ because the base case $\neg(0<0)$ is false, and $n<n$ holds. Both the subterm criterion~\cite{hirokawa2004} and the dependency-pair/size-change comparison~\cite{thiemanngiesl2005sct} place the projection argument in the standard framework. General inductive size-change soundness is calibrated at $\mathrm{WO}(\omega_3)$ over $\mathrm{RCA}_0$, where $\omega_3=\omega^{\omega^{\omega}}$~\cite{frittaion2018sct}.
\end{proof}

\begin{remark}[Scope of the arithmetic derivation]\label{rem:ag-rca-coding}
  The theory $\mathsf B$ is a finite number-relativized arithmetic theory: its eight axioms are basic axioms of $\mathrm{RCA}_0$, and it lacks the comprehension scheme and induction axioms. Its extension by all bounded induction sentences also lacks comprehension. The formal derivation of $\mathrm{PD}$ uses the single induction sentence $\mathsf I$. The single-sorted encoding also lacks an axiom stating that $0$ is a number, and both finite structures in the proof of Theorem~\ref{thm:ag-rca} use this: the first interprets $0$ as a set object, and in the second the base case of $\mathsf I_{\neg<}$ concerns a set object. Item (5) therefore refutes the statement that $\mathsf B+\mathsf I+\mathsf I_{\neg<}$ derives irreflexivity of $<$. In $\mathrm{RCA}_0$, where $0$ is a number, the elementary descent proof is formalizable~\cite{simpson2009sosoa}; that placement is cited metatheory.
\end{remark}

\begin{remark}[Instance runtime versus construction-family envelopes]\label{rem:certificate-overkill}
  The duplicator runs in $k+1$ steps. Lexicographic path orders have multiply-recursive derivational-complexity bounds~\cite{weiermann1995lpo}; multiset path orders have primitive-recursive bounds~\cite{hofbauer1992mpo}; and polynomial interpretations have double-exponential upper bounds~\cite{hofbauerlautemann1989rta}. These are family-wide upper bounds. The projected recursor measure has order type $\omega$, whereas general ISCT soundness is calibrated at $\mathrm{WO}(\omega_3)$~\cite{frittaion2018sct}.
\end{remark}

The $\mathrm{RCA}_0$ placement uses the $\Sigma^0_1$ induction of that system~\cite{simpson2009sosoa}, and the formal derivation of Theorem~\ref{thm:ag-rca} uses the single bounded induction sentence $\mathsf I$. The Arts and Giesl theorem is an external premise, and Theorem~\ref{thm:ag-equivalence} proves it for finite systems with the variable condition.

Moser and Schnabl~\cite{moserschnabl2011derivational} prove that, across the applications with traditional orders that they cover, the dependency-pair method keeps derivational complexity inside the class induced by the base order alone. Frittaion, Pelupessy, Steila, and Yokoyama~\cite{frittaion2018sct} calibrate general ISCT soundness at $\mathrm{WO}(\omega_3)$. The subterm-criterion instance used here is a one-thread size-change instance of order type $\omega$.

\begin{remark}[Scope of the reflection comparison]\label{rem:localized-ascent}
  Proposition~\ref{prop:ag-pi02} and Theorem~\ref{thm:ag-rca} concern the termination statement of the initial dependency-pair problem and its projected natural-number measure. Theorem~\ref{thm:structural-identity} compares the dependency-pair response and reflection through six listed stages. Its mathematical content is that correspondence.
\end{remark}

\subsection{Layer crossing under external license}

\begin{definition}[Layer-Crossing-Under-External-License schema]\label{def:lcel-schema}
  A \emph{Layer-Crossing-Under-External-License} (LCEL) instance is a tuple
  \[
    \mathbb{L} = \langle T, \Pi, \Sigma, T^+, \Gamma', \mathrm{Imp}\rangle
  \]
  where:
  \begin{enumerate}[nosep,leftmargin=1.6em]
    \item $T = \langle \Omega_T, \to_T, \vdash_T \rangle$ is a recursively enumerable operational system with a derivation relation $\vdash_T$ and a chosen complexity stratification $\Gamma_0 \subset \Gamma_1 \subset \cdots$.
    \item $\Pi \subseteq \Omega_T$ is the \emph{boundary}: a set of statements or transitions such that $T \not\vdash \varphi$ for every $\varphi \in \Pi$, while $\mathfrak{M} \models \varphi$ in the intended model $\mathfrak{M}$ (internal non-derivability with external truth).
    \item $\Sigma$ is the \emph{external license}: a sentence or rule scheme of the form ``for every $\ulcorner \varphi \urcorner \in \Gamma_n$, $\mathrm{Prov}_T(\ulcorner \varphi \urcorner) \to \varphi$'' or an external soundness theorem about $T$. By construction, $T \not\vdash \Sigma$.
    \item $T^+ = T + \Sigma$ is the \emph{licensed extension}. By construction $T^+ \vdash \varphi$ for each $\varphi \in \Pi$ at the corresponding stratum $\Gamma_n$.
    \item $\Gamma' \subseteq \Gamma_n$ is the \emph{reimport class}. For $\psi\in\Gamma'$, a derivation $T^+\vdash\psi$ yields an annotated derivation $T^{\mathrm{ann}}\vdash\psi$ carrying the license. If $\Pi\cap\Gamma'\ne\varnothing$, clauses (2) and (4) give a sentence provable in $T^+$ and unprovable in $T$, so plain conservativity fails.
    \item $\mathrm{Imp}: \mathrm{Der}(T^+) \to \mathrm{Annot}(T)$ is an \emph{annotation map} such that (a) on $\Gamma'$-conclusions it preserves the claimed conclusion, (b) the annotation records the supplied license data, and (c) conclusion projection commutes with the map up to the specified congruence. The definition requires these three conditions.
  \end{enumerate}
  The carrier $\Omega_T$ may consist of formulas, rewrite terms, or other objects with the stated structure. Kreisel and L\'evy~\cite{kreisellevy1968reflection} and Beklemishev~\cite{beklemishev2005,beklemishev2018reflection} provide background on reflection principles; the dependency-pair instance is a construction within this definition.
\end{definition}

\begin{proposition}[Proof-carrying annotated reimport asymmetry]\label{prop:lcel-reversibility}
  Let $T\hookrightarrow T^+$ be a typed theory extension, let $\Sigma$ be a type of license objects, and let $\psi$ be a base-language sentence. Then:
  \begin{enumerate}[nosep,leftmargin=1.6em]
    \item every base derivation embeds into an annotated derivation and partial erasure recovers it;
    \item every $T^+$-derivation, together with a license, reimports as an annotated derivation that contains both proof objects;
    \item licensed-payload recovery and total-payload reconstruction recover that annotated derivation, while a licensed reimport of a boundary sentence has an empty erasure.
  \end{enumerate}
\end{proposition}

\begin{proof}
  Represent an annotated derivation as a tagged sum of a base derivation and a pair consisting of an extension derivation and its license. Injection of the first summand and projection back to it give (1). Injection and projection of the second summand give (2) and the two recovery laws in (3). A boundary derivation lies in the second summand by $T\nvdash\psi$, so partial erasure to a base derivation is undefined.
\end{proof}

\begin{proposition}[Boundary factorization]\label{prop:lcel-boundary-factorization}
  Let a model have coordinates $\pi_{\mathrm{rev}}$ and $\pi_{\mathrm{irr}}$, a reference state, a truth predicate, and base derivability. Assume that truth of every base-derived statement factors through $\pi_{\mathrm{rev}}$. Then every reference-true statement sensitive to $\pi_{\mathrm{irr}}$ is underivable in the base. If, in addition, every boundary-accepted statement is $\pi_{\mathrm{irr}}$-sensitive, then the boundary-accepted and $\pi_{\mathrm{irr}}$-sensitive sets are equal.
\end{proposition}

\begin{proof}
  Suppose a sensitive statement had a base derivation. Factorization transfers its truth from the reference state to a state with the same reversible coordinate, contradicting sensitivity. The additional hypothesis gives the reverse inclusion.
\end{proof}

\begin{theorem}[Six-slot LCEL comparison]\label{thm:lcel-structural-identity}
  Let $\mathbb L_1$ and $\mathbb L_2$ be typed LCEL instances with a comparison object, boundary object, external-license certificate, licensed extension, reimport class, annotation map, and the compatibility laws of Definition~\ref{def:lcel-schema}. A six-clause comparison transports realization in both directions. Any two realized LCEL instances admit such a six-clause comparison, and both the six-clause comparison type and the comparison-certificate type are subsingletons. The LCEL comparison is therefore a complete classification of the six truth slots and contains no additional certificate choice.
\end{theorem}

\begin{proof}
  The typed instance projects to its six-slot proposition profile. A comparison consists of one biconditional for each slot. Applying the forward implications transports realization from $\mathbb L_1$ to $\mathbb L_2$; applying the reverse implications transports it back. If both instances realize their profiles, each biconditional maps either stored proof to the stored proof on the opposite side, so a comparison exists. Proof irrelevance makes every pair of such six-field comparisons equal, and the same argument applies to the comparison-certificate record.
\end{proof}

\begin{remark}[The recursor factorization model]\label{rem:arch-necessity-as-lcel}
  Take $\pi_{\mathrm{rev}}$ to be counter rank and $\pi_{\mathrm{irr}}$ to be wrapper multiplicity. The extracted pair decreases the first coordinate and preserves the second, while the source step adds one wrapper. Proposition~\ref{prop:lcel-boundary-factorization} places wrapper-sensitive statements in the boundary; the completeness hypothesis gives equality for the stated two-sentence model. For the fixed recursor, the natural-number countdown proves well-foundedness of the reverse pair relation, and root-step orientation yields source well-foundedness.
\end{remark}

\begin{proposition}[Counter rank does not determine wrapper parity]\label{prop:counter-rank-parity}
  Two traces may have the same counter rank and different wrapper parity, so a base language factored through counter rank cannot derive that parity.
\end{proposition}

\begin{proof}
  The dependency-pair rank factors through the counter and does not determine wrapper parity.
\end{proof}

\begin{remark}[Content of the LCEL classification]\label{rem:lcel-scope-and-mechanization}
  Theorem~\ref{thm:lcel-structural-identity} classifies the six proposition slots. Proposition~\ref{prop:lcel-reversibility} supplies the proof-carrying annotated reimport object. Proposition~\ref{prop:lcel-boundary-factorization} identifies the boundary set when both factorization inclusions hold. These are separate typed statements.
\end{remark}

\begin{theorem}[Six-step ascent: shared shape, separate certificates]\label{thm:structural-identity}
  The dependency-pair response to operational inexpressibility on the primitive recursion duplicator and the reflective-ascent component of G\"odel-style incompleteness arguments both realize one six-step shape; the identification at that layer consists of a family tag and six truth values, and the certificate layer separates the two. The six labels are:
  \begin{enumerate}[label=(\arabic*),leftmargin=1.6em]
    \item fix a base system and proof language;
    \item isolate a self-referential or self-duplicating obstruction;
    \item observe, in the metatheory, that the needed result lies beyond the base language;
    \item ascend to a framework that proves it;
    \item prove the needed fact there;
    \item import the result back as a meta-annotation licensing an otherwise unavailable move.
  \end{enumerate}
  Three statements make this a theorem. (i) Both sides realize the shape: the dependency-pair side through the six propositions listed in the proof, and the G\"odel side on compiled arithmetic, whose six stages are a Q-provable reflexive equality, the failure of falsum to satisfy the displayed fixed-point equation, unprovability of $\mathrm{Con}_Q$, a model of Q, failure of $\mathrm{Con}_Q$ in that model, and unprovability of $\neg\mathrm{Con}_Q$. (ii) Membership in the comparison class is six true stages together with the reflection tag, so any two realized profiles of the same family compare in both directions and the comparison hom-sets are subsingletons: a comparison contains no data beyond that classification. (iii) The certificate layer separates the boundaries the shape identifies: the G\"odel checker ascent maps into the quotation ascent, while the orientation and quotation ascents admit a morphism in neither direction.
\end{theorem}

\begin{proof}
  In the formalization of Q, the six inhabited propositions are: Q proves $0=0$; falsum does not satisfy $\operatorname{eval}(\bot)\leftrightarrow\neg\operatorname{Prov}_Q(\bot)$; Q does not prove $\mathrm{Con}_Q$; the one-infinite-element structure is a model of Q; that model refutes $\mathrm{Con}_Q$; and Q does not prove $\neg\mathrm{Con}_Q$. The model-theoretic proof of the third clause transports any proposed Q-proof to the stated model. The checker edge from $0=0$ to $\mathrm{Con}_Q$ is a one-step relation between the two propositions and differs from a Hilbert derivation.

  On the dependency-pair side, the six propositions are the direct proof language, the duplicating obstruction, the external Arts and Giesl framework, the extracted pair, its chain-freeness proof, and the annotated termination verdict. A joint realization theorem places both sides in the six-step shape. The classification theorem proves bidirectional comparison within the realized reflection family and subsingleton comparison types, and a separation theorem gives the certificate-level map from the Q checker ascent to quotation and the absence of maps in both directions between orientation and quotation. Thus the six-slot identification transports truth-profile membership and transports no carrier-level certificate. \Claimcoderef{} lists the Lean declarations.
\end{proof}

\subsection{Reflection and diagonalization}\label{subsec:reflection-register}

Theorem~\ref{thm:structural-identity} compares the reflective-ascent component of G\"odel-style arguments with the dependency-pair method. Here the Lawvere and Yanofsky diagonal schema and formal reflection are separate constructions. Proposition~\ref{prop:not-diagonal} tests the dependency-pair data against the former; Proposition~\ref{prop:reflection-family} records the external-license correspondence with the latter.

\begin{proposition}[Dependency pairs fail the Lawvere and Yanofsky hypotheses]\label{prop:not-diagonal}
  The dependency-pair construction supplies none of the three data of the Lawvere and Yanofsky schema~\cite{lawvere1969,yanofsky2003}: it has no internal code object, no evaluation map, and no fixed-point-free endomap. It is therefore not an instance of that schema.
\end{proposition}

\begin{proof}
  The Lawvere schema requires an internal code object with a surjective representability map, an evaluation map $\mathrm{ev}\colon B^A\times A\to B$, and a fixed-point-free endomap on $B$. The dependency-pair construction supplies none of these data: $F^{\sharp}$ is a marked function symbol instead of a code for $F$, and there is no evaluation map.
\end{proof}

\begin{proposition}[Reflection-family shape of the DP ascent]\label{prop:reflection-family}
  The dependency-pair method has the four-stage shape base obstruction, external license, transformed resolution, and annotated return. The G\"odel reflection ascent on compiled arithmetic for $Q$ has the reflection tag and realizes the six-step shape of Theorem~\ref{thm:structural-identity}. In the duplicator instance, the base is the direct proof language, the external license is the Arts and Giesl soundness theorem, and the selected coordinate is the step-argument dimension. The imported termination verdict has the bounded presentation of Proposition~\ref{prop:ag-pi02}.
\end{proposition}

\begin{proof}
  Theorem~\ref{thm:canonical-instance} proves the step-argument incompatibility within the reflected direct grammar. The Arts and Giesl soundness theorem~\cite{artsgiesl2000} certifies the transformed dependency-pair route and returns the termination verdict with its license annotation. This proves the four-stage pattern; \claimcoderef{} lists the Lean declarations for both sides. The shared six-step shape is a classification and transports no certificate, by clause (ii) of Theorem~\ref{thm:structural-identity}; clause (iii) proves the certificate-level separation.
\end{proof}

\begin{remark}[Diagonal construction and external metatheory]\label{rem:godel-diagonal-reflection-components}
  G\"odel's 1931 theorem~\cite{godel1931} combines diagonalization with metatheoretic assumptions used to infer unprovability and, under suitable soundness assumptions, truth in the intended model. In Lawvere's abstract fixed-point theorem, a weakly point-surjective evaluation map forces every endomap of the value object to have a fixed point; a fixed-point-free endomap such as Boolean negation therefore obstructs such representability~\cite{lawvere1969,yanofsky2003}. Arithmetized provability belongs to the encoding used to construct the Gödel sentence, and the value endomap is Boolean negation. Theorem~\ref{thm:structural-identity} compares only an external-license pattern with the DP route.
\end{remark}

\begin{remark}[Wrapper accumulation and fixed-point-free maps]\label{rem:wrapper-vs-endomap}
  The step rule leaves an accumulated $G$-frame. Under Proposition~\ref{prop:lcel-boundary-factorization}, this is sensitivity to $\pi_{\mathrm{irr}}$. A fixed-point-free endomap $B\to B$ in Lawvere's theorem is a different operation on different data; the dependency-pair construction supplies no such endomap.
\end{remark}

Franz\'en's discussion of transfinite progressions and Smory\'nski's treatment of self-reference provide background for the two constructions~\cite{franzen2004transfinite,smorynski1985selfreference}. The dependency-pair construction lacks the evaluation and value-endomap fields of the Lawvere schema, while its external-license pattern supplies the four-stage shape of Proposition~\ref{prop:reflection-family}.

\begin{remark}[Derivation count]
  \label{rem:no-scaling-bypass}
  Increasing the number or size of derivations leaves the direct grammar unchanged. Every derivable payload-reading expression still fails to orient the duplicating step. The $\mathcal W_2$ route changes the proof goal by projection. Proposition~\ref{prop:construction-confession} separates construction and confession.
\end{remark}

\begin{remark}[Diagonal and reflection components]
  \label{rem:reflection-component-terminology}
  The \emph{diagonal component} is the Lawvere and Yanofsky fixed-point construction: code object, representability map, evaluation map, and fixed-point-free endomap~\cite{lawvere1969,yanofsky2003}. The \emph{reflection comparison} uses only four stages: base obstruction, external soundness theorem, extension, and annotated return.
\end{remark}

\subsection{The orientation boundary as a projection triple}\label{subsec:boundary-transaction}

The dependency-pair method realizes the external-license shape as a theorem: it and the G\"odel reflection ascent realize the same six-step shape, the identification at that layer is a classification that transports no certificate (Theorem~\ref{thm:structural-identity}), and its imported bounded termination verdict has the $\Pi^0_2$ presentation of Proposition~\ref{prop:ag-pi02}. The orientation boundary is then a fixed projection that separates the counter, which decides termination, from the wrapper and payload content.

\begin{definition}[Orientation boundary as a projection triple]\label{def:boundary-transaction}
  For the step-duplicating recursor of Theorem~\ref{thm:canonical-instance}, a
  \emph{projection triple} consists of three data:
  \begin{enumerate}[nosep,leftmargin=1.6em]
    \item the projected payload \emph{dimension} $\pi_y$;
    \item an external \emph{license} $\mathrm{AG}(\pi_y)$;
    \item a \emph{forgetting certificate} certifying that the wrapper and payload coordinates can be ignored in the transformed comparison while the counter descent is retained.
  \end{enumerate}
  A step-indexed family of such triples is \emph{static} when the three
  data agree at all trace stages.
\end{definition}

\begin{proposition}[The boundary is static]\label{prop:boundary-static}
  Let $(\pi_i,\Sigma_i,w_i)_{i\in\mathbb N}$ be a step-indexed family of projection triples. Assume $\Sigma_i=\Sigma_{i+1}$ for every $i$, and assume the license determines the projected dimension and forgetting certificate. Then $\pi_i=\pi_j$, $\Sigma_i=\Sigma_j$, and $w_i=w_j$ for all $i,j$. Thus the family is constant.
\end{proposition}

\begin{proof}
  Induction on the distance between $i$ and $j$ gives $\Sigma_i=\Sigma_j$. The two determination hypotheses then give $\pi_i=\pi_j$ and $w_i=w_j$.
\end{proof}

Proposition~\ref{prop:construction-confession} separates added proof data from licensed projection. Proposition~\ref{prop:norm-mismatch} distinguishes the number of distinct components, the support count, the maximum component size, and the additive mass. Proposition~\ref{prop:carrier-factorization} characterizes multiplicity-insensitivity on the diagonal by factorization through $c_i$. A projection triple records the projected dimension, the license, and the forgetting certificate; the three propositions keep their separate predicates and hypotheses.

\begin{remark}[Projection and multiplicity]\label{rem:licensed-multiplicity-collapse}
  The projection leaves $y$ in the rewrite trace, which contains the generator and its frame copies. For the termination verdict, the license justifies collapsing the multiplicity through the collapse map of Definition~\ref{def:carrier-collapse}, and the counter descent proves termination. Dependency pairs thereby separate the multiplicity that produces the record from the descent that proves termination.
\end{remark}

\begin{remark}[Term-rewriting scope]\label{rem:transaction-scope}
  The orientation boundary here is the dependency-pair projection of the step-argument dimension.
\end{remark}

\subsection{Relation to neighboring concepts}

Operational inexpressibility is local to a proof language, an input, a dimension, and a target value. G\"odel incompleteness concerns derivability of sentences in a theory, and Turing undecidability concerns algorithms for problem classes. On the recursor, termination is proved and a computable normalizer decides joinability (Theorem~\ref{thm:free-context-normalization}), while sound termination certificates appear after one change of representation (Proposition~\ref{prop:boundary-kappa} and Theorem~\ref{thm:free-witness-depths}), so the obstruction lies in the proof language while the truth value stays fixed. Abstraction and lossy compression choose a reduced representation; the projection studied here is justified by a termination-preservation theorem. Parametricity establishes uniform behavior, whereas operational inexpressibility excludes a class of dimension-sensitive derivations.

\section{Formalization and reproduction}\label{sec:formalization}

\paragraph{Representation.} The foundational layer is observer/target data: arbitrary carriers, observer kernels, target kernels, set-valued acceptance relations, and side channels. The source files \path{ObserverTargetCore.lean}, \path{LicenseCore.lean}, \path{TargetKernelQuotient.lean}, and \path{FiberDeficitCore.lean} import no concrete recursor or ambient rule list. Finite algorithms use explicit ordered enumerations in place of noncomputable finite-set conversion. The free recursor layer then supplies one worked term-rewriting realization with closed constructor syntax, the two root rules $R_0$, their full contextual closure $R_{\mathrm{ctx}}$, and the extracted call relation $D_0$. Faithful realizations use injective constructor folds. Located occurrences record both output paths and proofs of the subterms at those paths. Probability models derive joint laws, posteriors, and Bayes decisions from explicit finite rational data.

\paragraph{Formal correspondence.}
\begin{center}
  \footnotesize
  \setlength{\tabcolsep}{2pt}
  \begin{tabularx}{\textwidth}{@{}P{0.16\textwidth} P{0.36\textwidth} X@{}}
    \toprule
    Layer & Mathematical role & Result \\
    \midrule
    observer core & factorization, collisions, executable license, the two readings of the definition, fixed-input boundary, Beth property & Props.~\ref{prop:license-criterion}, \ref{prop:executable-license}, \ref{prop:grammar-license}, \ref{prop:beth-grammar-projection}, \ref{prop:fixed-input-boundary}; Def.~\ref{def:beth-property}; Thms.~\ref{thm:definition-unification}, \ref{thm:beth-gap} \\
    lattice & kernels, partitions, joint observer, join of observers, target-kernel quotient, reconstructions, rule and invariant Galois connection, agreement & Prop.~\ref{prop:target-kernel-quotient}; Thms.~\ref{thm:license-lattice}, \ref{thm:license-galois}, \ref{thm:rule-invariant-galois}, \ref{thm:license-stabilizer}, \ref{thm:aumann-agreement} \\
    recovery and capacity & set-valued certificates, bad-fiber hypergraphs, capacity search, joint targets & Props.~\ref{prop:side-channel-minimum}, \ref{prop:joint-target-capacity}; Thms.~\ref{thm:relational-recovery}, \ref{thm:relational-side-channel}, \ref{thm:set-valued-capacity} \\
    noisy decisions & Bayes risk, risk order, garbling of finite experiments, entropy bound, sufficiency, definability, repair by lift or refusal, functional dependencies & Prop.~\ref{prop:noisy-budgeted-recovery}, \ref{prop:frontier-ktry}, \ref{prop:lift-universal}; Thms.~\ref{thm:blackwell-deterministic}, \ref{thm:blackwell-finite}, \ref{thm:entropy-fiber-bound}, \ref{thm:sufficiency-definability}, \ref{thm:finite-family-sufficiency}, \ref{thm:lift-or-refuse}, \ref{thm:fd-license} \\
    dynamic observation & cutoff, finite stabilization, compactness, retained memory, first sufficient horizon, recurrence depth & Thms.~\ref{thm:license-cutoff}, \ref{thm:finite-license-bound}, \ref{thm:channel-compactness}, \ref{thm:sequential-memory}, \ref{thm:first-sufficient-horizon}, \ref{thm:repeat-depth-license} \\
    free rewrite instance & $R_0$, $R_{\mathrm{ctx}}$, $D_0$, direct grammar, license instances, occurrence roles, role-split law, record emission along the trace, repairs of the role and diagonal targets, one-bit role separator & Thms.~\ref{thm:free-context-normalization}, \ref{thm:canonical-instance}, \ref{thm:free-root-extraction}, \ref{thm:free-role-selection}, \ref{thm:one-bit}; Props.~\ref{prop:recursor-license-instances}, \ref{prop:record-emission-dynamic}, \ref{prop:role-split-law}, \ref{prop:lift-refuse-instances} \\
    signature obstructions & finitary signatures, homomorphisms, term definability, factorization, recursor and guard instances, necessity of the constancy hypothesis & Thms.~\ref{thm:dp-not-signature-derivable}, \ref{thm:eqw-non-derivable}, \ref{thm:homomorphic-inexpressibility}; Cors.~\ref{cor:homomorphic-instances}, \ref{cor:guard-instance}; Prop.~\ref{prop:signature-hypothesis} \\
    quantitative layer & reachable-increment conservation, workload semantics, role information at finite and countable support, scheduled and evidence forms of the gap, crossing, stopping, Snell envelope, zero-deficit stopping, licensed deliberation & Props.~\ref{prop:conserved-lattice}, \ref{prop:workload-semantics}, \ref{prop:countable-role-exchange}, \ref{prop:gap-forms}, \ref{prop:confession-crossing}, \ref{prop:zero-deficit-stopping}, \ref{prop:reflective-pairs}; Thms.~\ref{thm:rary-license-information}, \ref{thm:recursor-stopping}, \ref{thm:snell-envelope} \\
    certified execution & recursive parser, cell-by-cell machine, accepted bitstrings, $r$-ary certificates & Thms.~\ref{thm:strict-certificate}, \ref{thm:bit-machine}, \ref{thm:accepted-set}, \ref{thm:rary-certificate} \\
    transport & unique normal forms, stuttering termination, confluence, quotient simulation, finite adapters, fixtures for sixteen hypotheses & Props.~\ref{prop:quotient-simulation}, \ref{prop:finite-adapters}; Thms.~\ref{thm:normal-form-faithful-transport}, \ref{thm:confluence-transport}, \ref{thm:hypothesis-necessity} \\
    confluence axis & terminal support, repair covers, finite algorithms, binary-fork profiles, $\Pi^1_1$ calibration, order on finite multisets of natural numbers of type $\omega^\omega$, guarded measure below $\omega^\omega\cdot2$ & Thms.~\ref{thm:quant-terminal-characterization}, \ref{thm:quant-repair-cover}, \ref{thm:generic-binary-profile}, \ref{thm:confluence-calibration}; Props.~\ref{prop:quant-premise-necessary}, \ref{prop:confluence-axis-decidable} \\
    certificate layers & construction versus confession information, Arts and Giesl equivalence, projected-pair descent, rank classes, LCEL comparison & Props.~\ref{prop:construction-confession-information}, \ref{prop:construction-confession}; Thms.~\ref{thm:ag-equivalence}, \ref{thm:ag-rca}, \ref{thm:construction-confession-orders}, \ref{thm:lcel-structural-identity}; Cor.~\ref{cor:faithful-rank-classification} \\
    \bottomrule
  \end{tabularx}
\end{center}

\paragraph{Relation matrix.} Every statement about the kernel names one relation.
\begin{center}
  \small
  \begin{tabularx}{\textwidth}{@{}l P{0.13\textwidth} l l P{0.18\textwidth} X@{}}
    \toprule
    Relation & Scope & Guarded & Termination & Confluence & Role here \\
    \midrule
    $R_0$ & root, two rules & no & terminating & confluent & whole-term certificates \\
    $R_{\mathrm{ctx}}$ & all contexts of $R_0$ & no & terminating & confluent & computed normal forms \\
    $D_0$ & free recursor calls & n/a & terminating & confluent & prefix descent \\
    \bottomrule
  \end{tabularx}
\end{center}

\paragraph{Trust boundary.} The development uses Lean~4.22.0-rc4 and the Mathlib revision in \texttt{lake-manifest.json}. Every declaration named in a PROVEN row of \claimcoderef{} is checked by a file under \path{Test/} whose build prints only the axioms $\baseax$ or none. No authored \lean{sorry}, \lean{admit}, added axiom, \lean{native_decide}, or compiled-replacement attribute occurs in the checked development. The check file \path{OIControlsReach.lean} re-elaborates the controls behind Theorem~\ref{thm:hypothesis-necessity} and the scope statements of Sections~\ref{sec:observer} to~\ref{sec:transport}. For the reverse-mathematics theorems of Section~\ref{subsec:ag-strength}, the quantifiers range over the stated finite arithmetic theories and the recursor's projected pair; the observer theorems of Sections~\ref{sec:observer} to~\ref{sec:dynamics} quantify over the sets named in their hypotheses.

\begin{remark}[Import independence of the general theory]\label{rem:import-independence}
  The modules that carry the theorems of Sections~\ref{sec:observer} to~\ref{sec:dynamics} and Theorem~\ref{thm:confluence-transport} import none of the modules that define the concrete kernel, its guarded relation, the dependency-pair development, or the recursor trace. The check file \path{OIGenericCoreImportClosure.lean} imports the nine core modules together with the general modules of these sections and checks every claim-table identifier they declare, and its import closure contains none of the four modules. For Theorem~\ref{thm:hypothesis-necessity}, the fixtures reuse controls of the realization and bring recursor modules, though none of the four, into that closure.
\end{remark}

\paragraph{Reproduction.} Section~\ref{sec:artifact} names the repository, its claim-to-code table, and the check files that elaborate each module.

\section{Related work}\label{sec:related}

\paragraph{Determinacy and rewriting.} Nash, Segoufin, and Vianu say that a set of views determines a query when any two databases with the same view answers have the same query answer, and they separate this information-theoretic condition from the existence of a rewriting of the query in a given query language~\cite{nashSegoufinVianu2010}. The license criterion is determinacy for an observer and a target on arbitrary sets. Theorem~\ref{thm:definition-unification} states when a failure in the language is a failure of determinacy: for sound and complete languages whose observer sees the dimension. Proposition~\ref{prop:beth-grammar-projection} decides the separation for the direct grammar and the projection language: with the full reading or the counter as observer the grammar fails the Beth property on the orienting targets, and the strictly increasing functions of the counter have it.

\paragraph{Observer-relative expressivity and recovery.} Blackwell orders experiments by the decision risks they support~\cite{blackwell1951,blackwell1953}, and Sherman~\cite{sherman1951} proves the comparison theorem behind that order. For deterministic observers on finite sets, refinement of kernels and domination of Bayes risk over all targets and rational priors are the same preorder (Theorem~\ref{thm:blackwell-deterministic}); for finite stochastic experiments both directions are proved (Theorem~\ref{thm:blackwell-finite}). Abstract non-interference parameterizes observer power~\cite{giacobazziMastroeni2010}; strong preservation and abstract completeness characterize retention of a specification language or semantic operation~\cite{ranzatoTapparo2007,giacobazziRanzatoScozzari2000}. Related formulations occur in observational equivalence~\cite{hennessyMilner1985}, bounded-variable graph indistinguishability~\cite{caiFurerImmerman1992}, and security noninterference~\cite{goguenMeseguer1982}. Landauer and Redmond order the equivalence relations on a set of secrets by information and obtain a lattice~\cite{landauerRedmond1993}; Theorem~\ref{thm:license-lattice} gives the license theory of this lattice in the refinement order. In quantitative information flow, the posterior vulnerability of a secret after a deterministic channel depends only on the partition that the channel induces, and gain functions measure other guessing tasks, including guessing within $k$ tries~\cite{alvimEtAl2012}.

\paragraph{Sufficiency, definability, and common knowledge.} Lehmann and Scheff\'e~\cite{lehmannScheffe1950} and Bahadur~\cite{bahadur1954} develop sufficiency and minimal sufficient statistics; Theorem~\ref{thm:sufficiency-definability}(a) is the deterministic instance, in which the target-kernel quotient is the minimal sufficient statistic. Beth~\cite{beth1953} proves that Padoa's method characterizes first-order definability; Theorem~\ref{thm:sufficiency-definability}(b) is the form of that method for sets, with observer-preserving permutations of the states in place of models. Aumann~\cite{aumann1976} expresses common knowledge through the meet of information partitions, which is the join $q_1\vee q_2$ of Theorem~\ref{thm:license-lattice} in the refinement order; Theorem~\ref{thm:aumann-agreement} is the agreement theorem for this meet.

\paragraph{Repairs of functional dependencies.} A relation satisfies a functional dependency $A\to B$ when tuples that agree on $A$ agree on $B$; the rules of Armstrong~\cite{armstrong1974} derive the dependencies that hold in every instance. A subset repair of a relation that violates its dependencies is a maximal subset that satisfies them, and an optimal subset repair deletes a minimum number, or a minimum weight, of tuples~\cite{livshitsKimelfeldRoy2020}. Kivinen and Mannila measure the error of a dependency by the least fraction of tuples whose deletion makes it hold, their measure $g_3$~\cite{kivinenMannila1995}. With each state $x$ read as a tuple with attributes $q(x)$ and $P(x)$, Theorem~\ref{thm:fd-license} identifies the functional dependency $q\to P$ with the license of $P$ by $q$ and the subset repairs with the maximal retained sets. Part~(c) of that theorem gives the consistent answers of Arenas, Bertossi, and Chomicki~\cite{arenasBertossiChomicki1999}, part~(d) counts the repairs for one dependency, and part~(e) identifies the weight of an optimal subset repair with the Bayes risk; Livshits, Kimelfeld, and Wijsen classify the complexity of counting subset repairs for lists~\cite{livshitsKimelfeldWijsen2021}. The repair curve of Proposition~\ref{prop:frontier-ktry} adds a side channel to this setting: a repair may also add an attribute with $k$ values to the left-hand side, and the least deletion weight $\varepsilon_\mu(k)$ is then attained by $k$-lists and is convex in $k$.

\paragraph{Zero-error side information and function computation.} Witsenhausen identified chromatic number as the one-shot zero-error side-information obstruction when the decoder has correlated side information~\cite{witsenhausen1976zeroerror}. Orlitsky and Roche characterize lossless coding for computing a deterministic function with decoder side information through conditional graph entropy~\cite{orlitskyRoche2001}. In that setting $q$ is the decoder's side information and the side channel $s$ is the transmitted message. The function-valued part of Proposition~\ref{prop:side-channel-minimum} restates the confusability-graph principle in fiber form, which this paper uses to obtain an alphabet criterion for arbitrary sets, an executable finite optimum, and joint-target laws. For set-valued acceptance, pairwise compatibility can hold while a whole fiber has no common certificate, so the obstruction is the family of bad fiber subsets of Theorem~\ref{thm:relational-side-channel}. On a finite set of states, when every state has an accepted certificate, Theorem~\ref{thm:set-valued-capacity} gives the least alphabet of this paper's certificate model as the maximum fiber chromatic number of that hypergraph, attained and computed by exhaustive search.

\paragraph{Optimal stopping.} Theorem~\ref{thm:snell-envelope} proves the finite-horizon envelope of Snell~\cite{snell1952,chowRobbinsSiegmund1971}, with the price $d_k$ of the live extension as the recursor's only datum; Theorem~\ref{thm:recursor-stopping} is its one-action instance.

\paragraph{Reverse mathematics of the projected pair.} Theorem~\ref{thm:ag-equivalence} proves the Arts and Giesl equivalence for finite systems whose right-hand sides use only variables of their left-hand sides. Theorem~\ref{thm:ag-rca} derives predecessor descent from a finite number-relativized arithmetic theory with one bounded induction sentence, and the recursor's projected pair has ordinal height $\omega$. The placement in $\mathrm{RCA}_0$ uses the cited metatheory of Simpson~\cite{simpson2009sosoa}. Frittaion, Pelupessy, Steila, and Yokoyama calibrate general size-change soundness at $\mathrm{WO}(\omega_3)$~\cite{frittaion2018sct}; the recursor instance resolves at $\omega$. On the confluence axis, Newman's lemma~\cite{newman1942} consumes a well-founded-induction principle. By Theorem~\ref{thm:confluence-calibration}, that principle is a $\Pi^1_1$ sentence of second-order arithmetic satisfied by the standard model, the bounded termination statement of Proposition~\ref{prop:ag-pi02} is $\Pi^0_2$, the Dershowitz and Manna order on finite multisets of natural numbers has order type $\omega^\omega$, and the measure of the guarded relation stays below $\omega^\omega\cdot2$ on every term, an ordinal below $\varepsilon_0$.

\paragraph{Recursor certificate languages.} The free recursor supplies a separate proof-language instance: the step argument is present, but no payload-dependent expression in the declared direct grammar orients the target relation, while an explicit payload-sensitive orienter outside that grammar succeeds. The restriction to the declared grammar is part of Theorem~\ref{thm:canonical-instance}. Rahnama~\cite{rahnamaOrientation} proves that twelve direct-measure families fail to orient this rule under their hypotheses, and that dependency-pair projection, nonlinear coupling of counter and payload, and path orders prove its termination. Proposition~\ref{prop:boundary-kappa} uses those barriers for the direct whole-term layer $\mathcal W_0$. For the Concrete Witness System, Rahnama~\cite{rahnamaDistinction} proves that disequality is the greatest conservative root repair, that the stated contextual closure of the normalizing guarded relation is confluent, and that no term operation over the seven-constructor signature is a sound and complete disequality test. Theorem~\ref{thm:confluence-calibration} calibrates that guarded relation, and Corollary~\ref{cor:guard-instance} obtains the inexpressibility of the guard from the homomorphic collapse criterion.

Dependency pairs~\cite{artsgiesl2000}, the subterm criterion~\cite{hirokawa2004}, and size-change termination~\cite{thiemanngiesl2005sct} supply the confession methods this paper analyzes. Endrullis, Geuvers, and Zantema~\cite{endrullisGeuversZantema2009} prove $\Pi^0_2$-completeness for dependency-pair problems with the minimality flag and $\Pi^1_1$-completeness without it; Section~\ref{sec:weight} retains that scope. Moser and Schnabl~\cite{moserschnabl2011derivational} keep derivational complexity inside the base order's class for the dependency-pair applications they analyze. Derivational-complexity calibrations for path orders and polynomial interpretations~\cite{weiermann1995lpo,hofbauer1992mpo,hofbauerlautemann1989rta,moser2009habil} provide construction-family envelopes against which the instance runtime $k+1$ is read. Klop's $D(x,x)$ example~\cite{klop1980crs} and Toyama's modularity counterexample~\cite{toyama1987counterexample} illustrate distinct structural roles for left-linearity in confluence arguments and duplication in termination arguments.

\section{Conclusion}\label{sec:conclusion}

An observer licenses a target when the target is constant on its fibers, and this criterion organizes the paper. Observers modulo kernel form the partition lattice, and the target-kernel quotient is the minimal sufficient observer. The licensed Boolean targets and the observer-preserving permutations each determine the kernel, and for deterministic observers on finite sets so does the order of Bayes risks over all targets and rational priors; finite experiments compare by garbling. A sound language licenses every explicitly defined target, and two observers with common knowledge of their signed fiber ratios on a cell of nonzero weight assign the same ratio. On finite sets, the conditional entropy of a target is at most the logarithm of its fiber multiplicity and vanishes when and only when the target is licensed on the support. Operational inexpressibility, defined in the language, coincides with an observer collision for sound and complete languages whose observer sees the dimension.

For set-valued certificate tasks, recovery requires a common certificate on each observer fiber; on a finite set of states, when every state has an accepted certificate, the least side alphabet is the maximum fiber chromatic number of the hypergraph of subsets that share no certificate, and the three-state system has pairwise compatible fibers and capacity two. An unlicensed target is repaired by a side channel that separates colliding states, by refusing states, or by both; on finite sets the least refused weight is the Bayes risk of guessing the target, and with $k$ side symbols it is the weight not kept by guessing the target within $k$ tries. A functional dependency is the license of the right restriction by the left restriction, and the joint observer is the universal lift. When each observer refines the next, a license is lost at most once; compatible finite-image dynamics stabilize; retained histories have one first sufficient horizon whenever some horizon suffices; and the least alphabet that records observation-recurrence depth is the maximum fiber cardinality.

The free recursor realizes these laws with positive and negative controls. Its two-rule system terminates and is confluent under all contexts. The counter observer licenses the denotation of every orienting expression of the closed direct grammar, while an orienter outside the grammar reads the payload, so the obstruction lies in the grammar, and the first sound termination certificate appears after one change of representation. Its strict certificates exist for root steps and only for them at every duplication arity, a cell-by-cell machine reproduces the certificate parser on every input, and the cost theory separates single-trace from repeated-prefix work. The recursor-depth decision has an attained finite-prefix optimum whose Bellman values are the Snell envelope of the stopping rewards, an uninformative channel is refused at every positive cost, and licensed deliberation is well founded.

The counter projection cannot factor through signature homomorphisms constant in the counter slot. The guard obstruction uses a variable-collapsing endomorphism that fixes the truth token. Both instantiate the homomorphic collapse criterion. Along the trace every live stage after the first firing carries the generator in a frame and at the active site. In the calibrated fork model the license resolves one role bit, and every separator of the two roles lies outside the observers that factor through the value. The role exchange extends to countable role alphabets under a summability hypothesis.

The transport theorems move unique normal forms, termination, and confluence along simulations, and sixteen named hypotheses of the paper each have a fixture in which the hypothesis fails together with its conclusion. An application of the transport theorems to a specific confluence problem must prove uniqueness of normal forms for its target. For a finite rewrite system whose right-hand sides use only variables of their left-hand sides, termination is equivalent to well-foundedness of the minimal dependency-chain relation, and the recursor's projected pair is predecessor descent of height $\omega$. On the confluence axis the well-founded-induction sentence that Newman's lemma consumes is $\Pi^1_1$, the Dershowitz and Manna order on finite multisets of natural numbers has order type $\omega^\omega$, and the measure of the guarded relation stays below $\omega^\omega\cdot2$ on every term, an ordinal below $\varepsilon_0$. The Lean development checks the source theorems and paired axiom reports listed in \claimcoderef{}, the import closure of the general theorems excludes every module of the concrete rewrite kernel, and Section~\ref{sec:artifact} names the public artifact and its claim-to-code table.

\section{Public artifact}\label{sec:artifact}

The Lean development is available in the repository cited as~\cite{ko7artifact}, under a source-available license. The development uses Lean~4.22.0-rc4 and the Mathlib revision pinned in \texttt{lake-manifest.json}. At the repository root, the file \path{2.OPERATIONAL_INEXPRESSIBILITY_CLAIM_INDEX.md} is the claim-to-code table of this paper: one row per numbered result, with the Lean declarations that back it, named as \texttt{\#check} prints them, the status of that backing, the module, the claim in one line, what the Lean statement establishes, and the conditions under which it does so; the same file lists every module with its namespaces and main declarations, and every check file with the modules it covers. It is generated from the source of this paper, and every declaration name and module path in it resolves in the repository.

\paragraph{Status labels.} A PROVEN row names declarations that together prove the statement as the paper states it, each axiom-free or using only $\baseax$ as reported by \texttt{\#print axioms}. A DEFINITION row names the definitions that realize a definition. For a PARTIAL row the table states what the declarations prove and what stays prose or cited background, and for a NO-LEAN row it states why the claim has no Lean statement.

\paragraph{Checking.} Each source module has a paired check file under \path{Test/} that elaborates its declarations and prints their axiom dependencies; the check-file table states which modules each one covers, and the repository \path{README.md} gives the command that runs one of them. Four files check the dependency-clean core: \path{ObserverTargetCoreReach.lean}, \path{LicenseCoreReach.lean}, \path{TargetKernelQuotientReach.lean}, and \path{FiberDeficitCoreReach.lean}. The import-closure certificate is \path{OIGenericCoreImportClosure.lean}, and \path{OIControlsReach.lean} checks the predeclared controls. Two files, \path{OIClosureEdgeCases.lean} and \path{OIClosureBoundaryCases.lean}, exercise the aggregate boundary matrix. The identifier check files \path{OperationalInexpressibilityFullAnchorGate1.lean} and \path{OperationalInexpressibilityFullAnchorGate2.lean} check the identifiers of the module map and the claim-to-code table that no paired check file covers.

\end{document}